\documentclass[english]{article}
\usepackage[T1]{fontenc}
\usepackage[latin9]{inputenc}
\usepackage[letterpaper]{geometry}
\usepackage{amsmath}
\usepackage{amssymb}
\usepackage{amsthm}
\usepackage{esint}
\usepackage[english]{babel}
\usepackage[pdftex]{graphicx} 
\usepackage{hyperref} 
\numberwithin{equation}{section}

\newtheorem{Theorem}{Theorem}[section]
\newtheorem{Definition}[Theorem]{Definition}
\newtheorem{Proposition}[Theorem]{Proposition}
\newtheorem{Lemma}[Theorem]{Lemma}
\newtheorem{Claim}[Theorem]{Claim}
\newtheorem{Corollary}[Theorem]{Corollary}
\newtheorem{Remark}[Theorem]{Remark}

\newcommand{\CC}{\mathbb{C}}
\newcommand{\RR}{\mathbb{R}}
\newcommand{\NN}{\mathbb{N}}
\newcommand{\ZZ}{\mathbb{Z}}
\newcommand{\nr}{{\text{near}}}
\newcommand{\fr}{{\text{far}}}
\renewcommand{\O}{\mathcal{O}}
\newcommand{\nn}{\nonumber}
\DeclareMathOperator{\sinc}{sinc}
\DeclareMathOperator{\spec}{spec}

\begin{document}

\title{
Homogenized description of defect modes in periodic structures with localized defects }

\author{Vincent Duch\^ene, Iva Vuki\'cevi\'c and Michael I. Weinstein}

\maketitle

\centerline{\small{ {\it Dedicated to George C. Papanicolaou on the Occasion of His 70th Birthday} }}
\bigskip

\abstract{ 
A spatially localized initial condition for an energy-conserving wave equation with periodic coefficients disperses (spatially spreads) and decays 
in amplitude as time advances.
This dispersion is associated with the continuous spectrum of the underlying differential operator and the absence of discrete eigenvalues. The introduction of spatially localized perturbations in a periodic medium, leads to {\it defect modes}, states in which energy remains trapped and spatially localized. In this paper we study weak, $\mathcal{O}(\lambda),\ 0<\lambda\ll1$, localized perturbations of one-dimensional periodic Schr\"odinger operators. Such perturbations give rise to such defect modes, and are associated with the emergence of discrete eigenvalues from the continuous spectrum. Since these isolated eigenvalues are located near a spectral band edge, there is strong scale-separation between the medium period ($\sim$ order $1$) and the localization length of the defect mode 
($\sim$ order  $|\textrm{defect eigenvalue}|^{-{\frac12}}=\lambda^{-1}\gg1$). Bound states therefore have a multi-scale structure: a ``carrier Bloch wave'' $\times$ a ``wave envelope'', which is governed by a homogenized Schr\"odinger operator with associated {\it effective mass},  depending on the spectral band edge which is the site of the bifurcation. Our analysis is based on a reformulation of the eigenvalue problem in Bloch quasi-momentum space, using the Gelfand-Bloch transform and a Lyapunov-Schmidt reduction to a closed equation for the near-band-edge frequency components of the bound state. A rescaling of the latter equation yields the homogenized effective equation for the wave envelope, and approximations to bifurcating eigenvalues and  eigenfunctions.}

\section{Introduction}\label{introduction}

A spatially localized initial condition for an energy-conserving wave equation with periodic coefficients disperses (spatially spreads) and decays 
in amplitude as time advances.
This (Floquet-Bloch) dispersion is associated with the continuous spectrum (extended states) of the underlying differential operator and the absence of discrete eigenvalues (localized bound states)~\cite{Kuchment-01,RS4}. The introduction of localized perturbations in a periodic medium leads to {\it defect modes}, states in which energy remains trapped and spatially localized. This phenomenon is of great importance in fundamental and applied science\ - from the existence of stable states of matter in atomic systems to the engineering of materials with desirable energy transport properties through localized {\it doping} of ordered materials. 

The process by which the system undergoes a transition from one with only propagating delocalized states to one which supports both localized and propagating states is associated with the emergence or bifurcation of discrete eigenvalues from the continuous spectrum associated with the unperturbed periodic structure. In this paper, we discuss this bifurcation phenomenon in detail for the Schr\"odinger operator 
\begin{equation}
H_Q\ =\ -\partial_x^2 + Q(x)
\label{HQ}
\end{equation}
where $Q(x)$ is a continuous, real-valued, periodic potential:
\begin{equation}
Q(x+1)\ =\ Q(x).
\label{Qperiodic}
\end{equation}

The spectrum, $\spec(H_Q)$, of the Schr\"odinger operator
 is continuous and is the union of closed intervals called {\it spectral bands}~\cite{RS4}. 
 The complement of the spectrum is a union of open intervals called {\it spectral gaps}.
 The spectrum is determined by the family of self-adjoint eigenvalue problems 
 parameterized by the {\it quasi-momentum} $k\in (-1/2,1/2]$:
 \begin{align}
 H_Q u(x;k)\ =\ E\ u(x;k) \ ,\label{HQuEu}\\
 u(x+1;k)\ =\ e^{2\pi i k}\ u(x;k) \ .\label{u-pseudoper}
 \end{align}
That is, we seek {\it $k-$ pseudo-periodic} solutions of the eigenvalue equation. For each $k\in (-1/2,1/2]$, the self-adjoint eigenvalue problem~\eqref{HQuEu}-\eqref{u-pseudoper} has discrete eigenvalue-spectrum (listed with multiplicity):
\begin{equation}
E_0(k) \le E_1(k)\le\dots\le E_b(k)\le\dots
\label{eigs}
\end{equation}
with corresponding $k-$ pseudo-periodic eigenfunctions:
\begin{equation}
u_b(x;k)\ =\ e^{2\pi ikx}\ p_b(x;k),\ \ p_b(x+1;k)\ =\ p_b(x;k),\ \ b\ge 0.
\label{ub-def}
\end{equation}
The $b^{th}$ spectral band is given by
\begin{equation}
\mathcal B_b\ =\ \bigcup_{\substack{k \in (-1/2,1/2]}} E_b(k).
\label{Band-b}
\end{equation}
The spectrum of $H_Q$ is given by:
\begin{equation}
\spec(H_Q) = \bigcup_{b\ge 0} \mathcal B_b\ = \bigcup_{b\geq 0}\ \bigcup_{\substack{k \in (-1/2,1/2]}} E_b(k). \label{spectrum-HQ}
\end{equation}
Since the boundary condition~\eqref{u-pseudoper} is invariant with respect to $k\mapsto k+1$, the functions $E_b(k)$ can be extended to all $\RR$ as periodic functions of $k$. The minima and maxima of $E_b(k)$ occur at $k=k_*\in\{0,1/2\}$; see Figure~\ref{fig:SketchOfSpectrumQ}. In cases where extrema border spectral gap, we have that $\partial_k^2E_b(k_*)$ is either strictly positive or strictly negative~\cite{eastham1973spectral,RS4}; see Lemma~\ref{lem:band-edge}.

Consider now the perturbed operator
$H_{Q+ V}$, where $V(x)$ is sufficiently localized in space.
  By Weyl's theorem on the stability of the essential spectrum, one has $\spec_{\rm cont}(H_{Q+ V})=\spec_{\rm cont}(H_{Q})$; see~\cite{RS4}. 
The effect of a localized perturbation is to possibly introduce discrete eigenvalues into the open spectral gaps. Note that in our setting, $H_{Q+V}$ does not have discrete eigenvalues embedded in its continuous spectrum; see~\cite{Rofe-Beketov:64},~\cite{Gesztesy-Simon:93}.

 In this paper we present a detailed study of the bifurcation of localized bound states into gaps of the continuous spectrum induced by a small and localized perturbation of $H_Q$:
\begin{equation}
 H_{Q+ \lambda V} \ \equiv \ -\partial_x^2 + Q(x)+\lambda V(x)\ ,\ \ \lambda>0,\label{HQV}
 \end{equation}
 where $\lambda$ is taken sufficiently small.
Here $Q(x)$ is a continuous, $1-$ periodic function defined on $\RR$ and $V(x)$ is spatially localized.  We next turn to a summary of our results. See Theorem~\ref{thm:Qzero} and Theorem~\ref{thm:per_result} for detailed statements.

Let $E_*=E_{b_*}(k_*),\ k_*\in\{0,1/2\}$, denote an endpoint (uppermost or lowermost) of the $(b_*)^{th}$ spectral band, bordering a spectral gap. We show that under the condition:
\begin{equation}
\partial_k^2 E_{b_*}(k_*)\ \times\ \int_{\mathbb{R}}\ |u_{b_*}(x;k_*)|^2 \ V(x)\ dx\ <\ 0\ ,
\label{intV}
\end{equation} 
the following holds:\ There exists a positive number, $\lambda_0$, such that for all $0<\lambda<\lambda_0$, $H_{Q+\lambda V}$ has a simple discrete eigenvalue
\begin{equation}
 E(\lambda)\ =\ E_*\ +\ \lambda^2\mu\ +\ \mathcal{O}(\lambda^{2+\alpha}),\ \ \alpha>0\ .
 \label{eig-expand}\end{equation}
which bifurcates from the edge, $E_*=E_{b_*}(k_*) $, of $\mathcal B_{b_*}$ into a spectral gap.

\begin{enumerate}
\item
If $\partial_k^2 E_{b_*}(k_*)>0$ and $ \int_{\mathbb{R}}\ |u_{b_*}(x;k_*)|^2 \ V(x)\ dx< 0$, then $\mu<0$ and $E(\lambda)$ lies near the lowermost edge of $\mathcal B_{b_*}$; see the center panel of Figure~\ref{fig:SketchOfSpectrumQ}.
\item 
If $\partial_k^2 E_{b_*}(k_*)<0$ and $ \int_{\mathbb{R}}\ |u_{b_*}(x;k_*)|^2 \ V(x)\ dx > 0$, then $\mu>0$ and $E(\lambda)$ lies near the uppermost edge of $\mathcal B_{b_*}$; see the right panel of Figure~\ref{fig:SketchOfSpectrumQ}.
\end{enumerate}
 For $0<\lambda<\lambda_0$, $\psi^\lambda(x)$, the eigenstate corresponding to the eigenvalue, $E(\lambda)$, is well-approximated in $L^\infty$ by, $g_0(\lambda x)$, where $g_0(y)$ denotes the unique eigenstate of the {\it homogenized} operator
 \begin{equation}
H_{b_*,{\rm eff}}\ =\ -\frac{d}{d y}\ A_{b_*,{\rm eff}}\ \frac{d}{dy} \ +\ B_{b_*,{\rm eff}}\ \delta(y)\ ,
 \label{ABdelta}
 \end{equation}
 with constant effective parameters $A_{b_*,{\rm eff}}$ and $B_{b_*,{\rm eff}}$. Here,  
 \begin{equation}
 A_{b_*,{\rm eff}}\ =\ \frac1{8\pi^2}\partial_k^2E_{b_*}(k_*)
 \label{Aeff}
 \end{equation}
 is the inverse {\it effective mass} associated to the spectral edge, $E_*=E_{b_*}(k_*)$, 
 \begin{equation}
 B_{b_*,{\rm eff}}\ =\ \int_{\mathbb{R}} |u_{b_*}(x;k_*)|^2 \ V(x)\ dx\ ,
 \label{Beff}
 \end{equation}
 and $\delta(y)$ denotes the Dirac delta mass at $y=0$. The unique discrete eigenvalue, $\mu_\star$, of the eigenvalue problem:  $H_{b_*,{\rm eff}} \psi=\mu\psi$ is easily 
 seen to be 
 \begin{equation}
 \mu_\star=- \frac{B_{b_*,{\rm eff}}^2}{4A_{b_*,{\rm eff}}}.
 \label{eig-Hbeff}\end{equation}

\begin{Remark}\label{EffectiveMass}
The notion of effective mass is well known in condensed matter physics~\cite{Ashcroft-Mermin:76}. The effective mass for an evolving wave-packet may be derived by multiscale perturbation theory and is related the general problem of homogenization of periodic structures; see the very influential book of Bensoussan, Lions \& Papanicolaou~\cite{BLP}; see also~\cite{BS:99,B:03,BS:03,Allaire-Piatnitski:05,BS:06}. 
\end{Remark}

\begin{Remark}\label{Qeq0}
For the case $Q\equiv0$, then $H_Q=H_0=-\partial_x^2$ and its spectrum consists of a semi-infinite interval, $\spec(H_0)=[0,\infty)$, the union of touching bands with no finite length gaps. Furthermore, $p_b(x;k)\equiv1$, for all $|k|\le 1/2$ and $b\ge0$. The only band edge which borders a gap is located at $E_*=E_0(0)=0$, where we have: $k_*=0$, $E_0(k)=4\pi^2 k^2$ and $\partial_k^2E_0(k_*)=8\pi^2$. 
In this case, our results describe the bifurcation of an eigenvalue from the edge of the continuous spectrum of $H_0$ induced by a small and localized perturbation: $H_{\lambda V}=-\partial_x^2+\lambda V$, under the condition $\int_{\RR} V<0$. The homogenized operator is
\begin{equation}
H_{0,{\rm eff}}=-\frac{d^2}{dy^2}+\lambda\int_{\mathbb{R}}V dx\cdot\quad ;
\label{H0eff}\end{equation}
 see the discussion below of~\cite{simon1976bound}.
\end{Remark}
 \begin{figure}[htp]
 \begin{center}
 \includegraphics[width=.7\textwidth]{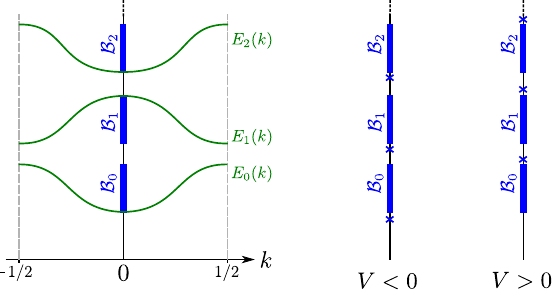}
 \end{center}
 \caption{{\footnotesize{Sketch of spectra. Eigenvalues, $E_b(k),k\in(-1/2,1/2], b=0,1,2,\dots$, are displayed in green. The continuous spectrum, is in blue, and discrete eigenvalues are indicated through cross markers. Left panel corresponds $\spec(H_Q)$, $Q$ periodic. The center (resp. righ) panel corresponds $\spec(H_{Q+\lambda V})$, where $\lambda V$ is small, localized negative (resp. positive).}}}
\label{fig:SketchOfSpectrumQ}
 \end{figure}
%%%%

\subsection{Previous related work} 
The case $Q\equiv0$, where $H_{Q+\lambda V}=-\Delta+V(x)$ was considered by Simon~\cite{simon1976bound} in one and two spatial dimensions. In one dimension, it is proved that if $V$ is sufficiently localized and $-\infty<\int_{\mathbb R} V<0$, then $H_{\lambda V}$ has a small negative eigenvalue $E(\lambda)$ of order $\lambda^2$; see the Corollary~\ref{cor:Elambdaprecise} and the discussion following it. The case of perturbations of one-dimensional periodic Schr\"odinger operators ($Q$ non-trivial, $1-$periodic) is treated by Gesztesy \& Simon~\cite{Gesztesy-Simon:93}, where sufficient conditions are given for the bifurcation of eigenvalues in the gaps of the continuous spectrum. Borisov and Gadyl'shin~\cite{Borisov-Gadylshin:08}  obtain results closely related to the current work, although using very different methods.
A formal asymptotic study, in terms of a Floquet-Bloch decomposition, in one and two spatial
 dimensions was given in Wang {\it et. al.}~\cite{Wang-Yang-Chen:07}. 
Parzygnat {\it et. al.}~\cite{parzygnat2010sufficient} formulate a variational principle for defect modes with frequencies in spectral gaps. They use formal trial function arguments to show the existence of such defect modes in spatial dimensions one and two. By formal asymptotic arguments, they obtain the condition~\eqref{intV}, for the case of the first spectral gap.
Deift \& Hempel~\cite{Deift-Hempel:86} obtained results on the existence and number of eigenstates in spectral gaps 
  for operators of the general type $H-\lambda W$, where $H$ has a band spectrum and $W$ is bounded. 
Figotin \& Klein~\cite{Figotin-Klein:97,Figotin-Klein:98} studied localized defect modes in context of acoustic and electromagnetic waves.
Results on bound states and scattering resonances of one-dimensional Schr\"odinger operators with compactly supported potentials appear in work of Bronski \& Rapti~\cite{BR:11} and Korotyaev~\cite{Korotyaev:09,Korotyaev:11}, respectively.
Bifurcations of defect modes into spectral gaps was considered in dimensions $d=1,2$ \underline{and} $3$ by Hoefer \& Weinstein~\cite{Hoefer-Weinstein:11} for operators of the form $-\Delta + Q(x) + \varepsilon^2V(\varepsilon x)$, where $Q$ is periodic on $\mathbb R^d$ and $V$ is spatially localized. This scaling was motivated by work of Ilan \& Weinstein~\cite{Ilan-Weinstein:10} on the bifurcation of nonlinear bound states from continuous spectra for the nonlinear Schr\"odinger / Gross-Pitaevskii equation. The works~\cite{Ilan-Weinstein:10,Hoefer-Weinstein:11} employ the general Lyapunov-Schmidt reduction strategy used in the present work; see also 
\cite{Pelinovsky-Schneider:07,Dohnal-Uecker:09,Dohnal-Pelinovsky-Schneider:09}.

\subsection{Outline, remarks on the proof and future directions}\label{sec:outline} In Section~\ref{sec:background} we present background material concerning spectral properties of Schr\"odinger operators with periodic potentials defined on $\mathbb R$. In Section~\ref{sec:results} we give precise technical statements of our main results: Theorem~\ref{thm:Qzero} and Theorem~\ref{thm:per_result}. 

Our strategy of  proof is to transform the eigenvalue problem, using the appropriate spectral transform (Fourier or Floquet-Bloch), to a formulation in frequency (quasi-momentum) space. Anticipating a bifurcation from the spectral edge, we express the eigenvalue problem in terms of coupled equations governing the frequency components located {\it near} the band edge and those which are {\it far} from the band edge. The precise frequency cutoff depends on the small parameter, $\lambda$. We employ a  Lyapunov-Schmidt reduction strategy~\cite{Nirenberg:74} in which we solve for the {\it far}-frequency components as a functional of the {\it near}-frequency components. This yields a reduction to a  closed {\it bifurcation equation} for the {\it near}-frequency components. In contrast to classical applications of this strategy, our reduced equation is infinite dimensional. For $\lambda$ small, in an appropriate scaled limit, the bifurcation equation is asymptotically exactly solvable; it is the eigenvalue problem for the homogenized / effective operator $H_{b_*,{\rm eff}}$.  In Section~\ref{sec:gentech}, we prove a general technical lemma, crucial to the analyses of Sections~\ref{sec:nonper} and~\ref{sec:per}, covering the kinds of bifurcation equations which arise. 
 Finally, Appendices~\ref{pf-of-lem:band-edge} and~\ref{regularity-app} contain the proof of results stated in Lemmata~\ref{lem:band-edge} and~\ref{lem:regularity-of-Eb}, and in Appendix~\ref{proof-of-Corollary-per}
we give proofs, by a bootstrap method, of Corollary~\ref{cor:Elambdaprecise} and 
 Corollary~\ref{cor:Elambdaprecise-per} which contain more detailed expansions and sharper error terms for the bifurcating eigenstates than those in Theorem~\ref{thm:Qzero} and 
 Theorem~\ref{thm:per_result} . 
 
We conclude this section with several possible extensions of the present work. 
\begin{enumerate}
\item The results of this paper describe the bifurcation of eigensolutions  in the case where the perturbing potential is small  in the strong sense (in norm). {\it What of the case where the perturbing potential converges weakly to zero?} This corresponds to the question of the effective behavior of high-contrast microstructures. In~\cite{DVW:12}, the authors consider a class of problems, depending on a small parameter, $\varepsilon$,  including the case  where the potential, $q^{(\varepsilon)}(x)=q(x,x/\varepsilon)$, converges {\it weakly} as $\varepsilon$ tends to zero.  In particular, we considered the small $\varepsilon$ limit of 
the  scattering and time-evolution properties for operators of the form
$H^{(\varepsilon)}=-\partial_x^2+q(x,x/\varepsilon),$
where $y\mapsto q(\cdot,y)$ is oscillatory (including periodic and certain almost periodic cases) and $x\mapsto q(x,\cdot)$ is spatially localized. An important subtlety arises in the case where $q_{av}(x)=\int_{\mathbb R} q(x,y) dy \equiv 0$, {\it i.e.} $q^\varepsilon$ tends to zero \underline{weakly}; see \cite{DW:11} for the case where $q_{\rm av}(x)\ne0$ is {\it generic}.
 In this case, classical homogenization theory breaks down at low energies. Indeed, the homogenized operator, obtained by averaging the potential over its fast variations, is $H_0=-\partial_x^2$, which does not capture key spectral and scattering information. Among these are the low energy behavior of the transmission coefficient (related to the spectral measure) and the existence of a bifurcating bound state at a very small negative energy. We show that 
 the correct behavior is captured by an effective Hamiltonian with effective potential well: $H_{\rm eff}^{(\varepsilon)}=-\partial_y^2-\varepsilon^2 \Lambda_{\rm eff}(y),\ \ \Lambda_{\rm eff}(y)>0$. Using Theorem~\ref{thm:Qzero} and the results of~\cite{simon1976bound}, we conclude that $H^{(\varepsilon)}$ has a bound state with negative energy of the order $\varepsilon^4$, with a precise expansion for $\varepsilon$ small. 

Thus, it would be interesting to use our approach in order to extend the results of the present paper to families of potentials, $q^{\varepsilon}$, which converge weakly to a nontrivial periodic potential, $Q(x)$; see~\cite[in progress]{DVW3}. 
\item Further, in~\cite{DVW:12} a {\it multi-scale} local energy  time-decay estimate, for localized initial conditions orthogonal to the bound state,  in which the different dispersive time-dynamics on different time-scales is explicit. In particular, the decay rate is $\mathcal{O}(t^{-\frac12})$ for times $t\ll\varepsilon^{-2}$ and  $\mathcal{O}(t^{-\frac32})$ for 
$t\ge\varepsilon^{-2}$. We believe that our methods can be extended to give  detailed properties of the resolvent $\left(-\partial_x^2+Q+\lambda V-E\right)^{-1}$ and therefore the spectral measure~\cite{RS4} near the band edges. Such information could be used to derive the detailed dispersive time-decay behavior. However, the decay estimates of the type obtained in ~\cite{DVW:12} can be expected to hold only for initial conditions which are spectrally localized near band edges. Initial conditions with spectral components away from the band edge can sample a regime where, for $Q$ non-zero, the dispersion relation has higher degeneracy, yielding different (slower) dispersive time-decay
~\cite{Firsova:96,Cai:06,cuccagna:08}.

\item Finally, it would be of interest to extend the methods of the current paper to the study of bifurcations of eigenvalues for multiplicatively small or weakly convergent spatially localized perturbations of the higher-dimensional periodic Schr\"odinger operator, $-\Delta +Q$. In spatial dimension $n=2$ and the case $Q=0$, Simon~\cite{simon1976bound} proved that the bound state generated by a multiplicatively small perturbation is exponentially close to the edge of the continuous spectrum. Such results has been extended by Borisov~\cite{Borisov11} in the periodic ($Q$ nontrivial) case. Formal asymptotics were obtained in Wang {\it et. al.}~\cite{Wang-Yang-Chen:07}. In spatial dimensions $n\ge3$, it is well known that for sufficiently small 
$\lambda$, $-\Delta +\lambda V$ does not have discrete spectrum, by Cwikel-Lieb-Rozenblum bound. 
Finally, Parzygnat {\it et. al.}~\cite{parzygnat2010sufficient} also treat the case of dimensions $n\ge3$, where the defect potential, $V$ is localized along in one or two dimensions.
\end{enumerate}

\subsection{Definitions and notation} 

We denote by $C$ a constant, which does not depend on the small parameter, $\lambda$. It may depend on norms of $Q(x)$ and $V(x)$, which are assumed finite. $C(\zeta_1, \zeta_2,\dots)$ is a constant depending on 
 $\zeta_1$, $\zeta_2$, $\dots$.
We write $A\lesssim B$ if $A\leq C \ B$, and $A\approx B$ if $A\lesssim B$ and $B\lesssim A$.

 \noindent $\chi$ and $\overline{\chi}$ are the characteristic functions defined by
 \begin{equation}
 \chi_{_\delta}(\xi) = \chi\left(|\xi| < \delta\right)  \equiv 
 \left\{ 
 \begin{array}{ll}
 1, \quad & |\xi| < \delta \\
 0, \quad & |\xi| \geq \delta
 \end{array}
 \right. , \qquad
 \overline{\chi}_{_\delta}(\xi)=\overline{\chi}(|\xi| < \delta)  \equiv  1-\chi(|\xi| < \delta)\ .
 \label{chi-def}\end{equation}
 
\noindent For $f,g \in L^2(\RR)$, the Fourier transform and its inverse are given by 
\[
\mathcal{F}\{f\}(\xi)  \equiv  \widehat{f}(\xi) = \int_\RR e^{-2\pi i x\xi}f(x)dx, \qquad
\mathcal{F}^{-1}\{g\}(x)  \equiv  \check{g}(x) = \int_\RR e^{2\pi i x\xi}g(\xi)d\xi.
\]

\noindent $\mathcal{T}$ and $\mathcal{T}^{-1}$ denote the Gelfand-Bloch transform and its inverse;
 see Section~\ref{sec:background}.

\noindent $L^{p,s}(\RR)$ is the space of functions $F:\RR\to\RR$ such that $(1+|\cdot|^2)^{s/2} F\in L^p(\RR)$, endowed with the norm
\begin{equation}
\|F\|_{L^{p,s}(\mathbb{R})} \equiv \|(1+|\cdot|^2)^{s/2} F\|_{L^p(\mathbb{R})} < \infty,\ \ 1\le p\le\infty.
\label{Lps-def}
\end{equation}

\noindent $W^{k,p}(\RR)$ is the space of functions $F:\RR\to\RR$ such that $\partial_x^j F\in L^p(\RR)$ for $0\leq j\leq k$, endowed with the norm
\[
\|F\|_{W^{k,p}(\mathbb{R})} \equiv \sum_{j=0}^k \|\partial_x^j F\|_{L^p(\mathbb{R})} < \infty,\ \ 1\le p\le\infty.
\]

\noindent{\bf Acknowledgements:}\ I.V. and M.I.W. acknowledge the partial support of U.S. National Science Foundation under U.S. NSF Grant DMS-10-08855, the Columbia Optics and Quantum Electronics IGERT NSF Grant DGE-1069420 and  NSF EMSW21- RTG: Numerical Mathematics for Scientific Computing. This research was carried out while V.D. was the Chu Assistant Professor of Applied Mathematics at Columbia University.

\section{Mathematical preliminaries} \label{sec:background}
In this section we summarize basic results on the spectral theory of Schr\"odinger operators with periodic potentials defined on $\mathbb R$. For a detailed discussion, see for example,~\cite{eastham1973spectral,RS4,magnus1979hill}.

\subsection{Floquet-Bloch states}
We seek solutions of the $k-$ pseudo-periodic eigenvalue problem
\begin{equation}
(-\partial_x^2 + Q(x))u(x) = Eu(x) \ , \quad u(x+1)=e^{2\pi i k}u(x)\ ,
\label{k-evp}\end{equation}
for $k\in(-1/2,1/2]$, the {\it Brillouin zone}.
Setting $u(x;k) = e^{2\pi ikx}p(x;k)$, we equivalently seek eigensolutions $(p(x;k),E)$ of the periodic elliptic boundary value problem:
\begin{equation}\label{eq:e-value}
\left( -(\partial_x + 2\pi ik)^2 + Q(x) \right)p(x;k) = E(k)p(x;k), \quad p(x+1;k)=p(x;k)
\end{equation}
for each $k \in (-1/2,1/2]$. 

The eigenvalue problem~\eqref{eq:e-value} has a sequence of eigenpairs $\{(p_b(x;k),E_b(k))\}_{b \geq 0}$ satisfying~\eqref{eigs} and~\eqref{ub-def}. The functions $p_b(x;k)$ may be chosen so that $\{p_b(x;k)\}_{b\ge0}$ is, for each fixed $k\in(-1/2,1/2]$, a complete orthonormal set in $L_{\rm per}^2([0,1])$. It can be shown that the set of Floquet-Bloch states $\{u_b(x;k)\equiv e^{2\pi ikx}p_b(x;k), \ b\in\NN, \ -1/2< k\leq 1/2\}$ is complete in $L^2(\RR)$, {\it i.e.} for any $f\in L^2(\RR)$, 
\[
f(x) \ - \ \sum_{0\le b\le N} \int_{-1/2}^{1/2} \langle u_b(\cdot,k),f \rangle_{L^2(\mathbb{R})} u_b(x;k)\ dk
\to\ 0\]
in $L^2(\mathbb R)$ as $N\uparrow\infty$.

Recall further that the spectrum of $-\partial_x^2 + Q(x)$ is continuous, and equal to the union of the closed intervals, the spectral bands; see~\eqref{Band-b},~\eqref{spectrum-HQ}.

\begin{Definition} 
 We say there is a spectral gap between the $b^{th}$ and $(b+1)^{st}$ bands if
\begin{equation}\sup_{|k|<1/2} |E_b(k)|\ <\ \inf_{|k|<1/2} |E_{b+1}(k)|.\label{whatisgap}
\end{equation}
\end{Definition}

Our study of eigenvalue bifurcation from the band edge $E_* \equiv E_{b_*}(k_*)$ into a spectral gap, requires regularity and detailed properties of $E_b(k)$ near its edges. These are summarized in the following results (see a sketch of $E_b(k)$ in Figure~\ref{fig:SketchOfSpectrumQ}, left panel). 
Proofs and references to proofs are given in Appendices~\ref{pf-of-lem:band-edge} and~\ref{regularity-app}. 

\begin{Lemma}\label{lem:band-edge} 
Assume $E_b(k_*)$ is an endpoint of a spectral band of $-\partial_x^2 + Q(x)$, which borders on a spectral gap;
 see \eqref{whatisgap}.
Then $k_*\in\{0,1/2\}$ and the following results hold:
\begin{enumerate}
\item $E_b(k_*)$ is a simple eigenvalue of the eigenvalue problem~\eqref{k-evp}.
\item \subitem $b$ even: $E_b(0)$ corresponds to the {\em left (lowermost)} end point of the band,
\subitem \phantom{$b$ even:} $E_b(1/2)$ corresponds to the {\em right (uppermost)} end point.
\subitem $b$ odd: $E_b(0)$ corresponds to the {\em right (uppermost)} end point of the band,
\subitem \phantom{$b$ odd:} $E_b(1/2)$ corresponds to the {\em left (lowermost)}  end point.
\item $\partial_k E_b(k_*) = 0$,\ $\partial_k^3E_b(k_*) = 0$;
\item \subitem $b$ even: $\partial_k^2 E_b(0) > 0$, $\partial_k^2 E_b(1/2) < 0$;
\subitem $b$ odd: $\partial_k^2 E_b(0) < 0$, $\partial_k^2 E_b(1/2) > 0$; 
\end{enumerate}
\end{Lemma}

\begin{Lemma}\label{lem:regularity-of-Eb}
 Let $E_b(k_1)$ denote a simple eigenvalue; thus $k_1=k_*$ as above applies. Then, the mappings $k\mapsto E_b(k),\ k\mapsto u_b(x;k)$, with $u_b$ normalized, can be chosen to be analytic for $k$ in a sufficiently small complex neighborhood of $k_1$.  Moreover, for $k$ real and in this neighborhood $(E_b(k),u_b(x;k))$ are Floquet-Bloch eigenpairs. 
\end{Lemma}

\subsection{The Gelfand-Bloch transform} 
Given $f \in L^2(\mathbb{R})$, we introduce the Gelfand-Bloch transform $\mathcal{T}$ and its inverse as follows
\begin{align*}
\mathcal{T}\{f(\cdot)\} & =  \widetilde{f}(x;k) = \sum_{n\in\mathbb{Z}}e^{2\pi inx}\widehat{f}(k+n),\qquad
\mathcal{T}^{-1}\{\widetilde{f}(x;\cdot)\}(x)  =  \int_{-1/2}^{1/2} e^{2\pi ixk}\widetilde{f}(x;k)dk.
\end{align*}
One can check that $\mathcal{T}^{-1}\mathcal{T} = {\rm Id}$.
 Let 
\begin{equation}
u(x;k)\ =\ e^{2\pi ikx} p(x;k)
\label{p-def}
\end{equation}
denote a Floquet-Bloch mode. Then, 
by the Poisson summation formula, we have that
\[ \langle u(\cdot,k),f \rangle_{L^2(\mathbb{R})} \ = \ \langle p(\cdot,k),\widetilde f(\cdot;k) \rangle_{L^2_{\rm per}([0,1])}
\ .\]

Define 
\begin{equation}\label{def:Tb}
\mathcal{T}_b \{f\}(k) \equiv \langle p_b(\cdot;k),\widetilde{f}(\cdot;k) \rangle_{L^2_{\rm per}([0,1])} \equiv \int_0^1 \overline{p_b(x;k)}\widetilde{f}(x;k)dx.
\end{equation}
By completeness of the $\{p_b(x;k)\}_{b \geq 0}$, %and $\{u_b(x;k)\equiv e^{2\pi ikx}p_b(x;k),b\in\NN,-1/2< k\leq 1/2\}, $
the spectral decomposition of $f\in L^2(\RR)$ in terms of Floquet-Bloch states is
\[
\widetilde{f}(x;k) \ = \ \sum_{b \geq 0} \mathcal{T}_b \{f\}(k) p_b(x;k) \ , \quad f(x) \ = \ \sum_{b \geq 0} \int_{-1/2}^{1/2} \mathcal{T}_b \{f\}(k) u_b(x;k) dk.
\]

 Recall the Sobolev space, $H^s$, the space of functions with square integrable derivatives  up to order $s$. It is natural to construct the following $\mathcal{X}^s$ norm in terms of Floquet-Bloch states:
\begin{equation}
\| \widetilde{\phi} \|_{\mathcal{X}^s}^2 \equiv \int_{-1/2}^{1/2} \sum_{b \geq 0} \left(1+|b|^2\right)^s |\mathcal{T}_b\{\phi\}(k)|^2dk.
\label{Xs-norm}
\end{equation}

\begin{Proposition}\label{prop:norm_equivalence}
$H^s(\mathbb{R})$ is isomorphic to $\mathcal{X}^s$ for $s \geq 0$. Moreover, there exist positive constants $C_1$, $C_2$ such that for all $\phi \in H^s(\mathbb{R})$, we have
$
C_1 \| \phi \|_{H^s(\mathbb{R})} \leq \| \widetilde{\phi} \|_{\mathcal{X}^s} \leq C_2 \| \phi \|_{H^s(\mathbb{R})}.
$
\end{Proposition}
\begin{proof}
Since $E_0(0) = \inf \spec(-\partial_x^2 + Q)$, then $L_0 \equiv -\partial_x^2 + Q - E_0(0)$ is a non-negative operator and $H^s(\mathbb{R})$ has the equivalent norm defined by 
$
\| \phi \|_{H^s} \approx \|(I + L_0)^{s/2} \phi \|_{L^2}.
$
Using orthogonality, it follows that
\begin{align*}
\| \phi \|_{H^s}^2 & \approx \|(I + L_0)^{s/2} \phi \|_{L^2}^2\ 
 =\  \sum_{b \geq 0} \int_{-1/2}^{1/2} |\mathcal{T}_b\{\phi\}(k)|^2 |1 + E_b(k) - E_0(0)|^s dk \\
& \approx  \sum_{b \geq 0} \left( 1 + |b|^2 \right)^s \int_{-1/2}^{1/2} |\mathcal{T}_b\{\phi\}(k)|^2 dk\ \equiv \| \widetilde{\phi} \|_{\mathcal{X}^s}^2.
\end{align*}
The last line follows from the Weyl asymptotics $E_b(k) \sim b^2$; see, for example,~\cite{Courant-Hilbert-I}. 
This completes the proof of Proposition~\ref{prop:norm_equivalence}
.\end{proof}

We conclude this section with a Lemma, which gives various estimates on the Floquet-Bloch states of $H_Q$  and the spectrum of $H_{Q+\lambda V}$, for a class of periodic potentials, $Q$, and localized potentials, $V$. These estimates are used within the proof of Theorem~\ref{thm:per_result}, in Section~\ref{sec:per}.
\begin{Lemma}\label{lem:estimates}
Assume that $Q$
 is continuous, $1-$periodic, and $V$ is such that $(1+|\cdot|)V(\cdot)\in L^1$. Let $\Omega$ be a small neighborhood of $k_1$ a simple eigenvalue, such that Lemma~\ref{lem:regularity-of-Eb} applies.
Then one has:
\begin{subequations}
\begin{align}
(a)\qquad &\ \sup_{k\in(-\frac12,\frac12]}\big\Vert p_{b}(\cdot;k)\big\Vert_{L^\infty}\ \leq \ \sup_{k\in(-\frac12,\frac12]} \sum_{n\in\mathbb Z}\left| \left\langle p_{b}(\cdot;k),e^{2\pi i n\cdot}\right\rangle_{L^2([0,1])} \right| \ <\ \infty, \label{pb-k-reg}\\ 
&\nn\\
(b)\qquad &\sup_{k\in\Omega}\big\Vert \partial_k p_{b}(\cdot;k)\big\Vert_{L^\infty}\ \leq\ \sup_{k\in\Omega}\sum_{n\in\mathbb Z}\left| \left\langle \partial_k p_{b}(\cdot;k),e^{2\pi i n\cdot}\right\rangle_{L^2([0,1])} \right| \ <\ \infty.\label{pb-dk-reg}
\end{align}
\end{subequations}
\end{Lemma}
\begin{proof}
We begin by proving that $p_b(x;k)$ is uniformly bounded for $x\in\RR$ and $k\in(-1/2,1/2]$. Since $p_b(\cdot;k)$ is $1-$periodic, it is bounded if its Fourier coefficients are summable. Thus we study
\[ \sum_{n\in\ZZ}\big|\langle p_b(\cdot;k),e^{2\pi i n\cdot}\rangle_{L^2([0,1])}\big| \ = \ \sum_{n\in\ZZ}\left| \int_0^1 p_b(x;k)e^{-2\pi i nx}\ dx\right|.\]
Since $k\in(-1/2,1/2]$, we can use integration by parts for $n\neq0$, the Cauchy-Schwarz inequality and equation~\eqref{eq:e-value} for $p_b(x;k)$ to obtain
\begin{align*}
\sum_{n\in\ZZ}\big|\langle p_b(\cdot;k),e^{2\pi i n\cdot}\rangle_{L^2([0,1])}\big|
&\leq \big\Vert p_b(x;k)\big\Vert_{L^2([0,1])}\big\Vert 1\big\Vert_{L^2([0,1])}\\
&\quad +\sum_{n\in\ZZ\setminus\{0\}}\left| \int_0^1 (Q(x)-E_b(k))p_b(x;k)\left(\frac1{2\pi i(n-k)}\right)^2e^{-2\pi i nx}\ dx\right|\\
&\leq 1 +\sum_{n\in\ZZ\setminus\{0\}}\frac1{4\pi^2 (n-k)^2}\big\Vert (Q(\cdot)-E_b(k))p_b(\cdot;k) \big\Vert_{L^2([0,1])}.
\end{align*}
Thus,  $ \sup_{k\in(-1/2,1/2]}\big\Vert p_b(\cdot;k)\big\Vert_{L^\infty} \leq\sup_{k\in(-1/2,1/2]} \sum_{n\in\ZZ}\big|\langle p_b(\cdot;k),e^{2\pi i n\cdot}\rangle_{L^2([0,1])}\big| <\infty$.
\medskip

We now turn to the study of $\partial_k p_b(x;k)$ in (b). Differentiating~\eqref{eq:e-value} with respect to  $k$ yields
\[
\left( -(\partial_x + 2\pi ik)^2 + Q(x) \right)\partial_k p_b(x;k) = E_b(k)\partial_k p_b(x;k)+\big(\partial_k E_b(k) + 4\pi i( \partial_x+2\pi i k) \big)p_b(x;k).
\]
Following the same method as above yields
\begin{align*}
\big\Vert \partial_k p_{b}(\cdot;k)\big\Vert_{L^\infty} &\leq \sum_{n\in\ZZ}\big|\langle \partial_k p_b(\cdot;k),e^{2\pi i n\cdot}\rangle_{L^2([0,1])}\big|  \\
&\leq C(\big\Vert Q\big\Vert_{L^\infty},E_b(k))\big\Vert \partial_k p_b(x;k)\big\Vert_{L^2([0,1])} +C(\partial_k E_b(k),\big\Vert Q\big\Vert_{L^\infty},E_b(k) ).
\end{align*}
The finiteness of $\big\Vert \partial_k p_b(\cdot;k)\big\Vert_{L^2([0,1])}$ and $\partial_k E_b(k)$ for $k\in\Omega$  is a consequence of Lemma~\ref{lem:regularity-of-Eb}; thus (b) follows.
\end{proof}

\section{Bifurcation of defect states into gaps;\ main results} \label{sec:results}

Consider the eigenvalue problem:
\[
\left(-\partial_x^2 + Q(x)+\lambda V(x) \right)\psi^\lambda\ = \ E^\lambda \psi^\lambda\ ,\ \psi\in L^2(\mathbb R) ,
\]
where $Q(x)$ is continuous, $1-$periodic, $\lambda>0$ is small, and $V(x)$ is spatially localized.  
Our first result concerns the case where $Q\equiv0$:
\begin{Theorem}[$Q\equiv 0$]\label{thm:Qzero}
Let $V$ be such that $\widehat V\in W^{1,\infty}(\mathbb{R})$; thus $\int_\RR (1+|x|)|V(x)|\ dx<\infty$ suffices. Assume 
$ \widehat V(0)\ =\ \int_\mathbb{R} V\ <\ 0$. 
There exists positive constants $\lambda_0$ and  $C(V,\lambda_0)$, such that for all $0<\lambda<\lambda_0$, there exists an eigenpair $(E^\lambda,\psi^\lambda)$, solution of the eigenvalue problem
\begin{equation}\label{eq:nonper}
\left(-\partial_x^2+\lambda V(x)\right)\psi^\lambda(x)\ =\ E^\lambda\psi^\lambda(x)
\end{equation}
with negative eigenvalue of the order $\lambda^2$. Specifically, 
\begin{align}
 \left| E^\lambda\ - \ \Big[-\frac{\lambda^2}{4}\ \left(\int_\mathbb{R} V\right)^2\Big]\ \right| \ &\leq \ C \lambda^{5/2} \ ,
\label{Elambda} \\
 \sup_{x\in\RR}\left|\psi^\lambda(x) \ - \ \exp\left(\frac{\lambda}{2}\left(\int_\mathbb{R} V\right)|x|\right) \right| \ &\leq \ C \lambda^{1/2}.
 \label{psi-lambda}
 \end{align}
The eigenvalue, $E^\lambda$, is unique in the neighborhood defined by~\eqref{Elambda}, and the corresponding eigenfunction, $\psi$, is unique up to a multiplicative constant.
\end{Theorem}

\begin{Remark}\label{rmk:homogenized-Q0} Theorem~\ref{thm:Qzero} shows, and is essentially proved by demonstrating, that for small positive $\lambda$, the leading order behavior of the eigenstate $\left(E^\lambda,\psi^\lambda(x)\right)$ is a scaling of the unique eigenstate of the attractive Dirac delta potential: 
\[\left(\ E^\lambda\ ,\ \psi^\lambda(x)\ \right) \ \approx \ \left(\ \lambda^2\theta_0^2\ ,\ g_0(\lambda x)\ \right) \ , \]
where $\theta_0 = -\frac{1}{2}\int_{\mathbb R } V>0$ and $g_0(y) = e^{-\theta_0 |y|}$ satisfy
\begin{equation}
\left[-\partial_y^2\ +\ \int_{\mathbb R}V\ \cdot\ \delta(y) \right]\ g_0(y) = -\theta_0^2\ g_0(y).
\end{equation}
\end{Remark}
The error bounds in Theorem~\ref{thm:Qzero} are not optimal. However, the bootstrap argument of  Appendix~\ref{proof-of-Corollary-per}  can be used to recover a higher order expansion on $E^\lambda$, similar to that obtained in~\cite{simon1976bound}.
\begin{Corollary}\label{cor:Elambdaprecise}
Assume $(1+|x|^2)V\in L^1$, and $\widehat V(0)\ =\ \int_\mathbb{R} V(z)\ dz\ <\ 0$ . Then $E^\lambda$, as defined in Theorem~\ref{thm:Qzero}, satisfies the precise estimate:
%%%
%%%
\begin{equation}
 E^\lambda\ =\ -\lambda^2\ \left[\theta(\lambda)\right]^2 ,\text{ with } \theta(\lambda) \ = \ -\frac12 \int_\mathbb{R} V \ -\ \frac14\ \lambda\iint_{\mathbb{R}^2}\ V(x)|x-y|V(y)\ dx dy \ \ + \ \O(\lambda^{3/2})\ .
\label{Elambda-precise}
\end{equation}
\end{Corollary}
Simon~\cite{simon1976bound} and Klaus~\cite{klaus:77} prove  expansion~\eqref{Elambda-precise}, under the conditions: $(1+|x|)|V(x)\in L^1(\RR)$ and  $\int_\RR V\le0$, with the error term $o(\lambda)$.  Corollary~\ref{cor:Elambdaprecise} gives a sharper error term under a more stringent decay condition on $V$.
That Theorem \ref{thm:Qzero} implies Corollary~\ref{cor:Elambdaprecise} is proved in Appendix~\ref{proof-of-Corollary-per} .

\begin{Theorem}[ $Q$ non-trivial, $1-$periodic]\label{thm:per_result} 
Let $Q$ be continuous, $1-$periodic, and let $V$ be such that $\int_{\mathbb{R}}(1+|x|)V(x) dx<\infty$ and $V\in L^\infty$.
 Let $E_{b_*}:k\in(-1/2,1/2]\to\mathbb{R}$ denote the band dispersion function associated with the $(b_*)^{th}$ band of the continuous spectrum of $-\partial_x^2 + Q(x)$. Fix a spectral band edge of the $(b_*)^{th}$ band; thus $E_*=E_{b_*}(k_*)$, where $k_*=0$ or $k_*=1/2$ (see Lemma~\ref{lem:band-edge}). 
 
Assume either
\begin{equation} \partial_k^2 E_{b_*}(k_*) > 0 \quad \text{ and }\quad \int_{\mathbb{R}} |u_{b_*}(x;k_*)|^2 V(x)dx <~0,\label{Epp+}\end{equation}\\
 or 
 \begin{equation}
  \partial_k^2 E_{b_*}(k_*) < 0\quad \text{ and }\quad \int_{\mathbb{R}}\ |u_{b_*}(x;k_*)|^2 V(x)dx >~0.\label{Epp-}\end{equation}
 Then, there is a positive constants, $\lambda_0$ and $C=C(\lambda_0,V,Q)$,  such that for all $\lambda<\lambda_0$, the following assertions hold:
 \begin{enumerate}
\item  There exists an eigenpair $\left(E^\lambda,\psi^\lambda(x)\right)$ of the eigenvalue problem
\begin{equation}\label{eq:per}
\left(-\partial_x^2+Q(x)+\lambda V(x)\right)\psi^\lambda(x)\ =\ E^\lambda\psi^\lambda(x),\ \ \psi^\lambda\in L^2(\mathbb R)\ .
\end{equation}
\item   Define
\begin{equation}
\alpha_0 \ \equiv \ \frac{ \int_{-\infty}^{\infty} |u_{b_*}(x;k_*)|^2 V(x)dx}{\frac1{4\pi^2}\partial_k^2 E_{b_*}(k_*) }\ <\ 0,\label{alpha0-def}
 \end{equation}
 where the inequality holds by ~\eqref{Epp+} and~\eqref{Epp-}.
Then, $E^\lambda$ and $\psi^\lambda(x)$ satisfy the following approximations:
\begin{align}
 \left| E^\lambda\ - \ \left(\ E_{b_*}(k_*) \ + \ \lambda^2 E_2\ \right) \ \right| \ &\leq \ C \lambda^{2+1/4} \ ,
 \label{Elambda-Qper} \\
\sup_{x\in\mathbb R}\ \left|\ \psi^\lambda(x) \ - \ u_{b_*}(x;k_*)\exp(\lambda\alpha_0|x|) \right| \ &\leq \ C \lambda^{1/4} \ ,\label{psilambda-Qper}
 \end{align}
where
\begin{equation}
 E_2 \ = \ -\frac{ \left| \int_{-\infty}^{\infty} |u_{b_*}(x;k_*)|^2 V(x)dx\right|^2}{\frac1{2\pi^2}\partial_k^2 E_{b_*}(k_*) }\ .\label{E2-def}
 \end{equation}
 Note  that the direction of bifurcation 
of $E^\lambda$ is given by:
 \[  {\rm sgn}\left( E_2\right) =\ -{\rm sgn}\left( \partial_k^2 E_{b_*}(k_*)\right)\ .\]
\item The eigenstate, $(E^\lambda,\psi^\lambda)$, is unique (up to a multiplicative constant for $\psi^\lambda$) in the neighborhood defined by~\eqref{Elambda-Qper},\eqref{psilambda-Qper}.
\end{enumerate}
\end{Theorem}

\begin{Remark}\label{rmk:homogenized-per}
By Theorem~\ref{thm:per_result}, the bifurcating eigenvalue $E^\lambda$ lies in the spectral gap of $-\partial_x^2 + Q(x)$ at a distance $\mathcal{O}(\lambda^2)$ near the spectral edge $E_*$; see Figure~\ref{fig:SketchOfSpectrumQ}.
Moreover, $E_2$ is the unique eigenvalue and $g_0(y) = e^{\alpha_0 |y|}$ is the unique (up to multiplication by a constant) eigenfunction of the effective (homogenized) Hamiltonian:
\[
H_{\rm eff}\ =\ -\frac{d}{dy}\ \frac1{8\pi^2}\partial_k^2 E_{b_*}(k_*)\ \frac{d}{dy}\ +\ \int_{-\infty}^{\infty} |u_{b_*}(x;k_*)|^2 V(x) dx\times \delta(y) .
\]
\end{Remark}
The following refinement of Theorem~\ref{thm:per_result} can be proved via the bootstrap argument presented in Appendix~\ref{proof-of-Corollary-per}.

\begin{Corollary}\label{cor:Elambdaprecise-per}
Assume $\int_{\mathbb{R}}(1+|x|^2)V(x) dx<\infty$ and that the hypotheses of Theorem~\eqref{thm:per_result} hold. Then,
\begin{align}
 E^\lambda-E_{b_*}(k_*)\ &=\ \lambda^2(E_2+\lambda E_3)\ +\ \O(\lambda^{3+1/4}) 
=\ -\lambda^2 \frac{8\pi^2}{\partial_k^2E_{b_*}(k_*)} \left[ \Theta(\lambda) \right]^2,
\label{Elambda-precise-per}\end{align}
where $E_2$ is as in~\eqref{E2-def},
\begin{multline*}E_3\equiv\frac{-8\pi^4}{(\partial_k^2 E_{b_*}(k_*))^2}\left(\int_{-\infty}^{\infty} |u_{b_*}(x;k_*)|^2 V(x)\ dx\right)\\
\times\left(\iint_{\RR^2} V(x)|u_{b_*}(x;k_*)|^2|x-y| |u_{b_*}(y;k_*)|^2V(y) \ dx\ dy \right)\ , 
\end{multline*}
and 
\begin{align}
\label{eq:Theta-relation}
\Theta(\lambda) =& -\frac{1}{2} \int_{\RR} |u_{b_*}(x;k_*)|^2 V(x) \ dx \\ 
&- \frac{1}{4} \lambda \ \frac{8\pi^2}{\partial_k^2E_{b_*}(k_*)}  \iint_{\RR^2} V(x)|u_{b_*}(x;k_*)|^2|x-y| |u_{b_*}(y;k_*)|^2V(y) \ dx\ dy + \mathcal{O}(\lambda^{1+1/4}). \nn 
\end{align}
\end{Corollary}

\begin{Remark}
For the case $Q \equiv 0$, the spectrum consists of only one semi-infinite band which we can label the $b = 0$ band. In this case, $u_0(x;k_* = 0) = 1$ and $E_0(k) = 4\pi^2 k^2$. Therefore, to leading order, relation~\eqref{eq:Theta-relation} simplifies to the result of Corollary~\ref{cor:Elambdaprecise} and the two results are consistent. 
\end{Remark}
 
\section{Key general technical results} \label{sec:gentech}

In this section, we study the operator $\widehat{\mathcal{L}}_0[\theta]$, defined by:
\begin{equation}
\widehat{f}(\xi)\ \mapsto\ \widehat{\mathcal{L}}_0[\theta]\widehat{f}(\xi) \equiv \left(4\pi^2 A \xi^2 + {\theta}^2\right) \widehat{f}(\xi) - B\ \chi\left(|\xi|<\lambda^{-\beta}\right) \int_{\mathbb R}\chi\left(|\eta|<\lambda^{-\beta}\right)\ \widehat{f}(\eta)\ d\eta .
\label{hatcalL0-def}
\end{equation}
Here, $A,B$ and $\beta$ are fixed positive constants. The operator $\widehat{\mathcal{L}}_0[\theta] $ appears
in the bifurcation equations we derived via the Lyapunov-Schmidt reduction; see Section~\ref{sec:outline}.

In $x-$ space, we have that $\mathcal{L}_0[\theta]$ is a rank one perturbation of $-A\partial_y^2+\theta^2$:
\begin{equation}
\mathcal{L}_0[\theta] f \ \equiv \ (-A\partial_y^2+\theta^2)f(y) \ - \ \frac{2 B}{ \lambda^{\beta}}\ 
\left\langle \frac{2}{ \lambda^{\beta}}\sinc \left(\frac{2\pi }{ \lambda^{\beta}}\ \cdot \right),f(\cdot)\right\rangle_{L^2}\ 
\sinc \left(\frac{2\pi y}{ \lambda^{\beta}} \right) ,
\label{calL0-def}
\end{equation}
where $\sinc(z)=\sin(z)/z$.
$\mathcal{L}_0[\theta]$ is a band-limited regularization of the operator: 
\begin{equation}
 \left(\ H^{A,B}\ +\ \theta^2\right)\ f \ \equiv \ \left(-A\partial_y^2 \ - \ B \delta(y)\ +\ \theta^2\right)f\ ,
 \label{HAB-def}
 \end{equation}
appearing in the effective equations governing the leading order behavior of bifurcating eigenstates; see Remarks~\ref{rmk:homogenized-Q0} and~\ref{rmk:homogenized-per}.

\subsection{The operator $\widehat{\mathcal{L}} _0$}

\begin{Lemma} \label{lem:homogeneous}
Fix constants $A>0$, $B>0$ and $\beta>0$. Define, for $\theta^2>0$, the linear operator
\begin{equation}\label{eq:homogeneous}
\widehat{f}(\xi)\ \mapsto\ \widehat{\mathcal{L}}_0[\theta]\widehat{f}(\xi) \equiv \left(4\pi^2 A \xi^2 + {\theta}^2\right) \widehat{f}(\xi) - B\ \chi\left(|\xi|<\lambda^{-\beta}\right) \int_{\mathbb R} \chi\left(|\eta|<\lambda^{-\beta}\right)\ \widehat{f}(\eta)\ d\eta .
\end{equation}
Note that $\widehat{\mathcal{L}}_0[\theta]:L^1(\mathbb{R})\to L^{1,-2}(\mathbb{R})$; see~\eqref{Lps-def}. 
\begin{enumerate}
\item There exists a unique $\theta_0^2>0$ such that $\widehat{\mathcal{L}}_0[\theta_0]$
has a non-trivial kernel.
\item The ``eigenvalue'' $\theta_0^2$ is the unique positive solution of
\begin{equation} \label{eq:theta0} 
 1 \ - \ B\int_{\mathbb R}\frac{\chi\left(|\xi|<\lambda^{-\beta}\right)}{4\pi^2A\xi^2 + \theta_0^2}d\xi \ = \ 0 \ .
\end{equation}
\item The kernel of $\widehat{\mathcal{L}} _0[\theta_0]$ is given by:
\begin{equation} \label{eq:f0}
{\rm kernel}\left(\widehat{\mathcal{L}} _0[\theta_0]\right)\ =\ {\rm span}\left\{ \widehat{f}_0(\xi) \right\},\ \ {\rm where}\ \ 
 \widehat{f}_0(\xi) \equiv\ \frac{\chi\left(|\xi|<\lambda^{-\beta}\right)}{4\pi^2A\xi^2 + \theta_0^2} \ .
\end{equation}
\item
 $\theta_0=\theta_0(\lambda)$ can be approximated as follows:
\begin{equation} \label{eq:esttheta0}
\left| \theta_0 \ - \ \frac{B}{2\sqrt{A}} \right| \ \le\frac{\theta_0}{2\pi^2}\ \frac{B}{A}\ \lambda^{\beta} .
\end{equation}
\item Define $g(x)=\exp(\alpha_0|x|)$, with $\alpha_0=-\frac{B}{2A}<0$. Then one has
\begin{equation} \label{eq:asymptotic-f0}
\sup_{x\in\RR} \Big\vert\ \mathcal{F}^{-1}\left\{ \widehat{f}_0 \right\}(x) \ - \ \frac{1}{B}g(x)\ \Big\vert \ \leq C(A,B) \lambda^{\beta}.
\end{equation}
\end{enumerate}
\end{Lemma} 
\begin{proof}
First note, by rearranging terms in the equation $\widehat{\mathcal{L}} _0[\theta_0] \widehat{g}=0$, that any element, $\widehat{g}(\xi)$, of the kernel of $\widehat{\mathcal{L}} _0[\theta]$, is a constant multiple of the function 
$ \widehat{f}_\lambda(\xi;\theta)\equiv \chi(|\xi|<\lambda^{-\beta})\times (4\pi^2A\xi^2 + \theta^2)^{-1}$.
Thus, if $\widehat{g}$ is non-trivial then it is strictly positive or strictly negative and therefore $\int_\mathbb{R}\widehat{g}\ne0$.

Next, note that a necessary condition for $\widehat{g}$ to lie in the kernel of $\widehat{\mathcal{L}} _0[\theta]$ is that equation~\eqref{eq:theta0} holds. To see this, divide the equation $\widehat{\mathcal{L}} _0[\theta_0]\widehat{g}=0$ by $\left(4\pi^2A\xi^2 + {\theta_0}^2\right)$ and integrate $d\xi$ over $\RR$. This yields:
 \begin{equation}
 \int_{-\infty}^{\infty} g(\xi) \ d\xi\ \times\ \Big[\ 1\ - \ B \int_{-\infty}^{\infty}\frac{\chi\left(|\xi|<\lambda^{-\beta}\right)}{4\pi^2A\xi^2 + \theta^2}d\xi \Big]\ = \ 0,
 \label{theta0-nec}
 \end{equation}
By the above discussion, if $\widehat{g}$ is non-trivial then $\int_\mathbb{R} \widehat{g}\ne0$. Hence
 $\theta^2$ satisfies
\[
  J({\theta}^2) \equiv1 \ - \ B\int_{-\infty}^{\infty}\frac{\chi\left(|\xi|<\lambda^{-\beta}\right)}{4\pi^2A\xi^2 + \theta^2}d\xi \ = \ 0.
  \]
 Since $J:(0,\infty)\to\mathbb{R}$ is smooth, $J'(X)>0$, $\lim\limits_{X\to0} J(X)=-\infty$ and $\lim\limits_{X\to\infty} J(X)=1$, the function $J$ has a unique positive root, which we denote by $\theta_0^2$.
 One can check by direct substitution and the condition $J(\theta_0^2)=0$, that
 any multiple of 
 \begin{equation}
 \widehat{f}_0(\xi)\ \equiv\ \widehat{f}_\lambda(\xi;\theta_0)\ =\ \chi(|\xi|<\lambda^{-\beta})\times (4\pi^2A\xi^2 + \theta_0^2)^{-1}
 \label{f0def}
 \end{equation}
satisfies $\widehat{\mathcal{L}} _0[\theta_0]\widehat{f}_0(\xi)=0$ .
 
The approximation to $\theta_0(\lambda)$,~\eqref{eq:esttheta0}, is obtained as follows. Let $\theta_0^2$ denote the unique solution of $J(\theta_0^2)=0$ and $\theta_0$ its positive square root. Then, 
\begin{align}
\frac{1}{B}\ &=\ \int_{\mathbb R}\frac{\chi\left(|\xi|<\lambda^{-\beta}\right)}{4\pi^2\ A\ \xi^2 + \theta_0^2}\ d\xi\ 
 = \ 
\int_{\mathbb R}\frac{1+(\chi\left(|\xi|<\lambda^{-\beta}\right)-1)}{4\pi^2\ A\ \xi^2 + \theta_0^2}d\xi  \nn \\
&= \ \frac{1}{2\sqrt{A}\ \theta_0}+\int_{\mathbb R}\frac{\chi\left(|\xi|<\lambda^{-\beta}\right)-1}{4\pi^2A\xi^2 + \theta_0^2}d\xi \ . 
\label{theta0-est1}\end{align}
The last term can be bounded as follows:
\begin{equation}
\left| \int_{\mathbb R}\frac{1-\chi\left(|\xi|<\lambda^{-\beta}\right)}{4\pi^2A\xi^2 + \theta_0^2}d\xi \right| \ = \ 
 \int_{|\xi|\ge\lambda^{-\beta}}\ \frac{d\xi}{4\pi^2\ A\ \xi^2 + \theta_0^2}\ \leq \ \int_{|\xi|\ge\lambda^{-\beta}}\ \frac{d\xi}{4\pi^2\ A\ \xi^2}
 \le\ \frac{ \lambda^\beta}{2\pi^2A}.
 \label{theta0-est2}
\end{equation}
Relations~\eqref{theta0-est1},~\eqref{theta0-est2}, after rearrangement of terms, yield~\eqref{eq:esttheta0}.

 Finally, let us turn to the asymptotic expression for $\mathcal{F}^{-1}\left\{ \widehat{f}_0 \right\}(x)$ given in~\eqref{eq:asymptotic-f0}. By residue computation, one has
$\widehat{g}(\xi) \ 
\ =\ \dfrac{-2\alpha_0 }{4\pi^2|\xi|^2+\alpha_0^2}\ \ =\ \dfrac{B}{4\pi^2A|\xi|^2+\frac{B^2}{4A}}.$
It follows that 
\begin{multline*} 
\sup_{x\in\RR} \Big\vert\ \mathcal{F}^{-1}\left\{ \widehat{f}_0 \right\}(x)  -  \frac{1}{B}g(x)\ \Big\vert \ \leq \ \Big\Vert\ \widehat{f}_0  - \frac{1}{B}\widehat{g}\ \Big\Vert_{L^1} \ \leq \ \int_{\RR} \left| \frac{\chi\left(|\xi|<\lambda^{-\beta}\right)}{4\pi^2A\xi^2 + \theta_0^2} \ - \ \frac{1}{4\pi^2A|\xi|^2+\frac{B^2}{4A}}\right| \ d\xi \\
 \leq \ \int_{\RR} \chi\left(|\xi|<\lambda^{-\beta}\right)\left| \frac{1}{4\pi^2A\xi^2 + \theta_0^2} \ - \ \frac{1}{4\pi^2A|\xi|^2+\frac{B^2}{4A}}\right| \ d\xi 
\ + \ \int_{\RR} \left| \frac{1-\chi\left(|\xi|<\lambda^{-\beta}\right)}{4\pi^2A\xi^2 + \frac{B^2}{4A}} \right| \ d\xi .
\end{multline*}
A bound on the second term follows from~\eqref{theta0-est2}. The first term is easily bounded, using~\eqref{eq:esttheta0}, by  $C(A,B)\lambda^\beta$, with some constant $C(A,B)>0$.
Estimate~\eqref{eq:asymptotic-f0} follows, and the proof of Lemma~\ref{lem:homogeneous} is now complete.
\end{proof}

We shall also require a result on the solvability of the inhomogeneous equation
\begin{equation} \label{eq:inhom}
\left( \widehat{\mathcal{L}}_0[\theta_0]\widehat{\varphi} \right)(\xi) = \widehat{h}(\xi), 
\end{equation}
where $\widehat{\mathcal{L}}_0[\theta_0]$ is defined in~\eqref{eq:homogeneous}.

\begin{Lemma} \label{lem:inhom_sol}
The equation~\eqref{eq:inhom} is solvable if and only if $\widehat{h}$ is such that $\chi\left(|\xi|<\lambda^{-\beta}\right) \widehat{h}(\xi)=\widehat{h}(\xi)$ and satisfies the orthogonality condition
\begin{equation}
\left\langle \widehat{f}_0, \widehat{h} \right\rangle_{L^2(\RR)} = 0,
\label{orthog}
\end{equation}
 where $\widehat{f}_0$, displayed in~\eqref{eq:f0}, spans the kernel of 
 $\widehat{\mathcal{L}}_0 [\theta_0]$. In that case,
 \begin{enumerate}
 \item any solution of the inhomogeneous equation~\eqref{eq:inhom} is of the form
\begin{equation}\label{eq:solinhom}
\widehat{\varphi}(\xi) \ \equiv \ (C+\widehat{h}(\xi)) \widehat{f}_0(\xi) \ \equiv \ (C+ \widehat{h}(\xi))\ \frac{\chi\left(|\xi|<\lambda^{-\beta}\right) }{4\pi^2 A \xi^2 + \theta_0^2},
\end{equation}
for some constant $C$.
\item The unique solution of~\eqref{eq:inhom} such that $\int_\RR \widehat{\varphi}=0$ is obtained by choosing $C=0$:
\begin{equation}\label{vphi-def}
\widehat{\varphi}(\xi) \ \equiv \ \widehat{h}(\xi)\ \widehat{f}_0(\xi).
\end{equation}
\end{enumerate}
\end{Lemma}
\begin{proof}The solvability condition $\chi\left(|\xi|<\lambda^{-\beta}\right) \widehat{h}(\xi)=\widehat{h}(\xi)$ is straightforward, and~\eqref{orthog} is obtained by taking the inner product of~\eqref{eq:inhom} with $\widehat{f}_0$, and using that $\widehat{\mathcal{L}}_0 [\theta_0]$ is symmetric, and $\widehat{\mathcal{L}}_0 [\theta_0]\widehat{f}_0=0$.

To show that~\eqref{eq:solinhom} solves the inhomogeneous equation~\eqref{eq:inhom} we simply insert the function~\eqref{eq:solinhom} into~\eqref{eq:inhom}, and use the properties: $\widehat{\mathcal{L}}_0 (\theta_0)\widehat{f}_0=0$ and $\left\langle \widehat{f}_0, \widehat{h} \right\rangle_{L^2} = 0$. This gives 
\begin{align*}
\left( \widehat{\mathcal{L}} _0[\theta_0]\widehat{\varphi}\right)(\xi) & = (4\pi^2 A \xi^2 + \theta_0^2) \widehat{h}(\xi)  \widehat{f}_0(\xi) - B\ \chi\left(|\xi|<\lambda^{-\beta}\right) \int_{-\infty}^\infty \widehat{h}(\eta)  \widehat{f}_0(\eta) d\eta \\
& = (4\pi^2 A \xi^2 + \theta_0^2)\ \frac{\chi\left(|\xi|<\lambda^{-\beta}\right) \widehat{h}(\xi)}{4\pi^2 A \xi^2 + \theta_0^2}\ -\ B\ \chi\left(|\xi|<\lambda^{-\beta}\right)   \left\langle \widehat{f}_0, \widehat{h} \right\rangle_{L^2(\RR)}\ =\ \widehat{h}(\xi) .
\end{align*}

The converse clearly holds by Lemma~\ref{lem:homogeneous}, since the difference of solutions of the inhomogeneous equation solves the homogeneous equation~\eqref{eq:homogeneous}.
 Finally, using the orthogonality condition $\left\langle \widehat{f}_0, \widehat h \right\rangle_{L^2} = 0$, one has that $\int_{\mathbb R}\widehat\varphi=C\int_\RR\widehat{f}_0 =0$ if and only if $C=0$.
\end{proof}

\subsection{A perturbation result for $\widehat{\mathcal{L}} _0$}
As discussed in the introduction, our strategy is to obtain a reduction  of the eigenvalue problem for $H_{Q+\lambda V}$ to an eigenvalue problem (the bifurcation equation) for functions supported at energies near the band-edge. These reduced equations have a general form which we study in this section.

Let $\mathcal{Z}_1$ and $\mathcal{Z}_2$ denote Banach spaces with 
 $\mathcal{Z}_1, \mathcal{Z}_2\subset L^1_{\rm loc}$. Assume that 
 for any $(f,g)\in \mathcal{Z}_1\times \mathcal{Z}_2$, 
\begin{equation}\label{norms}
\left| \left\langle f, g \right\rangle_{L^2} \right| \lesssim \big\Vert f \big\Vert_{\mathcal{Z}_2} \big\Vert g \big\Vert_{\mathcal{Z}_1} , \quad \quad \big\Vert f g \big\Vert_{\mathcal{Z}_2} \lesssim \big\Vert f \big\Vert_{\mathcal{Z}_2} \big\Vert g \big\Vert_{L^\infty}, \quad \text{ and } \quad \big\Vert (1+\xi^2)^{-1}f \big\Vert_{\mathcal{Z}_2} \lesssim \big\Vert f \big\Vert_{\mathcal{Z}_1}.
\end{equation}
Furthermore, we also assume that $\widehat{f}_0\in \mathcal{Z}_1\cap \mathcal{Z}_2$, where $\left(\theta_{0}^2, \widehat{f}_0\right)$ is the unique normalized solution of the homogeneous equation $\widehat{\mathcal{L}}_0[\theta]\widehat f=0$; see Lemma~\ref{lem:homogeneous}.
\begin{Remark}
In order to prove Theorems~\ref{thm:Qzero} and~\ref{thm:per_result}, we shall apply Lemma~\ref{lem:technical}, below, with
\begin{itemize}
\item Case $Q\equiv 0$: $\left(\mathcal{Z}_1,\mathcal{Z}_2\right) = \left(L^{\infty},L^1\right)$ in the case $Q=0$; and 
\item $Q$ non-trivial, $1-$periodic: $\left(\mathcal{Z}_1,\mathcal{Z}_2\right) = \left(L^{2,-1},L^{2,1} \right)$, where $L^{2,s}$ is the space of locally integrable functions such that
\[ \big\Vert F \big\Vert_{L^{2,s}} \ \equiv \ \big\Vert (1+|\xi|^2)^{s/2} F \big\Vert_{L^2(\mathbb R_\xi)} \ < \ \infty.\]
\end{itemize}
It is straightforward to check that such spaces satisfy~\eqref{norms}, and $\widehat{f}_0\in \mathcal{Z}_1\cap \mathcal{Z}_2$.
\end{Remark}

We seek a solution of the equation:
\begin{equation} \label{eq:near_general}
\widehat{\mathcal{L}} _0[\theta]\widehat f\ =\  R\big(\widehat f\,\big)\ ,
\end{equation}
where $\widehat{\mathcal{L}}_0(\theta)$ is the operator defined in~\eqref{eq:homogeneous} and the mapping 
 $\widehat f\mapsto  R\big(\widehat f\,\big)$ is linear and satisfies the following properties:

\noindent{\bf Assumptions $R_{\alpha,\beta}$:}\\ There exist constants $\alpha>0$, $\beta>0$ and $C_R>0$ such that for any $\widehat f\in \mathcal{Z}_2$
\begin{equation}
 \chi\left(|\xi|<\lambda^{-\beta}\right)  R\big(\widehat f\,\big)(\xi) \ = \  R\big(\widehat f\,\big)(\xi)\ , \quad \text{ and } \quad 
 \left\|  R\big(\widehat f\,\big) \right\|_{\mathcal{Z}_1} \ \leq\ C_R\lambda^{\alpha} \left\|\widehat f\ \right\|_{\mathcal{Z}_2}\ .
\label{assumptionsR}
\end{equation}

In the above setting we have the following 

\begin{Lemma} \label{lem:technical}
Let $(\theta_0^2,\widehat{f}_0(\xi))$ be the solution of $\widehat{\mathcal{L}} _0(\theta_0)\widehat{f}_0=0$, as defined in Lemma~\ref{lem:homogeneous}, where $A,B$ and $\beta>0$ are fixed. Let $R:\widehat f\in\mathcal{Z}_2\to\mathcal{Z}_1$ be a linear mapping satisfying assumptions $R_{\alpha,\beta}$ displayed in~\eqref{assumptionsR}, where $\mathcal{Z}_1, \mathcal{Z}_2$ satisfy~\eqref{norms}. Then there exists $\lambda_0>0$ such that for any $0<\lambda<\lambda_0$, the following holds:
\begin{enumerate}
\item There exists a unique solution $\left(\theta,\widehat f(\xi)\ \right)\in\RR^+\times \mathcal{Z}_2$ of the equation 
~\eqref{eq:near_general}, such that 
 \[\left\| \widehat{f} - \widehat{f}_0 \right\|_{\mathcal{Z}_2} \leq C \lambda^{\alpha}, \qquad 
 \int_{-\infty}^\infty \widehat{f}(\xi)-\widehat{f}_0(\xi)\ d\xi \ = \ 0, \]
with $C=C(A,B,C_R,\beta)$, independent of $\lambda$. 
\item Moreover, one has
$
\widehat{f}(\xi) \ = \ \chi\left(|\xi|<\lambda^{-\beta}\right) \widehat{f}(\xi)$ and $\left| \theta^2 - \theta_{0}^2 \right| \leq C\lambda^{\alpha} \ .
$
\end{enumerate}
\end{Lemma} 
\begin{proof}
Our strategy is to use a fixed point argument. We seek a solution $(\theta^2,f)$ to 
~\eqref{eq:near_general} of the form 
 \[
\theta^2\ \equiv\ \theta_0^2 \ +\ \theta_1^2 \qquad \text{and} \qquad \widehat{f}\ \equiv\ \widehat{f}_0 \ +\ \widehat{f}_1.
\]
Clearly, any solution $\widehat f$ of~\eqref{eq:near_general} satisfies $ \widehat{f}(\xi) = \chi\left(|\xi|<\lambda^{-\beta}\right) \widehat{f}(\xi)$. Therefore, since one has, by definition, $ \widehat{f}_0(\xi) = \chi\left(|\xi|<\lambda^{-\beta}\right) \widehat{f}_0(\xi)$, it follows that $ \widehat{f}_1(\xi) = \chi\left(|\xi|<\lambda^{-\beta}\right) \widehat{f}_1(\xi)$.
Substitution of these expressions into~\eqref{eq:near_general} yields
\begin{multline*}
(4\pi^2 A\xi^2 + \theta^2) \chi\left(|\xi|<\lambda^{-\beta}\right)\left(\widehat{f}_0+\widehat{f}_1\right)(\xi) \\
- \chi\left(|\xi|<\lambda^{-\beta}\right) B \int_{-\infty}^{\infty}\chi\left(|\eta|<\lambda^{-\beta}\right)\left(\widehat{f}_0+\widehat{f}_1\right)(\eta)d\eta 
= R\left(\widehat{f}_0+\widehat{f}_1\right)(\xi).
\end{multline*}
Rearranging terms yields the following equation for $\widehat{f}_1$, in which $\theta_1^2$ is a parameter to be determined:
\begin{equation}
\left(\widehat{\mathcal{L}} _{0}[\theta_0]\widehat{f}_1\right)(\xi) = -\theta_1^2\left(\widehat{f}_0+\widehat{f}_1\right)(\xi) + R\left(\widehat{f}_0+\widehat{f}_1\right)(\xi).
\label{eq:f1-1}
\end{equation}
By Lemma~\ref{lem:inhom_sol},~\eqref{eq:f1-1} is solvable in $L^2$ only if the right hand side is $L^2$- orthogonal to $\widehat{f}_0$:
\[
\left\langle \widehat{f}_0, -\theta_1^2\left(\widehat{f}_0+\widehat{f}_1\right) + R\left(\widehat{f}_0+\widehat{f}_1\right) \right\rangle_{L^2} = 0.
\]
Solving for $\theta_1^2$, we obtain
\begin{equation}\label{eq:theta}
\theta_1^2 = \frac{\left\langle \widehat{f}_0,R\left(\widehat{f}_0+\widehat{f}_1\right)\right\rangle_{L^2}}{\left\langle \widehat{f}_0,\widehat{f}_0\right\rangle_{L^2} + \left\langle \widehat{f}_0,\widehat{f}_1\right\rangle_{L^2}} \ .
\end{equation}
In summary, equation~\eqref{eq:near_general} can be rewritten equivalently as two coupled equations in terms of $\widehat{f}_1$ and $\theta_1^2$:~\eqref{eq:f1-1}--\eqref{eq:theta}.

Substitution of $\theta_1^2$ in~\eqref{eq:f1-1}, or equivalently projecting the right hand side of~\eqref{eq:f1-1} onto the orthogonal complement of ${\rm span}\{\widehat{f}_0\}$, yields the following closed equation for $\widehat{f}_1$: 
\begin{equation}
 \left( \widehat{\mathcal{L}} _0[\theta_0]\widehat{f}_1 \right)(\xi) = -\frac{\left\langle \widehat{f}_0,R\left(\widehat{f}_0+\widehat{f}_1\right)\right\rangle_{L^2}}{\left\langle \widehat{f}_0,\widehat{f}_0\right\rangle_{L^2} + \left\langle \widehat{f}_0,\widehat{f}_1\right\rangle_{L^2}}\left(\widehat{f}_0+\widehat{f}_1\right)(\xi) + R\left(\widehat{f}_0+\widehat{f}_1\right)(\xi).
 \label{f1eqn-projected}
 \end{equation}

By Lemma~\ref{lem:inhom_sol}, $\widehat{f}_1$ is a solution of~\eqref{f1eqn-projected} with $\int_\mathbb{R}\widehat{f}_1=0$ if and only if :
\begin{equation}
\widehat{f}_1(\xi) \ = \ \mathcal G (\widehat{f}_1)(\xi), \label{fixedpt}
\end{equation}
where
\begin{equation}
\mathcal{G}(\widehat{f}_1)(\xi)\ \equiv \ \frac{\chi\left(|\xi|<\lambda^{-\beta}\right)}{4\pi^2A\xi^2 + \theta_0^2}\left( -\frac{\left\langle \widehat{f}_0,R\left(\widehat{f}_0+\widehat{f}_1\right)\right\rangle_{L^2}}{\left\langle \widehat{f}_0,\widehat{f}_0\right\rangle_{L^2} + \left\langle \widehat{f}_0,\widehat{f}_1\right\rangle_{L^2}}\left(\widehat{f}_0+\widehat{f}_1\right)(\xi) + R\left(\widehat{f}_0+\widehat{f}_1\right)(\xi) \right).
\label{def:G}\end{equation}
We solve the fixed point equation~\eqref{fixedpt} by the contraction mapping principle.
Once $\widehat{f}_1$ has been obtained, $\theta_1^2$ is determined using~\eqref{eq:theta}.

 Introduce
\begin{equation}\label{def:S}
\mathcal S \ = \ \left\{ \widehat{f} \in \mathcal Z_2\ :
 \ 
\big\Vert \widehat{f} \big\Vert_{\mathcal Z_2}\leq C_H \lambda^\alpha \right\},\ \ {\rm for\ some\ fixed}\ C_H>0\ .
\end{equation}
Note that $\mathcal{S}$ is a closed subset of the Banach space ${\mathcal Z_2}$. 
We next show that there exists $\lambda_0>0$ such that for all $0<\lambda<\lambda_0$:
$\mathcal{G}:\mathcal{S}\to\mathcal{S}$
and $\mathcal{G}$ is a contraction mapping.
As a consequence, it will follow that for $0<\lambda<\lambda_0$, there is a unique solution $\widehat{f}_1\in\mathcal{S}$ of the equation $\widehat{f}_1=\mathcal{G}(\widehat{f}_1)$ and therefore of~\eqref{f1eqn-projected}. Moreover, $\| \widehat{f}_1\|\lesssim \lambda^\alpha$ by definition of $\mathcal S$, and one can check:
\[ \int_\mathbb{R}\widehat{f}_1=\int_\mathbb{R}\mathcal{G}(\widehat{f}_1)= \int_\mathbb{R}\widehat{f}_0(\xi)\left( -\frac{\left\langle \widehat{f}_0,R\left(\widehat{f}_0+\widehat{f}_1\right)\right\rangle_{L^2}}{\left\langle \widehat{f}_0,\widehat{f}_0\right\rangle_{L^2} + \left\langle \widehat{f}_0,\widehat{f}_1\right\rangle_{L^2}} \left(\widehat{f}_0+\widehat{f}_1\right)(\xi) + R\left(\widehat{f}_0+\widehat{f}_1\right)(\xi) \right) d\xi =  0 .\]

It then remains to obtain an estimate of ${\theta_1}^2={\theta_0}^2-\theta^2$. From~\eqref{eq:theta}, one has
\[ \big| \theta_1^2\big| \ \leq \ \left| \left\langle \widehat{f}_0,R\left(\widehat{f}_0+\widehat{f}_1\right)\right\rangle_{L^2}\right|\left| \frac{1}{\left\langle \widehat{f}_0,\widehat{f}_0\right\rangle_{L^2} + \left\langle \widehat{f}_0,\widehat{f}_1\right\rangle_{L^2}} \right| \ \lesssim\ \lambda^\alpha,\]
where we used~\eqref{norms} and~\eqref{assumptionsR}, and the fact that for $\lambda$ sufficiently small, $\left\langle \widehat{f}_0,\widehat{f}_0\right\rangle_{L^2} \geq c>0$, where $c$ is independent of $\lambda$. Lemma~\ref{lem:technical} is proved.
\end{proof}

\noindent{\it Proof that $\mathcal{G}:\mathcal{S}\to\mathcal{S}$ is a contraction mapping}: 
The result will follow from the two following claims, proved below:
\begin{Claim} \label{clm:G0small} There exists $C_H=C\left(\theta_0,A,C_R,\big\| \widehat{f}_0\big\|_{\mathcal Z_2}\right)>0$ such that
$
\left\| \mathcal G(0) \right\|_{\mathcal{Z}_2} \leq \frac{1}{2} C_H\lambda^\alpha\ .
$
\end{Claim}
\begin{Claim} \label{clm:Gcontractant} There exists $\lambda_0>0$ such that if $0\leq\lambda<\lambda_0$, then
$\left\| \mathcal G(\widehat{f}_1) - \mathcal G(\widehat{f}_2) \right\|_{\mathcal{Z}_2} \leq \frac{1}{2} \left\| \widehat{f}_1 - \widehat{f}_2 \right\|_{\mathcal{Z}_2}$.
\end{Claim}
It follows that $\mathcal{G}$ maps $\mathcal{S}\equiv \left\{ f\in \mathcal Z_2\ :
 \ 
\big\Vert \widehat{f} \big\Vert_{\mathcal Z_2}\leq C_H \lambda^\alpha \right\}$  into $\mathcal{S}$ since
\[
\big\Vert \mathcal{G}(f) \big\Vert_{\mathcal{Z}_2}\ \leq \ \big\Vert \mathcal{G}(f)-\mathcal{G}(0) \big\Vert_{\mathcal{Z}_2}+\big\Vert \mathcal{G}(0) \big\Vert_{\mathcal{Z}_2}\
\le \ \frac{1}{2}\big\Vert f-0 \big\Vert_{\mathcal Z_2}+\frac{1}{2}C_H\lambda^\alpha \ \leq\ C_H\lambda^\alpha\ .\]
Therefore, by Claim~\ref{clm:Gcontractant}, $\mathcal{G}:\mathcal{S}\to\mathcal{S}$ is a contraction mapping.

\begin{proof}[Proof of Claim~\ref{clm:G0small}:]
By definition, one has
\[
\mathcal{G}(0)(\xi)\ \equiv \ \frac{\chi\left(|\xi|<\lambda^{-\beta}\right)}{4\pi^2A\xi^2 + \theta_0^2}\left( -\left\langle \widehat{f}_0,R\left(\widehat{f}_0\right)\right\rangle_{L^2}\widehat{f}_0(\xi) + R\left(\widehat{f}_0\right)(\xi) \right).
 \]
 It follows, from our assumptions~\eqref{norms} on functional spaces $(\mathcal Z_1,\mathcal Z_2)$:
\begin{align}
\left\| \mathcal G(0) \right\|_{\mathcal Z_2} 
& \lesssim \left\| \frac{\chi\left(|\xi|<\lambda^{-\beta}\right)}{4\pi^2A\xi^2 + \theta_0^2} \right\|_{L^{\infty}} \left\langle \widehat{f}_0,R\left(\widehat{f}_0\right)\right\rangle_{L^2} \big\| \widehat{f}_0 \big\|_{\mathcal Z_2} + \left\| \frac{\chi\left(|\xi|<\lambda^{-\beta}\right) (1+\xi^2)}{4\pi^2A\xi^2 + \theta_0^2} \right\|_{L^\infty} \left\| \frac{ R\left(\widehat{f}_0\right) }{1+|\cdot|^2}\right\|_{\mathcal Z_2} \nn \\
&\lesssim \ \left\| R\left(\widehat{f}_0\right)\right\|_{\mathcal Z_1} \big\| \widehat{f}_0\big\|_{\mathcal Z_2}^2 + \left\| R\left(\widehat{f}_0\right)\right\|_{\mathcal Z_1} . \label{est:S0}
\end{align}
Claim~\ref{clm:G0small} is now obvious, using the smallness hypothesis on the operator $R$,~\eqref{assumptionsR} .
\end{proof}

\begin{proof}[Proof of Claim~\ref{clm:Gcontractant}:]
Let us decompose the mapping $\mathcal G$ as follows:
\begin{align*}
 \mathcal G(\widehat{f}_1 - \widehat{f}_2) &= \frac{\chi\left(|\xi|<\lambda^{-\beta}\right)}{4\pi^2A\xi^2 + \theta_0^2} \Bigg( -\frac{\left\langle \widehat{f}_0,R\left(\widehat{f}_0+\widehat{f}_1\right)\right\rangle_{L^2}}{\left\langle \widehat{f}_0,\widehat{f}_0\right\rangle_{L^2} + \left\langle \widehat{f}_0,\widehat{f}_1\right\rangle_{L^2}}\left(\widehat{f}_0+\widehat{f}_1\right)(\xi)
\\
&\qquad\qquad\qquad\qquad+\frac{\left\langle \widehat{f}_0,R\left(\widehat{f}_0+\widehat{f}_2\right)\right\rangle_{L^2}}{\left\langle \widehat{f}_0,\widehat{f}_0\right\rangle_{L^2} + \left\langle \widehat{f}_0,\widehat{f}_2\right\rangle_{L^2}}\left(\widehat{f}_0+\widehat{f}_2\right)(\xi) 
\ +\ R\left(\widehat{f}_1-\widehat{f}_2\right)(\xi)\Bigg)  \nonumber\\
& \equiv \ \frac{S_1[\widehat{f}_1](\xi)-S_1[\widehat{f}_2](\xi)}{4\pi^2A\xi^2 + \theta_0^2}+ \frac{\chi\left(|\xi|<\lambda^{-\beta}\right)\ R\left(\widehat{f}_1-\widehat{f}_2\right)(\xi)}{4\pi^2A\xi^2 + \theta_0^2} .
\end{align*}
The following estimate follows from our assumptions~\eqref{norms} on the spaces $(\mathcal Z_1,\mathcal Z_2)$:
\begin{align}
\left\| \mathcal G(\widehat{f}_1 - \widehat{f}_2) \right\|_{\mathcal Z_2} & \leq \ \left\| \frac{S_1[\widehat{f}_1](\xi)-S_1[\widehat{f}_2](\xi)}{4\pi^2A\xi^2 + \theta_0^2}\right\|_{\mathcal Z_2} + \left\| \frac{\chi\left(|\xi|<\lambda^{-\beta}\right)\ R\left(\widehat{f}_1-\widehat{f}_2\right)(\xi)}{4\pi^2A\xi^2 + \theta_0^2} \right\|_{\mathcal Z_2} \nn \\
& \lesssim \left\| \frac{\chi\left(|\xi|<\lambda^{-\beta}\right)}{4\pi^2A\xi^2 + \theta_0^2} \right\|_{L^{\infty}} \Big\| S_1[\widehat{f}_1]-S_1[\widehat{f}_2] \Big\|_{\mathcal Z_2} + \left\| \frac{\chi\left(|\xi|<\lambda^{-\beta}\right) (1+\xi^2)}{4\pi^2A\xi^2 + \theta_0^2} \right\|_{L^\infty} \left\| \frac{ R\left(\widehat{f}_1-\widehat{f}_2\right) }{1+|\cdot|^2}\right\|_{\mathcal Z_2} \nn \\
& \lesssim \ \left\| S_1[\widehat{f}_1]-S_1[\widehat{f}_2] \right\|_{\mathcal Z_2} + \left\| R\left(\widehat{f}_1-\widehat{f}_2\right)\right\|_{\mathcal Z_1} \ . \label{est:S}
\end{align}

The second term in~\eqref{est:S} is estimated using assumptions $R_{\alpha,\beta}$,~\eqref{assumptionsR}:
\begin{equation}
 \left\| R\left(\widehat{f}_1-\widehat{f}_2\right) \right\|_{\mathcal{Z}_1} \ \leq\ C_R \lambda^\alpha\left\| \widehat{f}_1-\widehat{f}_2\right\|_{\mathcal{Z}_2} . \label{est:diffS2}
\end{equation}

Let us now turn to the first term in~\eqref{est:S}.
\begin{align}
S_1[\widehat{f}_1]-S_1[\widehat{f}_2] &= - \frac{\left\langle \widehat{f}_0,R\left(\widehat{f}_0+\widehat{f}_1\right)\right\rangle_{L^2}}{\left\langle \widehat{f}_0,\widehat{f}_0\right\rangle_{L^2} + \left\langle \widehat{f}_0,\widehat{f}_1\right\rangle_{L^2}} \left(\widehat{f}_0+\widehat{f}_1\right) +\frac{\left\langle \widehat{f}_0,R\left(\widehat{f}_0+\widehat{f}_2\right)\right\rangle_{L^2}}{\left\langle \widehat{f}_0,\widehat{f}_0\right\rangle_{L^2} + \left\langle \widehat{f}_0,\widehat{f}_2\right\rangle_{L^2}} \left(\widehat{f}_0+\widehat{f}_2\right)\nonumber \\
&=-\ \frac{\left\langle \widehat{f}_0,R\left(\widehat{f}_1-\widehat{f}_2\right)\right\rangle_{L^2} \left(\widehat{f}_0+\widehat{f}_1\right) }{\left\langle \widehat{f}_0,\widehat{f}_0\right\rangle_{L^2} + \left\langle \widehat{f}_0,\widehat{f}_1\right\rangle_{L^2}}
\ - \ \frac{\left\langle \widehat{f}_0,R\left(\widehat{f}_0+\widehat{f}_2\right)\right\rangle_{L^2} \left(\widehat{f}_1-\widehat{f}_2\right)}{\left\langle \widehat{f}_0,\widehat{f}_0\right\rangle_{L^2} + \left\langle \widehat{f}_0,\widehat{f}_1\right\rangle_{L^2}} \nn \\
& \quad - \ \left\langle \widehat{f}_0,R\left(\widehat{f}_0+\widehat{f}_2\right)\right\rangle_{L^2} \left(\widehat{f}_0+\widehat{f}_2\right)\left( \frac{1}{\left\langle \widehat{f}_0,\widehat{f}_0\right\rangle_{L^2} + \left\langle \widehat{f}_0,\widehat{f}_1\right\rangle_{L^2}}-\frac{1}{\left\langle \widehat{f}_0,\widehat{f}_0\right\rangle_{L^2} + \left\langle \widehat{f}_0,\widehat{f}_2\right\rangle_{L^2}}\right) \nn \\
&= I+II+III. \label{eq:decompS2}
\end{align}

The result is a consequence of the following estimates:
\[
\left\langle \widehat{f}_0,g\right\rangle_{L^2}  \leq  C \big\| \widehat{f}_0\big\|_{\mathcal Z_1} \left\| g \right\|_{\mathcal Z_2}  \leq  C_1\left\| g \right\|_{\mathcal Z_2} ,\ \ 
\left\langle \widehat{f}_0,R\left(g\right)\right\rangle_{L^2}\leq  C \big\| \widehat{f}_0\big\|_{\mathcal Z_2} \left\| R\left(g\right) \right\|_{\mathcal Z_1}  \leq  C_2\lambda^\alpha \left\| g \right\|_{\mathcal Z_2}.
\]
with $C_1=C\big(\big\| \widehat{f}_0\big\|_{\mathcal Z_1}\big)$ and $C_2=C_2\big(\big\| \widehat{f}_0\big\|_{\mathcal Z_2},C_R\big)$.
Using the above, one checks that for sufficiently small $\lambda$,
\begin{align*}
\left\| I \right\|_{\mathcal Z_2}\ &\lesssim \ C_2\lambda^\alpha \left\| \widehat{f}_1-\widehat{f}_2 \right\|_{\mathcal Z_2}(\big\| \widehat{f}_0\big\|_{\mathcal Z_2}+C_H \lambda^\alpha), \\ 
\left\| II \right\|_{\mathcal Z_2}\ &\lesssim \ C_2\lambda^\alpha(\big\| \widehat{f}_0\big\|_{\mathcal Z_2}+C_H \lambda^\alpha) \left\| \widehat{f}_1-\widehat{f}_2 \right\|_{\mathcal Z_2},\\ 
\left\| III \right\|_{\mathcal Z_2}\ &\lesssim \  C_1 \left\| \widehat{f}_1-\widehat{f}_2 \right\|_{\mathcal Z_2} C_2 \lambda^\alpha(\big\| \widehat{f}_0\big\|_{\mathcal Z_2}+C_H \lambda^\alpha)^2.
\end{align*}
Thus if $ C_1\lambda^\alpha<1/2$, one has
\begin{equation}\label{est:diffS1}
\left\| S_1[\widehat{f}_1]-S_1[\widehat{f}_2] \right\|_{\mathcal Z_2} \ \leq \ \left\| I \right\|_{\mathcal Z_2}+\left\| II \right\|_{\mathcal Z_2}+\left\| III \right\|_{\mathcal Z_2} \ \lesssim\ \lambda^\alpha \left\| \widehat{f}_1-\widehat{f}_2 \right\|_{\mathcal Z_2} .
\end{equation}
Plugging~\eqref{est:diffS2} and~\eqref{est:diffS1} into~\eqref{est:S}, it follows the existence of a constant, $C_0>0$, such that $\| \mathcal G(\widehat{f}_1) - \mathcal G(\widehat{f}_2) \|_{\mathcal{Z}_2} \leq C_0\lambda^\alpha \| \widehat{f}_1 - \widehat{f}_2 \|_{\mathcal{Z}_2}$. Thus for $0<\lambda<\lambda_0\leq C_0^{-\frac1\alpha}$, we obtain a contraction and Claim~\ref{clm:Gcontractant} is proved. 
\end{proof}

\section{Proof of Th'm~\ref{thm:Qzero}; 
 Edge bifurcations for $-\partial_x^2+\lambda V(x)$ }\label{sec:nonper}

In this section we prove Theorem~\ref{thm:Qzero},  the special case: $Q\equiv0$  of our main result, Theorem~\ref{thm:per_result}.
In this case we study the bifurcation of solutions to the eigenvalue problem
\begin{equation}\label{eq:orig_eqn_nonper}
\left(-\partial_x^2+\lambda V(x)\right)\psi^\lambda(x)\ =\ E^\lambda\psi^\lambda(x),\ \ \psi^\lambda\in L^2(\mathbb R)
\end{equation}
into the interval $(-\infty,0)$, the semi-infinite spectral gap of $H_0\equiv-\partial_x^2$, for $V$ a spatially localized potential, and $\lambda>0$ sufficiently small. 
 Here, the Floquet-Bloch eigenfunctions are exponentials. Hence, calculations are more straightforward and error bounds on the approximations are sharper.

\subsection{Near and far frequency decomposition}
Taking the Fourier transform of~\eqref{eq:orig_eqn_nonper} yields
\begin{equation} \label{eq:FT_orig_eqn_nonper}
\left(4\pi^{2}\xi^2-E^\lambda\right)\widehat{\psi^\lambda}(\xi)+\lambda\int_{\mathbb R}\widehat{V}\left(\xi-\zeta\right)\widehat{\psi^\lambda}(\zeta)\ d\zeta = 0.
\end{equation}
We shall study~\eqref{eq:FT_orig_eqn_nonper} via the equivalent system for the
\begin{align*}
 &\textrm{{\em near frequency components}: } \{\widehat{\psi^\lambda}(\xi): |\xi|<\lambda^{r}\ \}\ \textrm{  and} \\
& \textrm{ {\em far frequency components}: } \{\widehat{\psi^\lambda}(\xi): |\xi|\geq \lambda^{r}\ \}
\  \textrm{of}\  \psi^\lambda.
\end{align*}
 
Let $r$ be a parameter, chosen to satisfy: $0<r<1$. 
Recall the cutoff functions, $\chi$ and $\overline{\chi}$, introduced in~\eqref{chi-def} and 
$1={\vphantom{\overline{\chi}}\chi}_{_{\lambda^r}}(\xi)+{\overline{\chi}}_{_{\lambda^r}}(\xi)$.
 Multiplying~\eqref{eq:FT_orig_eqn_nonper} by this identity we get
\begin{multline*}
0 =  \left(4\pi^{2}|\xi|^2-E^\lambda\right)\left({\vphantom{\overline{\chi}}\chi}_{_{\lambda^r}}+\overline{\chi}_{_{\lambda^r}}\right)(\xi)\ \widehat{\psi^\lambda}(\xi)
\\
 +\lambda\int_{-\infty}^{\infty}\left({\vphantom{\overline{\chi}}\chi}_{_{\lambda^r}}+\overline{\chi}_{_{\lambda^r}}\right)(\xi)\ \widehat{V}(\xi-\zeta)\left({\vphantom{\overline{\chi}}\chi}_{_{\lambda^r}}+
 \overline{\chi}_{_{\lambda^r}}\right)(\zeta)\ \widehat{\psi^\lambda}(\zeta)d\zeta.
\end{multline*}
Introduce notation for the near- and far- frequency components of $\psi^\lambda$:
\begin{equation}\label{eqn:def-near-far_nonper} 
\widehat{\psi}_\nr(\xi)\ \equiv\ {\vphantom{\overline{\chi}}\chi}_{_{\lambda^r}}(\xi)\ \widehat{\psi^\lambda}(\xi) 
\quad \text{and} \quad \widehat{\psi}_\fr(\xi)\equiv \overline{\chi}_{_{\lambda^r}}(\xi)\ \widehat{\psi^\lambda}(\xi).\end{equation}
Then, the eigenvalue equation is equivalent to the following coupled system: 
\begin{align}\label{eq:near_nonper}
\left(4\pi^{2}|\xi|^2-E^\lambda\right)\widehat{\psi}_\nr(\xi)+\lambda{\vphantom{\overline{\chi}}\chi}_{\lambda^{r}}(\xi)\int_{-\infty}^{\infty}\widehat{V}(\xi-\zeta)\big( \widehat{\psi}_\nr(\zeta)+\widehat{\psi}_\fr(\zeta)\big)d\zeta&=0,\\
\label{eq:far_nonper}
\left(4\pi^{2}|\xi|^2-E^\lambda\right)\widehat{\psi}_\fr(\xi)+\lambda\overline{\chi}_{_{\lambda^r}}(\xi)\int_{-\infty}^{\infty}\widehat{V}(\xi-\zeta)\big( \widehat{\psi}_\nr(\zeta)+\widehat{\psi}_\fr(\zeta)\big)d\zeta&=0.
\end{align}

In what follows we shall set $E^\lambda=-\lambda^2\theta^2$, where $\theta=\theta(\lambda)$ is expected to be $\mathcal{O}(1)$ as $\lambda\downarrow0$. This anticipates that the bifurcating eigenvalue, $E^\lambda$, will be real, negative and $\mathcal{O}(\lambda^2)$. 

\subsection{Analysis of the far frequency components}\label{subsec:far_nonper}

We view~\eqref{eq:far_nonper} as an equation for $\widehat{\psi}_\fr$, depending on ``parameters'' $(\widehat{\psi}_\nr;\lambda) $. The following proposition studies the mapping $(\widehat{\psi}_\nr;\lambda) \mapsto \widehat{\psi}_\fr$.

\begin{Proposition}\label{prop:far_nonper_intermsof_near_nonper}
Let $\widehat{\psi}_\nr\in L^1$. There exists $\lambda_0>0$, such that for $0<\lambda<\lambda_0$, the following holds.
Set
$ E^\lambda\ \equiv\ -\lambda^2\theta^2,\ \ \textrm{with}\ \ |\theta|\leq \pi \lambda^{r-1},\ \ r\in(0,1)\ .
$
There is a unique solution $\widehat{\psi}_\fr\ =\ \widehat{\psi}_\fr\left[\widehat{\psi}_\nr;\lambda\right] $ of the far frequency equation~\eqref{eq:far_nonper}. The mapping
$
(\widehat{\psi}_\nr;\lambda) \mapsto \widehat{\psi}_\fr\left[\widehat{\psi}_\nr;\lambda\right]
$ 
maps $L^1(\mathbb{R}) \times \mathbb{R}$ to $ L^1(\mathbb{R})$ and satisfies the bound:
\begin{equation}
\big\| \widehat{\psi}_\fr\big\|_{L^1} \leq C(\|\widehat V\|_{L^\infty})\ \lambda^{1-r }\ \| \widehat{\psi}_\nr\|_{L^1}.
\label{1stbound}\end{equation}
\end{Proposition}
\begin{proof}
We seek to solve~\eqref{eq:far_nonper} for $\widehat{\psi}_\fr$ as a functional of 
 $\widehat{\psi}_\nr$. 
Since $|\xi|\geq \lambda^r$, with $0<r<1$, and $|\theta|\le\pi\lambda^{r-1}$, we have $\left|4\pi^{2}\xi^2-E^\lambda\right|=\left|4\pi^{2}\xi^2+\lambda^{2}\theta^{2}\right|\ge 3\pi^2\lambda^{2r}$, which is bounded away from zero for any fixed $\lambda>0$.
Dividing~\eqref{eq:far_nonper} by $4\pi^{2}\xi^2-E^\lambda=4\pi^{2}\xi^2+\lambda^{2}\theta^{2}$, we obtain that 
\eqref{eq:far_nonper} is equivalent to the equation:
 \begin{equation}
\left(I+\mathcal{\widehat{T}}_{\lambda}\right)\widehat{\psi}_\fr(\xi)=-\left(\mathcal{\widehat{T}}_{\lambda}\widehat{\psi}_\nr\right)(\xi)\ ,
\label{ieqnfar}\end{equation}
where
\[
 \left(\mathcal{\widehat{T}}_{\lambda}\ \widehat{g}\right)(\xi)\ \equiv
 \int_{\zeta}\mathcal{K}_{\lambda}(\xi,\zeta)\ \widehat{g}(\zeta)\ d\zeta
\qquad \text{ and } \qquad
 \mathcal{K}_{\lambda}(\xi,\zeta)\ \equiv\ \lambda\frac{\overline{\chi}_{\lambda^{r}}(\xi)}{4\pi^{2}\xi^2+\lambda^{2}\theta^{2}}\widehat{V}(\xi-\zeta) \ .
\]

We next show that the integral operator $\mathcal{\widehat{T}}_{\lambda}$, viewed as an operator from $L^1$ to $L^1$ has small norm, for $\lambda$ small. This implies the invertibility of 
 $I+\mathcal{\widehat{T}}_{\lambda}$ and the assertions of Proposition~\ref{prop:far_nonper_intermsof_near_nonper}.

Let $ \widehat g\in L^1$. One has
\[
\big\|\mathcal{\widehat{T}}_{\lambda} \widehat g\big\|_{L^1}
\ \leq \ C\ \lambda \int_{|\xi|\ge\lambda^r}\ \frac{1}{4\pi^{2}\xi^2+\lambda^{2}\theta^{2}} \ d\xi \ \ \big\| \widehat{V}\big\|_{L^\infty}\big\| \widehat g \big\|_{L^1}\ 
\lesssim\ \lambda^{1-r}\ \big\| \widehat{V}\big\|_{L^\infty}\big\| \widehat g\big\|_{L^1}.
\]
Thus  $\mathcal{\widehat{T}}_{\lambda}$ is bounded from $L^1$ to $L^1$ with norm bound:
$ \big\| \mathcal{\widehat{T}}_{\lambda}\big\|_{L^1\to L^1} \ \lesssim\ \lambda^{1- r}\ \big\| \widehat{V}\big\|_{L^\infty}$.

 For $r\in\left(0,1\right)$, $\big\Vert \mathcal{\widehat{T}}_{\lambda}\big\Vert _{L^{1}\rightarrow L^{1}}\to 0$ as $\lambda\to0$. Therefore $I+\mathcal{\widehat{T}}_{\lambda}$ is invertible, for $\lambda $ sufficiently small.
 Moreover, 
 \[
\left\| \widehat{\psi}_\fr\right\|_{L^1}=\left\|\big(I+\mathcal{\widehat{T}}_{\lambda}\big)^{-1}\big(\mathcal{\widehat{T}}_{\lambda}\widehat{\psi}_\nr\big) \right\|_{L^1}\leq\big\| \big(I+\mathcal{\widehat{T}}_{\lambda}\big)^{-1}\big\|_{L^1\to L^1}\big\| \mathcal{\widehat{T}}_{\lambda}\big\|_{L^1\to L^1}\big\| \widehat{\psi}_\nr \big\|_{L^1} ,\]
which implies the bound~\eqref{1stbound}. Proposition~\ref{prop:far_nonper_intermsof_near_nonper} is proved.
\end{proof}

\subsection{Analysis of the near frequency components}\label{sec:near_nonper}

Now that we have constructed $\widehat{\psi}_\fr$ as a functional of $\widehat{\psi}_\nr$ and $\lambda$ (Proposition~\ref{prop:far_nonper_intermsof_near_nonper}), it is possible to treat~\eqref{eq:near_nonper}, for $\lambda$ small, as a {\it closed equation} for a {\it low frequency projected eigenstate}, $\widehat{\psi}_\nr(\xi;\lambda)$, and corresponding eigenvalue $E^\lambda$.
Substitution of $\widehat{\psi}_\fr= \widehat{\psi}_\fr[\widehat{\psi}_\nr,\lambda]$ into~\eqref{eq:near_nonper} yields:
\begin{equation}\label{eq:near_nonper_substituted}
\left(4\pi^{2}|\xi|^2-E^\lambda\right)\widehat{\psi}_\nr(\xi)+\lambda{\vphantom{\overline{\chi}}\chi}_{\lambda^{r}}(\xi)\int_{\zeta}\widehat{V}(\xi-\zeta)\widehat{\psi}_\nr(\zeta)d\zeta +\lambda{\vphantom{\overline{\chi}}\chi}_{\lambda^{r}}(\xi)\widehat R(\xi) = 0,
\end{equation}
where $\widehat R$ is defined by 
\begin{equation}
\label{Rdef}
\widehat R(\xi)\ \equiv\ \int_{\zeta} \widehat{V}(\xi-\zeta)\
 \widehat{\psi}_\fr[\widehat{\psi}_\nr,\lambda](\zeta)\ d\zeta\ .
 \end{equation}
Recall that $ \widehat{\psi}_\fr[\widehat{\psi}_\nr,\lambda]$ is in $L^1$, and of size $\mathcal{O}\left(\lambda^{1-r}\ \|\widehat{\psi}_\nr\|_{L^1}\right)$
 by Proposition~\ref{prop:far_nonper_intermsof_near_nonper}. 
 \

Our next goal is, via appropriate expansion, reorganization and scaling, to re-express~\eqref{eq:near_nonper_substituted} as a simple leading order asymptotic equation plus controllable corrections. 
 The terms in~\eqref{eq:near_nonper_substituted} are supported in the near (low) frequency regime.
 Note that for $|\xi|<\lambda^r$ and $|\zeta|<\lambda^r$ we have $|\xi-\zeta| \leq |\xi| + |\zeta| < 2\lambda^r$. Taylor expansion of $\widehat{V}(\xi-\zeta)$ gives
$
\widehat{V}(\xi-\zeta) = \widehat{V}(0) + (\xi - \zeta) \widehat{V}'(\eta),
$
for some $\eta=\eta(\zeta,\xi)$ such that $|\eta|<2\lambda^r$. Using this expansion in the second term of~\eqref{eq:near_nonper_substituted} yields 
\begin{equation}\label{eq:near_nonper_2}
\left(4\pi^{2}|\xi|^2-E^\lambda\right)\widehat{\psi}_\nr(\xi)+\lambda{\vphantom{\overline{\chi}}\chi}_{\lambda^{r}}(\xi)\widehat{V}(0)\int_{\zeta} \widehat{\psi}_\nr(\zeta)d\zeta =  \lambda{\vphantom{\overline{\chi}}\chi}_{\lambda^{r}}(\xi)\mathcal{R}\left[\widehat{\psi}_\nr;\lambda\right](\xi),
\end{equation}
where $\mathcal{R}\left[\widehat{\psi}_\nr;\lambda\right]\ \equiv\ \mathcal{R}_1\ +\ \mathcal{R}_2$, with
\begin{align*}
&\mathcal{R}_1(\xi)\ \equiv\ -\widehat{R}(\xi) \ = \ - \int_{\zeta} \widehat{V}(\xi-\zeta)\
 \widehat{\psi}_\fr[\widehat{\psi}_\nr,\lambda](\zeta)\ d\zeta \ ,\ \ 
\mathcal{R}_2(\xi)\ \equiv\ -\int_{\zeta} (\xi - \zeta)\widehat{V}'(\eta)  \widehat{\psi}_\nr(\zeta)\ d\zeta\ .
\end{align*}

We now introduce the scaled near-frequency Fourier component, $ \widehat{\Phi}_{\lambda}$, by
\begin{equation}
\widehat{\psi}_\nr(\xi;\lambda)=\frac{1}{\lambda} \widehat{\Phi}_{\lambda}\left(\frac{\xi}{\lambda}\right),
\label{xi-theta-defs}
\end{equation}
Note that
\begin{equation}
\left\| \widehat{\psi}_\nr(\cdot;\lambda) \right\|_{L^1} \ = \ \left\| \frac{1}{\lambda} \widehat{\Phi}_{\lambda}\left(\frac{\cdot}{\lambda}\right) \right\|_{L^1} \ = \ \left\| \widehat{\Phi}_{\lambda}\right\|_{L^1} . 
\label{scaleL1}
\end{equation}
We also denote $E^\lambda=-\lambda^2\theta^2$, and restrict to $\theta=\theta(\lambda)$ satisfying the constraint in the hypotheses of Proposition~\ref{prop:far_nonper_intermsof_near_nonper}. 
Substitution of~\eqref{xi-theta-defs} into~\eqref{eq:near_nonper_2}, defining $\xi' = \lambda \xi$
and dividing by $\lambda$ yields the following rescaled near-frequency equation:
\begin{equation} \label{eq:near_nonper_rescaled}
\left(4\pi^{2}|\xi'|^2 + \theta^2\right) \widehat{\Phi}_{\lambda}(\xi')+{\vphantom{\overline{\chi}}\chi}_{\lambda^{r-1}}(\xi')\widehat{V}(0) \int_{\zeta'} \widehat{\Phi}_{\lambda}(\zeta')d{\zeta'} = {\vphantom{\overline{\chi}}\chi}_{\lambda^{r-1}}(\xi') \mathcal{R}'\big(\widehat{\Phi}_{\lambda}\,\big)(\xi')
\end{equation}
where $\mathcal{R}'\big(\widehat{\Phi}_{\lambda}\,\big)(\xi')\equiv \mathcal{R}\left[\widehat{\psi}_\nr;\lambda\right](\lambda\xi')\equiv\mathcal{R}_1'(\xi')+\mathcal{R}_2'(\xi')$, with
\begin{align}
&\mathcal{R}_1'(\xi')\ \equiv\ - \int_{\zeta} \widehat{V}(\lambda\xi'-\zeta)\
 \widehat{\psi}_\fr[\widehat{\psi}_\nr,\lambda](\zeta)\ d\zeta \ , \label{R1'}\\
&\mathcal{R}_2'(\xi')\ \equiv\ -\int_{\zeta} (\lambda\xi' - \zeta)\widehat{V}'(\eta)  \widehat{\psi}_\nr(\zeta)\ d\zeta\ \ = \ -\lambda \int_{\zeta} (\xi' - \zeta')  \widehat{V}'(\eta)\widehat{\Phi}_{\lambda}(\zeta')d\zeta' \ .
\label{R2'}\end{align}

Equation~\eqref{eq:near_nonper_rescaled} is in the form of the class of equations 
to which Lemma~\ref{lem:technical} applies. We shall use Lemma~\ref{lem:technical} to obtain a non-trivial eigenpair solution $(\ \widehat\Phi_\lambda,\theta(\lambda)\ )$ of~\eqref{eq:near_nonper_rescaled}. Toward verification of the hypotheses of Lemma~\ref{lem:technical}, we next bound the right hand side of~\eqref{eq:near_nonper_rescaled}.

\begin{Proposition} \label{prop:R_nonper_bound}
Let $V$ be such that 
$\big\Vert\widehat{V}\big\Vert_{W^{1,\infty}}  \ \equiv \ \big\Vert\widehat{V}\big\Vert_{L^{\infty}} + \big\Vert\widehat{V}'\big\Vert_{L^{\infty}} < \infty\ .$ 
Then, the right hand side of the rescaled near-frequency equation~\eqref{eq:near_nonper_rescaled} satisfies the bound
\begin{equation}
\left\|{\vphantom{\overline{\chi}}\chi}_{\lambda^{r-1}}(\xi)\mathcal{R}'\big(\widehat{\Phi}_{\lambda}\,\big)\right\|_{L^{\infty}}\ \leq\ 
 C\big(\big\Vert\widehat{V}\big\Vert_{W^{1,\infty}}\big)\ \left( \lambda^{1-r} + \lambda^r \right) \ \left\|\widehat{\Phi}_{\lambda}\right\|_{L^1}.
\end{equation}
\end{Proposition}
\begin{proof}
We proceed by estimating each term individually. 

\noindent{\em Estimation of $\mathcal{R}'_1(\xi')$, given by~\eqref{R1'}:} 
By Proposition~\ref{prop:far_nonper_intermsof_near_nonper}, one has
\begin{align} 
&\big\| \widehat{\psi}_\fr[\widehat{\psi}_\nr,\lambda] \big\|_{L^1(\mathbb{R})}
\leq \ C\big(\|\widehat V\|_{L^\infty}\big)\ \lambda^{1-r}\ \big\|\widehat{\psi}_\nr\big\|_{L^1(\mathbb{R})}\ .\label{error-bound1}
\end{align}
Plugging~\eqref{error-bound1} into~\eqref{R1'}, and making use of~\eqref{scaleL1}, we have 
\begin{align*}
\left\| \mathcal{R}'_1 \right\|_{L^\infty}\ &= \ 
\left\|  \int_\RR\widehat{V}(\lambda\xi'- \zeta) \widehat{\psi}_\fr[\widehat{\psi}_\nr,\lambda](\zeta)\ d\zeta \right\|_{L^\infty_{\xi'}} \\
&
\leq \ \big\|\widehat V\big\|_{L^\infty}\ \big\|\ \widehat{\psi}_\fr[\widehat{\psi}_\nr,\lambda]\ \big\|_{L^1} \
 \leq \ C(\|\widehat V\|_{L^\infty})\ \lambda^{1-r}\ \big\|\widehat{\Phi}_{\lambda} \big\|_{L^1}\ .
\end{align*}

\noindent {\em Estimation of $\mathcal{R}'_2(\xi')$, given by~\eqref{R2'}:}\
We have the bound
\[
\big\|  {\vphantom{\overline{\chi}}\chi}_{\lambda^{r-1}}(\xi')\mathcal{R}'_2 \big\|_{L^\infty} =  \big\| {\vphantom{\overline{\chi}}\chi}_{\lambda^{r-1}}(\xi')\int_{\zeta'}\lambda(\xi'-\zeta') \widehat{V}'(\eta)\widehat{\Phi}_{\lambda}(\zeta')d\zeta'\big\|_{L^\infty_{\xi'}}\ \leq\ 2 \lambda^r \big\|\widehat V'\big\|_{L^\infty} \big\| \widehat{\Phi}_{\lambda}\big\|_{L^1},
\]
using that $\widehat{\Phi}_{\lambda}(\zeta')= {\vphantom{\overline{\chi}}\chi}_{\lambda^{r-1}}(\zeta')\widehat{\Phi}_{\lambda}(\zeta')$, so that $|\xi'-\zeta'|\leq 2\lambda^{r-1}$.
 Proposition~\ref{prop:R_nonper_bound} is proved.
\end{proof}

\begin{Remark}
We expect that by using a higher order Taylor approximation of $\widehat{V}(\xi-\zeta)$ in the second term of equation~\eqref{eq:near_nonper_substituted}, it should be possible to obtain a variant of Proposition~\ref{prop:R_nonper_bound} with a bound which is higher order in $\lambda$. This would require a higher order  variant of Lemma~\ref{lem:technical}.
\end{Remark}

\subsection{Completion of the proof}
We now prove Theorem~\ref{thm:Qzero} by an application of Lemma~\ref{lem:technical} to equation~\eqref{eq:near_nonper_rescaled}, using the remainder estimate of Proposition~\ref{prop:R_nonper_bound}.

\begin{proof}[Proof of Theorem~\ref{thm:Qzero}] 
We construct $\psi^\lambda$, solution to~\eqref{eq:FT_orig_eqn_nonper} as
$ \widehat{\psi^\lambda} \ = \ \widehat{\psi}_\fr \ + \ \widehat{\psi}_\nr, $
where $\widehat{\psi}_\fr,\widehat{\psi}_\nr$ satisfy~\eqref{eq:near_nonper}--\eqref{eq:far_nonper}. The far-frequency component, $\widehat{\psi}_\fr$, is uniquely determined by $\widehat{\psi}_\nr$ and $\lambda$ sufficiently small; see Proposition~\ref{prop:far_nonper_intermsof_near_nonper}. Now set $\widehat{\psi}_\nr(\xi) \ \equiv \ \frac{1}{\lambda} \widehat{\Phi}_{\lambda}\left(\frac{\xi}{\lambda}\right)$. Since $\widehat{V}\in W^{1,\infty}$, Proposition~\ref{prop:R_nonper_bound} implies that the rescaled near-frequency equation~\eqref{eq:near_nonper_rescaled} can be written as
\begin{equation}\label{eq:near_nonper_rescaled_2}
\left(4\pi^{2}|\xi'|^2 + \theta^2\right) \widehat{\Phi}_{\lambda}(\xi')+{\vphantom{\overline{\chi}}\chi}_{\lambda^{r-1}}(\xi')\widehat{V}(0) \int_{\zeta'} \widehat{\Phi}_{\lambda}(\zeta')d{\zeta'} = {\vphantom{\overline{\chi}}\chi}_{\lambda^{r-1}}(\xi')\mathcal R(\widehat{\Phi}_{\lambda})(\xi'),
\end{equation}
with $
\left\| \mathcal R(u) \right\|_{L^{\infty}} \leq \ C \ \lambda^\alpha\ \| u \|_{L^1},
$
where $\alpha=\min(1-r,r)$ and $
C= C(\|\widehat{V}\|_{W^{1,\infty}})$. From now on, we set \[r=1/2=\alpha\]
 as this yields optimal estimates.
 Applying Lemma~\ref{lem:technical} to~\eqref{eq:near_nonper_rescaled_2} with 
$A = 1$, $-B = \widehat{V}(0)=\int_{\mathbb R}V$ (assumed to be negative),
 we deduce that there exists a solution $\left( \theta^2, \widehat{\Phi}_{\lambda} \right)$ of the rescaled near-frequency equation~\eqref{eq:near_nonper_rescaled_2}, satisfying
\begin{equation}\label{eq:bounds_rescaled_near}
\|\widehat{\Phi}_{\lambda} - \widehat{f}_0\|_{L^1}\ \leq\ C\ \lambda^{\frac12} \qquad \text{ and } \qquad |\theta^2 - \theta_0^2|\ \leq\ C\ \lambda^{\frac12}\ . 
\end{equation} 
Here $\left( \theta_0^2(\lambda), \widehat{f}_0 \right)$ is the unique (normalized) solution of the homogeneous equation 
\[
\widehat{\mathcal{L}} _{0,\lambda}\left( \theta_0,\widehat{f}_0 \right) = (4\pi^2 \xi^2 + \theta^2) \widehat{f}_0 + \chi\left(|\xi|<\lambda^{-\frac{1}{2}}\right) \widehat{V}(0) \int_{\mathbb R}\chi\left(|\eta|<\lambda^{-\frac{1}{2}}\right)\widehat{f}_0(\eta)d\eta = 0,
\]
as described in Lemma~\ref{lem:homogeneous}. 
Thus $\widehat{\psi}_\nr(\xi) \ = \ \frac{1}{\lambda} \widehat{\Phi}_{\lambda}\left(\frac{\xi}{\lambda}\right)$ and $E^\lambda\ =\ -\lambda^2\theta^2(\lambda)$ are well-defined. 

In conclusion, the eigenpair solution to~\eqref{eq:FT_orig_eqn_nonper} ({\em i.e.}~\eqref{eq:nonper}) , $(E^\lambda,\psi^\lambda)$, is uniquely determined by 
\[ E^\lambda\ \equiv\ -\lambda^2\theta^2(\lambda),\quad \text{ and } \quad \psi^\lambda\equiv\mathcal F^{-1}(\widehat{\psi}_\nr+\widehat{\psi}_\fr).\]

Estimate~\eqref{Elambda}, the small $\lambda$ expansion of the eigenvalue $E^\lambda$, follows from~\eqref{eq:bounds_rescaled_near}. The approximation,~\eqref{psi-lambda}, of the corresponding eigenstate, $\psi^\lambda=\psi_\nr+\psi_\fr$, is obtained as follows. First, by~\eqref{eq:bounds_rescaled_near} we have
 \begin{equation}\label{est:A}
 \left\| \widehat{\psi}_\nr(\eta) \ - \ \lambda \dfrac{{\vphantom{\overline{\chi}}\chi}_{\lambda^{1/2}}(\eta)}{4\pi^2|\eta|^2+\lambda^2\theta_0^2} \right\|_{L^1} \ = \ \left\|\widehat{\Phi}_{\lambda} - \widehat{f}_0 \right\|_{L^1}\ \lesssim\ \lambda^{1/2}. 
 \end{equation}
The high frequency components are small, as is seen from the following calculation:
 \begin{equation}\label{est:B}
 \left\|\lambda\dfrac{\overline{\chi}_{\lambda^{1/2}}(\eta) }{4\pi^2|\eta|^2+\lambda^2\widehat V(0)^2} \right\|_{L^1}\ \leq\ \lambda \int_{|\eta|\geq \lambda^{1/2}} \dfrac{d\eta }{4\pi^2|\eta|^2} 
\lesssim\  \lambda^{1/2}.
 \end{equation}
 Finally, from Proposition~\ref{prop:far_nonper_intermsof_near_nonper}, one has (with $r=1/2$) 
 \begin{equation}\label{est:C}
 \big\| \widehat{\psi}_\fr\big\|_{L^1}\ \leq\ 
C\ (\|\widehat V\|_{L^\infty})\ \lambda^{1/2}\ \| \widehat{\psi}_\nr\|_{L^1}, 
 \end{equation}
 and $\big \| \widehat{\psi}_\nr\big\|_{L^1}= \big\|\widehat{\Phi}_{\lambda} \big\|_{L^1} \to\big\| \widehat{f}_0 \big\|_{L^1}$ (as $\lambda\to0$). 
Altogether,~\eqref{est:A},~\eqref{est:B} and~\eqref{est:C} yield
\[
\left\| \psi^\lambda \ - \ \mathcal F^{-1}\left\{\lambda \dfrac{1}{4\pi^2|\cdot|^2+\lambda^2\theta_0^2}\right\} \right\|_{L^\infty}\ \leq \ \left\| \widehat{\psi^\lambda} \ - \ \lambda \dfrac{1}{4\pi^2|\cdot|^2+\lambda^2\theta_0^2} \right\|_{L^1}\ \lesssim \ \lambda^{1/2}. 
\]
Note, by residue computation, that
$\mathcal F^{-1}\left\{ (4\pi^2|\cdot|^2+\lambda^2\theta_0^2)^{-1}\right\}
  =\ \frac12(\lambda\theta_0)^{-1}\exp(-\lambda\theta_0|x|)$,
with $\theta_0=-\frac12\int_\RR V>0$ .
Thus estimate~\eqref{psi-lambda} holds. This completes the proof of Theorem~\ref{thm:Qzero}.
\end{proof}

\section{Proof of Th'm~\ref{thm:per_result}; Edge bifurcations of $-\partial_x^2+Q+\lambda V$}\label{sec:per}

Let $Q(x)$ denote a  non-trivial, continuous, $1-$periodic function, $Q(x+1)=Q(x)$.
In this section we study the bifurcation of solutions to the eigenvalue problem
\begin{equation} \label{eq:orig_eqn_per}
\left(-\partial_x^2 + Q(x) + \lambda V(x) \right) \psi^\lambda(x) = E^\lambda \psi^\lambda(x),\ \ \psi\in L^2(\mathbb R) 
\end{equation}
 into the spectral gaps of $-\partial_x^2+Q(x)$.
We proceed by the same general approach of Section~\ref{sec:nonper}. That is,
 by appropriate spectral localization, in this case by applying the Gelfand-Bloch transform, we reduce~\eqref{eq:orig_eqn_per} to an equivalent {\it near-frequency} eigenvalue problem supported on frequencies lying near a spectral band edge of $-\partial_x^2 + Q(x)$.

\subsection{Near and far frequency components}\label{sec:far_near_per}

We take the Gelfand-Bloch transform of~\eqref{eq:orig_eqn_per} and get
\begin{equation}\label{eq:GB_of_orig}
-\left(\partial_x + 2\pi ik\right)^2 \widetilde{\psi^\lambda}(x;k) + Q(x) \widetilde{\psi^\lambda}(x;k) + \lambda \left(V\psi^\lambda\right)^{\sim}(x;k) = E^\lambda \widetilde{\psi^\lambda}(x;k),
\end{equation}
where
\[
\left(V\psi^\lambda\right)^{\sim}(x;k) = \sum_{n\in\mathbb{Z}} e^{2\pi i nx} \left(V\psi^\lambda\right)^{\wedge}(k+n) = \sum_{n\in\mathbb{Z}} e^{2\pi i nx} \left(\widehat{V}\star\widehat{\psi^\lambda}\right)(k+n).
\]
Here, the quasi-momentum, $k$, varies over the interval $(-1/2, 1/2]$.

As in Section~\ref{sec:nonper}, we express $\psi$ in terms of its near- and far-frequency components around a band edge $E_{b_*}(k_*)$, for fixed $b_*$ and $k_*$:
\begin{equation}\label{eq:psi_decomp}
\psi^\lambda = \psi_\nr + \psi_\fr = \mathcal{T}^{-1}\left\{\widetilde{\psi}_\nr(k)p_{b_*}(x;k)\right\} + \mathcal{T}^{-1}\left\{\sum_{b=0}^\infty\widetilde{\psi}_{\fr,b}(k)p_b(x;k)\right\} ,
\end{equation}
where we define, for $b=0,1,\dots$\ :
 \begin{align*}
 \widetilde{\psi}_\nr(k) &\equiv \chi\left(|k-k_*|<\lambda^r\right) \mathcal{T}_{b_*}\{\psi^\lambda\}(k)= \chi\left(|k-k_*|<\lambda^r\right) \left\langle p_{b_*}(\cdot,k), \widetilde{\psi^\lambda}(\cdot,k) \right\rangle_{L^2([0,1])} ,\\
 \widetilde{\psi}_{\fr,b}(k) &\equiv \chi\left(|k-k_*| \geq \lambda^r\delta_{b_*,b}\right) \mathcal{T}_{b}\{\psi^\lambda\}(k)=\chi\left(|k-k_*| \geq \lambda^r\delta_{b_*,b}\right) \left\langle p_{b}(\cdot,k), \widetilde{\psi^\lambda}(\cdot,k) \right\rangle_{L^2([0,1])},
\end{align*}
where $\delta_{i,j}$ denotes Kronecker's delta function and $r$ a parameter chosen to satisfy $r>0$.
 Equivalently, one has
 \[ \psi^\lambda(x) \ = \ \int_{-1/2}^{1/2} \left( \widetilde{\psi}_\nr(k)u_{b_*}(x;k)+\sum_{b=0}^\infty \widetilde{\psi}_{\fr,b}(k)u_b(x;k) \right) \ dk.\]
Recall that $\{p_b(x;k)\}_{b \geq 0}$  form a complete orthonormal set in $L_{\rm per}^2([0,1])$, and satisfy 
\begin{align}
&\left( -\left(\partial_x + 2\pi ik\right)^2 + Q(x) \right) p_b(x;k) = E_b(k)p_b(x;k), \quad x\in[0,1],
p_b(x+1;k)\ =\ p_b(x;k)\ .\label{bloch-transf-evp}
\end{align}

Therefore, taking the inner product of~\eqref{eq:GB_of_orig} with $p_b(x;k)$, and using self-adjointness of the operator $-\left(\partial_x + 2\pi ik\right)^2 + Q$ as well as the identity~\eqref{bloch-transf-evp}, yields
\begin{equation}\label{eq:kpsi}
\left(E_b(k) - E^\lambda \right) \left\langle p_{b}(\cdot,k), \widetilde{\psi^\lambda}(\cdot,k) \right\rangle_{L^2([0,1])} + \lambda\left\langle p_{b}(\cdot,k), \left(V\psi^\lambda\right)^{\sim}(\cdot,k) \right\rangle_{L^2([0,1])} = 0.
\end{equation}
or equivalently, using notation~\eqref{def:Tb},
\begin{equation}\label{eq:kpsi_prime}
\left(E_b(k) - E^\lambda \right) \mathcal{T}_b\left\{\psi^\lambda\right\}(k) + \lambda \mathcal{T}_b\left\{V\psi^\lambda\right\}(k) = 0.
\end{equation}

We can now decompose equation~\eqref{eq:kpsi} into near- and far-frequency equations, around $E_{b_*}(k_*)$, the edge of the $b_*$-th band of the continuous spectrum.
The coupled equations for $\psi_\nr$ and $\psi_\fr$ read:
\begin{align}\label{eq:near_per}
&\left( E_{b_*}(k) - E^\lambda \right) \chi\left(|k|<\lambda^r\right) \left\langle p_{b_*}(\cdot,k), \widetilde{\psi^\lambda}(\cdot,k) \right\rangle_{L^2([0,1])}\\
 &\qquad\qquad \ + \lambda \chi\left(|k|<\lambda^r\right) \big\langle p_{b_*}(\cdot,k),\left[V\left(\psi_\nr + \psi_\fr \right)\right]^{\sim}(\cdot,k) \big\rangle_{L^2([0,1])}=0,
\nn\end{align}
and for $b\in\NN$:
\begin{align}\label{eq:far_per}
&\left( E_{b}(k) - E^\lambda \right) \chi\left(|k| \geq \lambda^r\delta_{b_*,b}\right) \left\langle p_{b}(\cdot,k), \widetilde{\psi^\lambda}(\cdot,k) \right\rangle_{L^2([0,1])}\\
&\qquad\qquad\ \ +\ \lambda \chi\left(|k| \geq \lambda^r\delta_{b_*,b}\right) \big\langle p_{b}(\cdot,k),\left[V\left(\psi_\nr + \psi_\fr \right)\right]^{\sim}(\cdot,k) \big\rangle_{L^2([0,1])}=0.
\nn\end{align}

Equivalently, we write the near and far frequency equations in the form
\begin{equation}\label{eq:near_per_prime}
\left( E_{b_*}(k) - E^\lambda \right) \widetilde{\psi}_\nr(k) + \lambda \chi\left(|k|<\lambda^r\right)\left(\mathcal{T}_{b_*}\left\{V\psi_\nr\right\}(k) + \mathcal{T}_{b_*}\left\{V\psi_\fr\right\}(k) \right) = 0,
\end{equation}
\begin{equation}\label{eq:far_per_prime}
\left( E_b(k) - E^\lambda \right) \widetilde{\psi}_{\fr,b}(k) + \lambda \chi\left(|k|\geq\lambda^r\delta_{b_*,b}\right) \left(\mathcal{T}_b\left\{V\psi_\nr\right\}(k) + \mathcal{T}_b\left\{V\psi_\fr\right\}(k) \right) = 0.
\end{equation}
Equations~\eqref{eq:near_per_prime} and~\eqref{eq:far_per_prime} are, for the case of non-trivial periodic potentials, $Q(x)$, the analogues of~\eqref{eq:near_nonper}-\eqref{eq:far_nonper}.

\subsection{Analysis of the far frequency Floquet-Bloch components}

In this section we study the far frequency equation~\eqref{eq:far_per_prime}. We will show that we can write it in terms of the near frequency solution and will determine a bound on the far solution in terms of the near solution. The next result is therefore the analogue of Proposition~\ref{prop:far_nonper_intermsof_near_nonper} and facilitates the reduction of the eigenvalue problem to a closed equation for the near-frequency components of the eigenstate.

\noindent {\em For clarity of presentation and without any loss of generality, we assume henceforth that we are localizing near the lowermost end point of the $b_*$-th band and that $k_*=0$. Thus, by Lemma~\ref{lem:band-edge},
\[ b_*\ \ \textrm{is even},\ \ k_* = 0,\ \ {\rm with}\ \ E_{b_*}(0)=E_*.\]
N.B. For $k_*=0$, note that $p_b(x;k_*)=u_b(x;k_*)$ and we use these expressions interchangeably. For $k_*=1/2$ one has to distinguish between $p_b(x;k_*)$ and $u_b(x;k_*)$.
}

\begin{Proposition}\label{prop:far_per_intermsof_near_per}
Assume $b_*$ is even and consider $E_*=E_{b_*}(0)$ the lowermost edge of the $b_*$-th band. There exists $\lambda_0>0$, such that for $0<\lambda<\lambda_0$, the following holds. Set
\begin{equation}
E^\lambda = E_* - \lambda^2 \theta^2,\ \ \ \theta \leq \lambda^{r-1}\frac{1}{2}| \partial_k^2 E_{b_*}(0)|^{1/2} \ , \qquad 0<r<\frac{1}{2}\ .\label{far-per-E}
\end{equation}
 Then for any $\psi_\nr\in L^2(\mathbb{R})$, there is a unique solution $\psi_\fr\left[\psi_\nr,\lambda\right]\in L^2(\mathbb{R})$ of the
far-frequency system~\eqref{eq:far_per_prime}. The mapping
$\left(\psi_\nr;\lambda\right) \mapsto
\psi_\fr$
 maps
$ L^2(\mathbb{R})\times(0,\lambda_0)$ to $H^2(\mathbb{R})$ and $\psi_\fr$ satisfies the bound

\begin{equation}\label{eq:far_per_bound}
\left\| \psi_\fr\left[ \psi_\nr;\lambda \right] \right\|_{H^2(\mathbb{R})} \
\leq \ C\left( \|V\|_{L^{\infty}} \right) \lambda^{1-2r} \left\| \psi_\nr \right\|_{L^2(\mathbb{R})}.
\end{equation}
\end{Proposition}
\begin{Remark}
 Recall that we have assumed $(1+|x|)V(x)\in L^1(\RR)$ and $V\in L^\infty$. It is in the proof of the bound 
~\eqref{eq:far_per_bound} that we have used  $V\in L^\infty$. We believe it possible to work under the milder assumption $(1+|x|)V(x)\in L^1(\RR)$. In this case, we would bound $\psi_\fr$ in $H^1(\RR)$ and the analysis that would follow would be a bit more technical. We leave this an exercise.
\end{Remark}
\begin{proof}
We begin by showing
that there exists $\lambda_0>0$ such that for all $0<\lambda<\lambda_0$, there is a constant $C_1>0$ such that 
\begin{align}
\left| E_{b_*}(k) - E_* \right|\ &\ge\ C_1 \lambda^{2r},\ \ \lambda^r \le |k|\le 1/2 \label{Eb*-E} \ ,\\
\left| E_b(k) - E_* \right| & \ge\ C_1, \qquad b \neq b_*,\ \ |k|\le1/2\ . \label{Enotb*-E}
\end{align}
Note first that~\eqref{Enotb*-E} is an immediate consequence of $E_*$ being the endpoint of the $(b_*)^{th}$ spectral gap. 
 To prove~\eqref{Eb*-E} recall, by Lemma~\ref{lem:band-edge} that
  $E_*=E_{b_*}(0)$, an eigenvalue at the edge of a spectral gap, is simple, and 
 $k\mapsto E_{b_*}(k)-E_*$ is continuous. 
Therefore, for any $\lambda_0$, such that $0<\lambda_0<1/2$
\begin{equation}
\min_{\lambda_0\le |k|\le 1/2} |E_{b_*}(k)-E_*|\ge C(\lambda_0)>0.
\label{cont-imp}
\end{equation}
For $|k|\le\lambda_0$, we approximate $E_{b_*}(k)$ by a Taylor expansion. 
 In particular, since $E_{b_*}(k)$ is smooth for $k$ near $k_*=0$, $\partial_k E_{b_*}(0)=0$ and $\partial^2_k E_{b_*}(0)\ne 0$, we have $E_{b_*}(k)-E_{b_*}(0) - \frac{1}{2}\partial^2_k E_{b_*}(0)\ k^2 = \mathcal{O}(|k|^3).$ Therefore, we can choose $\lambda_0>0$ sufficiently small so that for all $\lambda\le\lambda_0$ we have 
\begin{equation}
|E_{b_*}(k)-E_{b_*}(0)|\ge \frac{1}{3}\left|\partial^2_k E_{b_*}(0)\right|\lambda^{2r},\ \ 
\textrm{for all $ \lambda\le|k|\le\lambda_0$}\ .
\label{ts-imp}
\end{equation}
 It follows from~\eqref{cont-imp} and~\eqref{ts-imp} that for sufficiently small $\lambda_0>0$, 
\[
\frac{1}{2}\ge |k|\ge\lambda>0\ \ \implies\ \ |E_{b_*}(k)-E_*|\ge \min\left\{\frac{1}{3}\left|\partial^2_k E_{b_*}(0)\right|\lambda^{2r},C(\lambda_0)\right\}\ .
\]
Thus if $E^\lambda = E_* - \lambda^2 \theta^2,\ \ \ \theta \leq \lambda^{r-1}\frac{1}{2}| \partial_k^2 E_{b_*}(0)|^{1/2} $, then for $0<\lambda<\lambda_0$ sufficiently small, there is a positive constant $C_1$, such that 
\begin{equation}
|E_{b_*}(k)-E^\lambda |\ \ge \ C_1\lambda^{2r} .
\label{lowerbound-far}
\end{equation}

By~\eqref{Eb*-E} and~\eqref{Enotb*-E}, the far-frequency system,~\eqref{eq:far_per_prime}, may be re-written as
\begin{equation}\label{eq:far_per_div}
\widetilde{\psi}_{\fr,b}(k)\ +\ \lambda \frac{\chi\left( |k| \geq \lambda^r \delta_{b_*,b} \right)}{E_b(k)-E^\lambda} \ \mathcal{T}_b\left\{V\psi_\fr\right\}(k) \ = \ - \lambda \frac{\chi\left( |k| \geq \lambda^r \delta_{b_*,b} \right)}{E_b(k)-E^\lambda} \mathcal{T}_b\left\{V\psi_\nr\right\}(k),\ \ b\ge0.
\end{equation}
We wish to rewrite this equation in terms of $\psi_\fr(x)$. In order to do so, we multiply \eqref{eq:far_per_div} by $u_b(x;k)=p_b(x;k)e^{2\pi ikx}$, sum over $b\geq0$ and integrate with respect to $k\in(-1/2,1/2]$. This yields
\begin{equation}\label{eq:Klambda}
\left( I + \mathcal{K}_\lambda \right) \psi_\fr(x) = - \left( \mathcal{K}_\lambda \psi_\nr \right) (x), 
\end{equation}
where we define 
\begin{align*}
\left( \mathcal{K}_\lambda g \right) (x) & \equiv \int_{-1/2}^{1/2} \sum_{b\geq0} \lambda \frac{\chi\left( |k| \geq \lambda^r \delta_{b_*,b} \right)}{E_b(k)-E^\lambda} \ \mathcal{T}_b\left\{V \ g\right\}(k) p_b(x;k) e^{2\pi ikx} \ dk. 
\end{align*}

We next show that the operator $\mathcal{K}_\lambda$, viewed as an operator from $L^2$ to $H^2$ has small norm, for $\lambda$ small. 
Let $g \in L^2$. Using Proposition~\ref{prop:norm_equivalence}, one has
\begin{align*}
\left\| \mathcal{K}_\lambda g \right\|_{H^2}^2 & \lesssim \left\| \widetilde{ \mathcal{K}_\lambda g} \right\|_{\mathcal{X}^2}^2 = \int_{-1/2}^{1/2} \sum_{b\geq0} (1+b^2)^2 \left| \mathcal{T}_b \left\{ \mathcal{K}_\lambda g \right\} (k) \right|^2 \ dk \\
& = \lambda^2 \int_{-1/2}^{1/2} \sum_{b\geq0} (1+b^2)^2 \frac{\chi\left( |k| \geq \lambda^r \delta_{b_*,b} \right)}{|E_b(k)-E^\lambda|^2} \left| \mathcal{T}_b\left\{V \ g\right\}(k) \right|^2 \ dk.
\end{align*}
Now, by~\eqref{lowerbound-far}, for $|k|\ge\lambda^r$ one has $ \left|E_{b_*}(k)-E^\lambda\right|^{-1} \leq C_1 \lambda^{-2r}$, and recall $0<r<1/2$. For $b \neq b_*$, we use Weyl asymptotics to write $\left| (E_b(k)-E^\lambda)^{-1} \right| \sim \left|(b^2 - E_*)^{-1}\right| \sim (b^2 + 1)^{-1}$. We therefore have
\begin{align*}
\left\| \mathcal{K}_\lambda g \right\|_{H^2}^2 & \lesssim  \lambda^2 \int_{-1/2}^{1/2} \sum_{b\geq0} \left| \mathcal{T}_b \left\{ V \ g \right\}(k)\right|^2 \ dk \ + \ \lambda^{2-4r} \int_{-1/2}^{1/2} (1+{b_*}^2)^2 \chi\left( |k| \geq \lambda^r \right) \left| \mathcal{T}_{b_*} \left\{ V \ g \right\}(k) \right|^2 \ dk \\
& \lesssim  \lambda^{2-4r} \left\| \left( V \ g \right)^{\sim} \right\|_{\mathcal{X}^0}^2 \lesssim \lambda^{2-4r} \left\| V \right\|_{L^\infty}^2 \left\| g \right\|_{L^2}^2.
\end{align*}
Thus, since $r\in(0,1/2)$, one can choose $\lambda_0>0$ such that if $0<\lambda<\lambda_0$, then $\| \mathcal{K}_\lambda \|_{L^2 \rightarrow H^2} <1 $.
In particular, $\mathcal{K}_\lambda$ is a contraction from $L^2$ to $L^2$, and therefore $I+\mathcal{K}_\lambda$ is invertible. The existence and uniqueness of $\psi_\fr\in L^2(\mathbb{R})$ solution to~\eqref{eq:far_per_prime} is now given through~\eqref{eq:Klambda}. Moreover, one has
\begin{align*}
\left\| \psi_\fr \right\|_{H^2} & = \left\| \left(I + \mathcal{K}_\lambda \right)^{-1} \left( \mathcal{K}_\lambda \psi_\nr \right) \right\|_{H^2} \leq \| \left( I + \mathcal{K}_\lambda \right)^{-1} \|_{H^2 \rightarrow H^2} \left\| \mathcal{K}_\lambda \right\|_{L^2 \rightarrow H^2} \left\| \psi_\nr \right\|_{L^2} \\
& \lesssim \lambda^{1-2r} \left\| V \right\|_{L^\infty} \left\| \psi_\nr \right\|_{L^2},
\end{align*}
which implies the bound \eqref{eq:far_per_bound}. The proof of Proposition \ref{prop:far_per_intermsof_near_per} is complete.
\end{proof}

\subsection{Analysis of the near frequency Floquet-Bloch component}

With the properties of the map $\psi_\nr\mapsto\psi_\fr[\psi_\nr,\lambda]$ now understood via Proposition~\ref{prop:far_per_intermsof_near_per}, we now view and study~\eqref{eq:near_per_prime} as a closed eigenvalue problem for $(E^\lambda,\psi_\nr)$:
\begin{equation}
\left( E_{b_*}(k) - E^\lambda \right) \widetilde{\psi}_\nr(k) + \lambda\chi_{_{\lambda^r}}(k)\mathcal{T}_{b_*} \left\{ V \psi_\nr \right\}(k)
 + \lambda\ \chi_{_{\lambda^r}}(k) \mathcal{T}_{b_*} \left\{ V \psi_\fr\left[\psi_\nr;\lambda\right]\ \right\}(k) = 0.
\label{closed}\end{equation}
Equation~\eqref{closed} is localized in the region $|k|<\lambda^r,\ 0<r<1/2$. 
By careful expansion and rescaling of~\eqref{closed} we shall obtain an equation, which at leading order in $\lambda$, is a perturbation of the general class of equations to which Lemma~\ref{lem:homogeneous} applies. The size of the perturbation is estimated in Proposition~\ref{prop:R_per_bound} and the perturbed equation is then solved by applying Lemma~\ref{lem:technical}.

 In Lemmata~\ref{lem:R0},~\ref{lem:R1} and~\ref{lem:R2} we expand the first two terms in~\eqref{closed} about $k_*=0$ using Taylor's Theorem, making explicit the leading and higher order contributions.

\begin{Lemma}\label{lem:R0}
Denote $E^\lambda = E_* - \lambda^2\theta^2= E_{b_*}(0) - \lambda^2\theta^2$, as in Proposition~\ref{prop:far_per_intermsof_near_per}.
There exists $k'$ such that $|k'| < \lambda^r$, and
\[ \left( E_{b_*}(k) - E^\lambda \right) \widetilde{\psi}_\nr(k) \ = \ \left( \frac12\partial_k^2 E_{b_*}(0) \ k^2 +\lambda^2\theta^2 \right) \widetilde{\psi}_\nr(k) \ + \ \lambda R_0\left[\widetilde\psi_\nr;\lambda\right](k,k'),\]
where
\begin{equation} R_0\left[\widetilde\psi_\nr;\lambda\right](k,k')\ =\ \frac{1}{\lambda}\ \frac1{4! }\ k^4\ \partial_k^4 E_{b_*}(k')\ \widetilde{\psi}_\nr(k). \label{R0-eqn}
\end{equation}

\end{Lemma}
\begin{proof}
Taylor expanding $E_{b_*}(k)$ about $k_*=0$ to fourth order and making use of $E^\lambda = E_{b_*}(0) - \lambda^2\theta^2$ and $\partial_k^j E_{b_*}(0)=0$ for $j = 1, \ 3$,
one obtains
$
E_{b_*}(k) - E^\lambda = \frac12\partial_k^2 E_{b_*}(0)k^2 + \lambda^2\theta^2 + \frac{1}{4!}\partial_k^4 E_{b_*}(k') k^4,
$ 
which is equivalent to~\eqref{R0-eqn}.\end{proof}
\begin{Lemma}\label{lem:R1}
One can decompose
\[
\mathcal{T}_{b_*} \left\{ V \psi_\nr \right\}(k)  =  \left\langle p_{b_*}(\cdot;0), p_{b_*}(\cdot;0) \mathcal{T}\left\{ V \mathcal{F}^{-1}\left\{ \widetilde{\psi}_\nr \right\}\right\}(\cdot;k) \right\rangle_{L^2([0,1])} + R_1\left[\widetilde\psi_\nr;\lambda\right](k),
\]
with
\begin{align}
R_1\left[\widetilde\psi_\nr;\lambda\right](k)\ & =\ \left\langle p_{b_*}(\cdot;0),\mathcal{T}\left\{V\mathcal{E}_1 \right\}(\cdot,k)\right\rangle_{L^2([0,1])} \nn\\
& \qquad + \left\langle p_{b_*}(\cdot;k) - p_{b_*}(\cdot;0),\mathcal{T}\left\{ V \psi_\nr \right\}(\cdot,k) \right\rangle_{L^2([0,1])} ,\label{R1def} 
\end{align}
where $\mathcal{E}_1 \equiv \mathcal{T}^{-1}\left\{ \widetilde{\psi}_\nr(k) \big(p_{b_*}(x;k) - p_{b_*}(x;0)\big ) \right\} $.
\end{Lemma}
\begin{proof}
Let us recall that by definition~\eqref{eq:psi_decomp}, one has
\begin{equation}
\psi_\nr(x) = \mathcal{T}^{-1}\left\{ \widetilde{\psi}_\nr(\cdot) p_{b_*}(x;\cdot) \right\}.
\label{psi-near-phi}
\end{equation}
Since $\widetilde{\psi}_\nr(k) = \chi \left( |k| < \lambda^r \right) \widetilde{\psi}_\nr(k) $, we decompose:
\begin{align}
\psi_\nr(x) & = \mathcal{T}^{-1}\left\{ \widetilde{\psi}_\nr(\cdot) p_{b_*}(x;\cdot) \right\}(x) 
 = p_{b_*}(x;0)\mathcal{F}^{-1}\{ \widetilde{\psi}_\nr\} + \mathcal{E}_1(x) \label{eq:decomp-psi_near_per} 
\end{align}
where
\begin{equation}
\mathcal{E}_1(x) \ \equiv \ \mathcal{T}^{-1}\left\{ \widetilde{\psi}_\nr(\cdot) \big(p_{b_*}(x;\cdot) - p_{b_*}(x;0)\big ) \right\}. \label{eq:E1_f_t}
\end{equation}
Above, we used that $\mathcal{T}^{-1}$ commutes with multiplication by a $1$- periodic function of $x$, and that when acting on a function which is localized near $k=0$, and which does not depend on $x$, $\mathcal{T}^{-1}$ is equivalent to the standard inverse Fourier transform; see Section~\ref{sec:background}.

The proof of Lemma~\ref{lem:R1} is now straightforward.
\end{proof}

We next give a precise expression of the leading order term in Lemma~\ref{lem:R1}.
\begin{Lemma}\label{lem:R2}
One can decompose
\begin{align}
&\left\langle p_{b_*}(\cdot;0), p_{b_*}(\cdot;0) \mathcal{T}\left\{ V \mathcal{F}^{-1}\{ \widetilde{\psi}_\nr \}\right\}(\cdot;k) \right\rangle_{L^2([0,1])}  \label{firstterm}
 \\
 &\qquad\ =\left(\int_{-\infty}^{\infty} |p_{b_*}(x;0)|^2 V(x) dx\right) \int_{-\infty}^{\infty}\widetilde{\psi}_\nr(l)dl \ + \ R_2\left[\widetilde\psi_\nr\right](k),
\nn\end{align}
with
\begin{equation} R_2\left[\widetilde\psi_\nr\right](k) \ = \ \int_{-\infty}^{\infty} dx\ |p_{b_*}(x;0)|^2 \ V(x)  \int_{-\infty}^{\infty} \big(e^{2i\pi (l-k)x}-1\big)\widetilde{\psi}_\nr(l)\ dl.
\label{R2def}\end{equation}
\end{Lemma}
\begin{proof}
By the definition of $\mathcal{T}$, one has
\begin{align*}
\mathcal{T}\left\{ V \mathcal{F}^{-1}\{ \widetilde{\psi}_\nr\}\right\}(x;k) & =  \sum_{n\in \mathbb{Z}} e^{2\pi i nx} \mathcal{F} \left\{ V \mathcal{F}^{-1}\{ \widetilde{\psi}_\nr\}\right\} (k+n)\\
& =  \sum_{n\in \mathbb{Z}} e^{2\pi i nx} \int_{-\infty}^{\infty} \widehat{V}(k+n-l) \widetilde{\psi}_\nr(l)\ dl.
\end{align*}

Since $|k|<\lambda^r$ and $\widetilde{\psi}_\nr(l)$ is localized on $|l|<\lambda^r$, the leading order term is obtained when replacing $\widehat{V}(k+n-l)$ with $\widehat{V}(n)$. The first term of~\eqref{firstterm} now follows from the identity:
\begin{multline*}
 \sum_{n \in \mathbb{Z}} \left\langle p_{b^*}(\cdot;0),p_{b^*}(\cdot;0) e^{2\pi i n\cdot} \right\rangle_{L^2([0,1])} \widehat{V}(n) 
 = \int_0^1 |p_{b^*}(x;0)|^2 \sum_{n \in \mathbb{Z}} e^{2\pi inx} \widehat{V}(n) dx\nn\\ 
 = \sum_{n \in \mathbb{Z}} \int_0^1 |p_{b^*}(x;0)|^2 V(x+n) dx 
 = \int_{-\infty}^{\infty} |p_{b^*}(x;0)|^2 V(x) dx.
\end{multline*}
Here, we used the Poisson summation formula and that $x\mapsto p_{b^*}(x;0)$ is $1-$ periodic. 

Similarly, one has
\[ \sum_{n \in \mathbb{Z}} \left\langle p_{b^*}(\cdot;0),p_{b^*}(\cdot;0) e^{2\pi i n\cdot} \right\rangle_{L^2([0,1])} \widehat{V}(n+k-l) = \int_{-\infty}^\infty |p_{b^*}(x;0)|^2 e^{2i\pi (l-k)x}V(x) dx .\]
This completes the proof of Lemma~\ref{lem:R2}. 
\end{proof}

\noindent {\bf The rescaled closed equation.}\ Using Lemmata~\ref{lem:R0},~\ref{lem:R1} and~\ref{lem:R2}, one can express the near frequency equation~\eqref{closed} as follows:
\begin{align}\label{eq:near_per_3}
&(\frac12 \partial_k^2 E_{b_*}(0) k^2 + \lambda^2\theta^2) \widetilde{\psi}_\nr(k)\ +\ \lambda\ \chi\left(|k|<\lambda^r\right) \left(\int_{-\infty}^{\infty} |p_{b_*}(x;0)|^2 V(x) dx\right) \int_{-\infty}^{\infty}\widetilde{\psi}_\nr(l)dl \nn\\
&\qquad\qquad\ = -\lambda \chi\left(|k|<\lambda^r\right) \mathcal{R}\left[\psi_\nr;\lambda\right](k),
\end{align}
where
$
\mathcal{R}\left[\psi_\nr;\lambda\right](k) \equiv \mathcal{T}_{b_*}\{V\psi_\fr\}+ R_0+R_1+R_2.
$

Seeking to extract the dominant and higher order terms in $\lambda$, we introduce the scaled near-frequency components: 
 \begin{equation}
 \widetilde{\psi}_\nr(k) = \frac{1}{\lambda}\widehat{\Phi}_{\lambda}\left(\frac{k}{\lambda}\right)\ =\ \frac{1}{\lambda}\widehat{\Phi}_{\lambda}\left(\kappa\right),\ \ {\rm where}\ \ k=\lambda\ \kappa\ .
 \label{phi2Phi}
 \end{equation}
 Expressing~\eqref{eq:near_per_3} in terms of $\widehat{\Phi}_\lambda$ and $\kappa$ we obtain, after dividing out by $\lambda$, 
 \begin{multline}
 \left(\frac12 \partial_k^2 E_{b_*}(0) \kappa^2 + \theta^2\right) \chi_{_{\lambda^{r-1}}}(\kappa)\widehat{\Phi}_{\lambda}(\kappa)
\ +\ \left(\int_{\mathbb{R}} |p_{b_*}(\cdot;0)|^2 V\right)\ \chi_{_{\lambda^{r-1}}}(\kappa)\int_{\mathbb R}\chi_{_{\lambda^{r-1}}}(\eta)\widehat{\Phi}_{\lambda}(\eta)d\eta \\
= - \chi\left(|\kappa|<\lambda^{r-1}\right) \mathcal{R}\left[\psi_\nr;\lambda\right](\lambda\kappa)\ \equiv\ R(\widehat \Phi_\lambda). \ 
\label{eq:near_per_rescaled}
\end{multline}

Equation~\eqref{eq:near_per_rescaled} is of the form 
$ \widehat{\mathcal L}_0[\theta]\widehat{\Phi}_{\lambda}(\kappa)\ =\ R(\widehat \Phi_\lambda)$, 
where $\widehat{\mathcal L}_0[\theta]$ is given by~\eqref{hatcalL0-def} with parameters
\[ A=\frac{1}{8\pi^2} \partial_k^2 E_{b_*}(0) \ , \ B=-\int_{\mathbb{R}} |p_{b_*}(x;0)|^2 V(x) dx \ , \ {\rm and}\ \beta = 1-r \ .\]
In order to solve~\eqref{eq:near_per_rescaled} via Lemma~\ref{lem:technical} we need a bound 
on $R(\widehat \Phi_\lambda)$ of the form~\eqref{assumptionsR}. 

\begin{Proposition} \label{prop:R_per_bound} Assume that $V$ is such that $(1+|\cdot|)V(\cdot)\in L^1$ and $V\in L^\infty$. Then
 $R(\widehat{\Phi}_\lambda)$, defined in~\eqref{eq:near_per_rescaled}, satisfies the bound
\begin{equation}
\left\|R(\widehat{\Phi}_\lambda)\right\|_{L^{2,-1}} = \left\| \chi\left(|\cdot|<\lambda^{r-1}\right) \mathcal{R}\left[\psi_\nr;\lambda\right](\lambda\cdot) \right\|_{L^{2,-1}} \le\ C \lambda^{\alpha(r)}\ \left\| \widehat{\Phi}_{\lambda} \right\|_{L^{2,1}}. 
 \label{remainder-bound}
\end{equation}
where $\alpha(r)=\max\left\{\frac{1}{2}-2r,2r,\frac{r+1}{2} \right\}$.
The constant $C$ depends on $\big\Vert (1+|\cdot|)V\big\Vert_{L^1},\big\Vert V\big\Vert_{L^\infty}$ as well as
 $\displaystyle   \sup_{|k|<\lambda^r}\| p_{b_*}(\cdot;k)\|_{L^{\infty}} ,\ \sup_{|k|<\lambda^r}\sum_{n\in\mathbb{Z}} \big| \big\langle p_{b_*}(\cdot;k),e^{2\pi in\cdot} \big\rangle_{L^2([0,1])} \big|,\ \sup_{|k|<\lambda^r} \big|\partial_k^4 E_{b_*}(k)\big|,\ 
\sup_{|k|<\lambda^r}\big\Vert\partial_k p_{b_*}(\cdot;k)\big\Vert_{L^{\infty}}$, \\   and is finite by Lemmata~\ref{lem:regularity-of-Eb} and~\ref{lem:estimates}. 
\end{Proposition}
\begin{proof}[Proof of Proposition~\ref{prop:R_per_bound}]
Recall that $\mathcal{R}(\lambda\kappa)$,  the right hand side of~\eqref{eq:near_per_rescaled} has the form 
\begin{align}
&\mathcal{R}\left[\psi_\fr\left[\psi_\nr;\lambda\right],\psi_\nr;\lambda\right](\lambda\kappa) = \chi\left(|\kappa|<\lambda^{r-1}\right) \Big(\mathcal{T}_{b_*} \left\{ V \psi_\fr \right\}(\lambda\kappa) +  R_0\left[\widetilde\psi_\nr;\lambda\right](\lambda\kappa,k') \nn \\
&+ R_1\left[\widetilde\psi_\nr;\lambda\right](\lambda\kappa)+ R_2\left[\widetilde\psi_\nr\right](\lambda\kappa)\Big)\ \equiv\ (I)\ +\ (II)\ +\ (III)\ +\ (IV)\ .\label{Rlamkap}
\end{align}
We proceed by estimating each of the terms: $(I),\ (II),\ (III)\ $ and $(IV)$.
\item[{\bf $(I)$\ Estimation of 
$\chi\left(|\kappa|<\lambda^{r-1}\right) \mathcal{T}_{b_*} \left\{ V \psi_\fr \right\}(\lambda\kappa)$}:]
\ \ We have 
\begin{align*}
\left\| \chi\left(|\cdot|<\lambda^{r-1}\right)\mathcal{T}_{b_*} \left\{ V \psi_\fr \right\}(\lambda\cdot) \right\|_{L^{2,-1}}^2 
& =  \int_{-\infty}^{\infty} \frac{\chi\left(|\kappa|<\lambda^{r-1}\right)}{1+\kappa^2} \left| \mathcal{T}_{b_*} \left\{ V \psi_\fr \right\}(\lambda\kappa) \right|^2 d\kappa  \\
&\leq  \left\| \mathcal{T}_{b_*} \left\{ V \psi_\fr \right\} \right\|_{L^{\infty}}^2.
\end{align*}
We now consider $\mathcal{T}_{b_*} \left\{ V \psi_\fr \right\}(\cdot)$ in detail. By definition, one has
\begin{align*}
\mathcal{T}_{b_*} \left\{ V \psi_\fr \right\}(k) &= \left\langle p_{b_*}(\cdot;k), \mathcal{T}\left\{V\psi_\fr\right\}(\cdot,k)\right\rangle_{L^2([0,1])} \\
& = \left\langle p_{b_*}(\cdot;k), \sum_{n\in\mathbb{Z}} e^{2\pi in\cdot} \int_{-\infty}^{\infty} \widehat{V}(k + n - l) \widehat{\psi}_\fr(l)dl \right\rangle_{L^2([0,1])} \\
& = \sum_{n\in\mathbb{Z}} \left\langle p_{b_*}(\cdot;k),e^{2\pi in\cdot} \right\rangle_{L^2([0,1])} \int_{-\infty}^{\infty} \frac{\widehat{V}(k + n - l)}{ (1+|l|^2)^{1/2}} (1+|l|^2)^{1/2}\widehat{\psi}_\fr(l)dl .
\end{align*}
Moreover,
\begin{align*}
\left| \int_{-\infty}^{\infty} \frac{\widehat{V}(k + n - l)}{ (1+|l|^2)^{1/2}} (1+|l|^2)^{1/2}\widehat{\psi}_\fr(l)dl \right|  &\leq\ \big\| \widehat V \big\|_{L^\infty} \left\| \psi_\fr \right\|_{H^2}  \lesssim \lambda^{1-2r} \left\| \psi_\nr \right\|_{L^2}\\
& \lesssim\ \lambda^{1-2r} \left\| \widetilde{\psi}_\nr \right\|_{L^2} \
 = \ \lambda^{1-2r} \lambda^{-\frac{1}{2}} \left\| \widehat{\Phi}_{\lambda} \right\|_{L^2} ,
\end{align*}
where we used Proposition~\ref{prop:far_per_intermsof_near_per}, definition~\eqref{phi2Phi} and, by Proposition~\ref{prop:norm_equivalence},
\begin{align} \left\| \psi_\nr \right\|_{L^2}^2 \ &= \ \left\|\mathcal{T}^{-1}\big\{ \widetilde{\psi}_\nr (k)p_{b_*}(x;k) \big\}\right\|_{L^2}^2 \ \lesssim \ \left\| \widetilde{\psi}_\nr (k)p_{b_*}(x;k) \right\|_{\mathcal X^0}^2 \nn \\
&=\ \int_{-1/2}^{1/2}|\widetilde{\psi}_\nr (k)|^2\ dk \ = \ \big\| \widetilde\psi_\nr \big\|_{L^2}^2.\label{eq:norm-equivalence}
\end{align}

Finally, it follows
\begin{equation}\label{est(I)}
\left\| \mathcal{T}_{b_*} \left\{ V \psi_\fr \right\} \right\|_{L^{\infty}} \leq \lambda^{\frac{1}{2}-2r} \ C \left\| \widehat{\Phi}_{\lambda} \right\|_{L^{2,1}} \ .
\end{equation}
with $C=C\left( \big\Vert \widehat{V} \big\Vert_{L^{\infty}}, \ \big\Vert V\big\Vert_{L^{\infty}}, \ \sup_{|k|<\lambda^r}\sum_{n\in\mathbb{Z}} \left| \left\langle p_{b_*}(\cdot;k),e^{2\pi in\cdot} \right\rangle_{L^2([0,1])} \right| \right)$.

\item[{\bf $(II)$\ Estimation of 
$\chi\left(|\kappa|<\lambda^{r-1}\right) R_0\left[\widetilde\psi_\nr;\lambda\right](\lambda\kappa,k') $, given in~\eqref{R0-eqn}}:]
We have (constants implicit)
\begin{multline*}
\left\| \chi\left(|\cdot|<\lambda^{r-1}\right) \lambda^2 (\cdot)^4 \widehat{\Phi}_{\lambda}(\cdot) \right\|_{L^{2,-1}(\mathbb{R})}^2 
 = \lambda^4 \int_{-\infty}^{\infty} \frac{\kappa^8}{1+\kappa^2} \chi\left(|\kappa|<\lambda^{r-1}\right) \left| \widehat{\Phi}_{\lambda}(\kappa)\right|^2 d\kappa \\
 =\lambda^4 \int_{-\infty}^{\infty} \frac{\kappa^8}{(1+\kappa^2)^2} \chi\left(|\kappa|<\lambda^{r-1}\right) (1+\kappa^2) \left| \widehat{\Phi}_{\lambda}(\kappa)\right|^2 d\kappa \\ \lesssim \lambda^4 \sup_{|\kappa|<\lambda^{r-1}}\left| \frac{\kappa^8}{(1+\kappa^2)^2} \right| \left\| \widehat{\Phi}_{\lambda} \right\|_{L^{2,1}}^2 \ \lesssim\  \lambda^{4r} \left\| \widehat{\Phi}_{\lambda} \right\|_{L^{2,1}}^2.
\end{multline*}
Therefore,
\begin{align}
\left\| \chi\left(|\kappa|<\lambda^{r-1}\right) R_0\left[\widetilde\psi_\nr;\lambda\right](\lambda\kappa,k') \right\|_{L^{2,-1}} &\equiv \left\| \chi\left(|\kappa|<\lambda^{r-1}\right)\frac{1}{4!}\partial_k^4 E_{b_*}(k') \lambda^2 \kappa^4 \widetilde{\Phi}_{\lambda}(\kappa) \right\|_{L^{2,-1}} \nn\\
&\lesssim\ \lambda^{2r}\ 
\sup_{|k'|<\lambda^{r}} |\partial_k^4 E_{b_*}(k')|\ \left\| \widehat{\Phi}_{\lambda} \right\|_{L^{2,1}}.
\label{est(II)}\end{align}

\item[{\bf $(III)$\ Estimation of 
$\chi\left(|\kappa|<\lambda^{r-1}\right) R_1\left[\widetilde\psi_\nr;\lambda\right](\lambda\kappa)$, given in~\eqref{R1def}:}] 

Recall
\begin{align}
R_1\left[\widetilde\psi_\nr;\lambda\right](k) &= \left\langle p_{b_*}(\cdot;0),\mathcal{T}\left\{V\mathcal{E}_1 \right\}(\cdot,k)\right\rangle_{L^2([0,1])} \nn\\
& \qquad + \left\langle p_{b_*}(\cdot;k) - p_{b_*}(\cdot;0),\mathcal{T}\left\{ V \psi_\nr \right\}(\cdot,k) \right\rangle_{L^2([0,1])} .\label{R1defA}
\end{align}
where $\mathcal{E}_1 \equiv \mathcal{T}^{-1}\left\{ \widetilde{\psi}_\nr(k) \big(p_{b_*}(x;k) - p_{b_*}(x;0)\big ) \right\} $.

Let us first obtain an estimate on $\mathcal{E}_1$. Using Taylor expansion of $p_{b_*}(x;\cdot)$ around $0$, one has
\begin{align*}
\left|\mathcal{E}_1(x)\right| \ &= \ \left| \int_{-1/2}^{1/2}e^{2\pi i kx} \widetilde{\psi}_\nr(k) \big(p_{b_*}(x;k) - p_{b_*}(x;0)\big ) \ dk \right| \\
\ &\leq \ \sup_{x\in\RR,|k'|<\lambda^r}\vert \partial_k p_{b_*}(x;k') \vert \int_{-\infty}^{\infty}| k \chi\left(|k|<\lambda^{r}\right) \widetilde{\psi}_\nr(k) | \ dk \\
\ &\leq \ \lambda \sup_{|k'|<\lambda^r}\Vert \partial_k p_{b_*}(\cdot;k') \Vert_{L^\infty} \int_{-\infty}^{\infty}| \kappa \chi\left(|\kappa|<\lambda^{r-1}\right) \widehat{\Phi}_\lambda(\kappa) | \ d\kappa \\
\ &\leq \ \lambda \sup_{|k'|<\lambda^r}\Vert \partial_k p_{b_*}(\cdot;k') \Vert_{L^\infty} \left(\int_{|\kappa|<\lambda^{r-1}}\frac{\kappa^2}{1+\kappa^2} \ d\kappa\right)^{1/2} \left\|\widehat{\Phi}_\lambda \right\|_{L^{2,1}} \\
\ &\leq \ 2\lambda^{\frac{1+r}2}\sup_{|k'|<\lambda^r}\Vert \partial_k p_{b_*}(\cdot;k') \Vert_{L^\infty}\left\|\widehat{\Phi}_\lambda \right\|_{L^{2,1}} ,
\end{align*}
so that we deduce
\begin{equation}\label{estE1}
\big\|\mathcal{E}_1\big\|_{L^\infty} \ \leq \ 2\lambda^{\frac{1+r}2}\sup_{|k'|<\lambda^r} \|\partial_k p_{b_*}(x;k')\|_{L^{\infty}} \left\|\widehat{\Phi}_\lambda \right\|_{L^{2,1}} .
\end{equation}

Estimation of the first term of~\eqref{R1defA} is as follows. One has
\begin{multline*}
\left\| \chi\left(|\kappa|<\lambda^{r-1}\right) \left\langle p_{b_*}(\cdot;0),\mathcal{T}\left\{V\mathcal{E}_1\right\}(\cdot,\lambda\kappa)\right\rangle_{L^2([0,1])} \right\|_{L^{2,-1}_\kappa}^2 = \\
\int_{-\infty}^{\infty} \frac{\chi\left(|\kappa|<\lambda^{r-1}\right)}{1+\kappa^2}\left|\left\langle p_{b_*}(\cdot;0),\mathcal{T}\left\{V\mathcal{E}_1\right\}(\cdot,\lambda\kappa)\right\rangle_{L^2([0,1])}\right|^2 d\kappa.
\end{multline*}
Turning to the integrand of the above expression, we rewrite the inner product
\begin{align*}
\left\langle p_{b_*}(\cdot;0),\mathcal{T}\left\{V\mathcal{E}_1\right\}(\cdot,\lambda\kappa)\right\rangle_{L^2([0,1])} &= \int_{0}^1 \mathcal{T}\left\{ p_{b_*}(\cdot;0) \mathcal{E}_1(\cdot) V(\cdot)\right\}(x;\lambda\kappa)\\
& = \int_{0}^1 \sum_{n\in\mathbb Z}e^{2\pi i n x}\mathcal F\left\{ p_{b_*}(\cdot;0) \mathcal{E}_1(\cdot) V(\cdot)\right\} (\lambda\kappa+n)\ dx \\
& = \mathcal F\left\{ p_{b_*}(\cdot;0) \mathcal{E}_1(\cdot) V(\cdot)\right\} (\lambda\kappa),
\end{align*}
where we used that $p_{b_*}(x;0)$ is $1-$periodic, so that it commutes with $\mathcal{T}$, and the Poisson summation formula. It follows that
\[
\left|\left\langle p_{b_*}(\cdot;0),\mathcal{T}\left\{V\mathcal{E}_1\right\}(\cdot,\lambda\kappa)\right\rangle_{L^2([0,1])}\right|\leq \big\|p_{b_*}(\cdot;0) \mathcal{E}_1(\cdot) V(\cdot)\big\|_{L^1}
\leq \big\|\mathcal{E}_1\big\|_{L^\infty}\int |p_{b_*}(x;0)| |V(x)|\ dx.
\]
Using~\eqref{estE1}, one deduces
\begin{equation}\label{estR11}
\left\| \chi\left(|\kappa|<\lambda^{r-1}\right) \left\langle p_{b_*}(\cdot;0),\mathcal{T}\left\{V\mathcal{E}_1\right\}(\cdot,\lambda\kappa)\right\rangle_{L^2([0,1])} \right\|_{L^{2,-1}_\kappa} \ \leq \ C\lambda^{\frac{1+r}{2}}\left\|\widehat{\Phi}_\lambda \right\|_{L^{2,1}} ,
\end{equation}
with $C=C(\sup_{|k|<\lambda^r}\|\partial_k p_{b_*}(\cdot;k)\|_{L^{\infty}},\int |p_{b_*}(x;0)| |V(x)|\ dx)$.
\medskip

The last term in~\eqref{R1defA} is estimated as follows. Note that
\begin{align*}
&\left| \left\langle p_{b_*}(\cdot;\lambda\kappa) - p_{b_*}(\cdot;0) , \mathcal{T}\left\{ V \psi_\nr \right\}(\cdot,\lambda\kappa) \right\rangle_{L^2([0,1])}\right| \\
&\qquad = \left| \int_0^1 \big(p_{b_*}(x;\lambda\kappa) - p_{b_*}(x;0)\big)\sum_{n\in \mathbb{Z}} e^{2\pi inx} \mathcal{F}\{ V \psi_\nr  \}(\lambda\kappa+n) dx \right| \\
&\qquad = \left| \int_0^1 \big(p_{b_*}(x;\lambda\kappa) - p_{b_*}(x;0)\big)\sum_{n\in \mathbb{Z}}( V \psi_\nr ) (x+n) e^{-2\pi i (\lambda\kappa+n)x} dx \right| \\
&\qquad \leq \int_{-\infty}^{\infty} \left| \big(p_{b_*}(x;\lambda\kappa) - p_{b_*}(x;0)\big)  V(x) \psi_\nr (x)\right|\ dx  \\
&\qquad \leq \lambda \kappa \sup_{|k'|<\lambda^r} \|\partial_k p_{b_*}(\cdot;k')\|_{L^{\infty}} \| \psi_\nr \|_{L^{\infty}} \|V\|_{L^1} ,
\end{align*}
where we used the Poisson summation formula along with the periodicity of $p_{b_*}(x;\lambda\kappa) - p_{b_*}(x;0)$ and its Taylor expansion as $|\lambda\kappa|<\lambda^r$.
Now, note that 
\[\| \psi_\nr \|_{L^{\infty}} = \| \mathcal{T}^{-1}\big\{ \widetilde{\psi}_\nr (k)p_{b_*}(x;k) \big\} \|_{L^{\infty}} \leq \sup_{|k|<\lambda^r}\| p_{b_*}(\cdot;k)\|_{L^{\infty}} \int_{-\lambda^r}^{\lambda^r} | \widetilde{\psi}_\nr(l)|\ dl \]
and
\[ \int_{-\infty}^\infty | \widetilde{\psi}_\nr(l)|\ dl \ = \ \int_{-\infty}^\infty | \widehat{\Phi}_\lambda(\eta)|\ d\eta 
 \ = \ \int_{-\infty}^\infty \frac{1}{(1+\eta^2)^{1/2}}(1+\eta^2)^{1/2}|\widehat{\Phi}_\lambda(\eta)|\ d\eta\leq C\left\| \widehat{\Phi}_\lambda \right\|_{L^{2,1}}.
\]
It follows
\begin{multline}
\left\| \chi\left(|\kappa|<\lambda^{r-1}\right) \left\langle p_{b_*}(\cdot;\lambda\kappa) - p_{b_*}(\cdot;0) , \mathcal{T}\left\{ V \psi_\nr \right\}(\cdot,\lambda\kappa) \right\rangle_{L^2([0,1])}\right\|_{L^{2,-1}_\kappa}  \\
 \leq \ C\lambda \left\|\widehat{\Phi}_\lambda \right\|_{L^{2,1}}\left(\int \frac{\kappa^2 \chi\left(|\kappa|<\lambda^{r-1}\right)}{1+\kappa^2} \right)^{1/2}\lesssim \lambda^{\frac{1+r}2}\left\|\widehat{\Phi}_\lambda\right\|_{L^{2,1}} ,\label{estR12}
\end{multline}
with $C=C\left(\sup_{|k|<\lambda^r}\| p_{b_*}(\cdot;k)\|_{L^{\infty}} ,\sup_{|k'|<\lambda^r}\|\partial_kp_{b_*}(\cdot;k')\|_{L^{\infty}}, \ \|V\|_{L^1}\right)$.
\medskip

Estimates~\eqref{estR11} and~\eqref{estR12} yield
\begin{equation}\label{est(III)}
\left\| \chi\left(|\kappa|<\lambda^{r-1}\right) R_1[\widetilde\psi_\nr](\lambda\kappa) \right\|_{L^{2,-1}} \leq C\left\|\widehat{\Phi}_{\lambda} \right\|_{L^{2,1}}\ \lambda^{\frac{1+r}{2}} .
\end{equation}
with $C=C\left(\sup_{|k|<\lambda^r}\| p_{b_*}(\cdot;k)\|_{L^{\infty}} ,\sup_{|k|<\lambda^r}\|\partial_kp_{b_*}(\cdot;k)\|_{L^{\infty}}, \ \|V\|_{L^1}\right) $.
\medskip

\item[{\bf $(IV)$ Estimation of $\chi\left(|\kappa|<\lambda^{r-1}\right) R_2[\widetilde\psi_\nr](\lambda\kappa)$, given in~\eqref{R2def}:}] Recall
\[
R_2\left[\widetilde\psi_\nr\right](k) \ = \ \int_{-\infty}^{\infty} dx\ |p_{b_*}(x;0)|^2 \ V(x)  \int_{-\infty}^{\infty} \big(e^{2i\pi (l-k)x}-1\big)\widetilde{\psi}_\nr(l)dl.
\]
We now use that $\left| e^{2i\pi (l-k)x}-1\right| \leq 2\pi |l-k||x|$. It follows
\[
\left| R_2\left[\widetilde\psi_\nr\right](\lambda\kappa)\right| \ \leq \ 2\pi \lambda \int_{-\infty}^{\infty} dx\ |p_{b_*}(x;0)|^2 \ |x\ V(x)|  \int_{-\infty}^{\infty} |\kappa-\eta| |\widehat{\Phi}_{\lambda}(\eta)|d\eta.
\]
We therefore define
\begin{equation}
\mathcal{I}(\kappa)\ =\ - \chi\left(|\kappa|<\lambda^{r-1}\right) \
 \int_{-\infty}^{\infty}|\kappa-\eta|\chi\left(|\eta|<\lambda^{r-1}\right)|\widehat{\Phi}_{\lambda}(\eta)|d\eta.
 \label{Ikappa}
 \end{equation}
 
The integral, $\mathcal{I}(\kappa)$, is bounded in $L^{2,-1}(\mathbb{R})$ as follows:
\begin{align*}
\left\| \mathcal{I} \right\|_{L^{2,-1}}^2 
& \leq   \int_{-\infty}^{\infty} \frac{\chi\left(|\kappa|<\lambda^{r-1}\right)}{1+\kappa^2} \int_{|\eta|<\lambda^{r-1}} \frac{|\kappa-\eta|^2}{1+\eta^2}d\eta d\kappa \left\|\widehat{\Phi}_{\lambda} \right\|_{L^{2,1}}^2 \\
& =   \left\|\widehat{\Phi}_{\lambda} \right\|_{L^{2,1}}^2 \int_{\kappa}\int_{\eta} \frac{|\kappa-\eta|^2}{(1+\kappa^2)(1+\eta^2)}\chi\left(|\kappa|<\lambda^{r-1}\right)\chi\left(|\eta|<\lambda^{r-1}\right)d\kappa d\eta.
\end{align*}
One easily checks that
\[
\int_{\kappa}\int_{\eta} \frac{|\kappa-\eta|^2}{(1+\kappa^2)(1+\eta^2)}\chi\left(|\kappa|<\lambda^{r-1}\right)\chi\left(|\eta|<\lambda^{r-1}\right)d\kappa d\eta \ \lesssim\ \lambda^{r-1},
\]
so that one obtains eventually
\begin{equation}\label{est(IV)}
\left\| \chi\left(|\kappa|<\lambda^{r-1}\right) R_2[\widetilde\psi_\nr](\lambda\kappa) \right\|_{L^{2,-1}} \leq C \left\|\widehat{\Phi}_{\lambda} \right\|_{L^{2,1}}\ \lambda^{\frac{r+1}{2}} ,
\end{equation}
with $C=C\left(\sup_{|k|<\lambda^r}\| p_{b_*}(\cdot;k)\|_{L^{\infty}} , \big\Vert x V(x) \big\Vert_{L^1_x} \right)$.
\bigskip

Altogether,~\eqref{est(I)},~\eqref{est(II)},~\eqref{est(III)}, and~\eqref{est(IV)} yield the estimate of Proposition~\ref{prop:R_per_bound}.
\end{proof}

\subsection{Completion of the proof of Theorem~\ref{thm:per_result}}

We now prove Theorem~\ref{thm:per_result} by an application of Lemma~\ref{lem:technical} to equation~\eqref{eq:near_per_rescaled}, where the remainder is estimated in Proposition~\ref{prop:R_per_bound}.

\begin{proof}[Proof of Theorem~\ref{thm:per_result}] 
We seek $E^\lambda\equiv E_{b_*}(0)-\lambda^2\theta^2$ and $\psi^\lambda$ of the form
\begin{align*}
\psi^\lambda&= \psi_\nr + \psi_\fr = \mathcal{T}^{-1}\left\{\widetilde{\psi}_\nr(k)p_{b_*(x;k)}\right\} + \mathcal{T}^{-1}\left\{\sum_{b=0}^\infty\widetilde{\psi}_{\fr,b}(k)p_b(x;k)\right\} \\
&=\int_{-1/2}^{1/2} \left( \widetilde{\psi}_\nr(k)u_{b_*}(x;k)+\sum_{b=0}^\infty \widetilde{\psi}_{\fr,b}(k)u_b(x;k) \right) \ dk.
\end{align*}
where $\widetilde{\psi}_\nr$, $\widetilde{\psi}_\fr$ satisfy equations~\eqref{eq:near_per_prime}--\eqref{eq:far_per_prime}; see Section~\ref{sec:far_near_per}. 

By application of Proposition~\ref{prop:far_per_intermsof_near_per}, one has that $\psi_\fr$ is uniquely defined as a function of $\psi_\nr$ and $\lambda$, and that
$ \big\Vert \psi_\fr[\psi_\nr;\lambda]\big\Vert_{H^2} \ \leq \ \lambda^{1-2r}\big\Vert \psi_\nr\big\Vert_{L^2}.$
Then, defining $\widehat{\Phi}_{\lambda}$ as in~\eqref{phi2Phi}, one has
 \begin{equation}
 \widetilde{\psi}_\nr(k) = \frac{1}{\lambda}\widehat{\Phi}_{\lambda}\left(\frac{k}{\lambda}\right)\ =\ \frac{1}{\lambda}\widehat{\Phi}_{\lambda}\left(\kappa\right), \qquad k=\lambda\kappa.
 \label{phi2Phi_completion}
 \end{equation}
By Proposition~\ref{prop:R_per_bound},  the rescaled (from~\eqref{eq:near_per_prime}) near-frequency equation~\eqref{eq:near_per_rescaled} can be written as
 \begin{multline}\label{eq:near_per_rescaled_completion}
\left(\frac12\partial_k^2 E_{b_*}(0)\kappa^2 + \theta^2\right) \chi_{_{\lambda^{r-1}}}(\kappa)\widehat{\Phi}_{\lambda}(\kappa) +  \chi_{_{\lambda^{r-1}}}(\kappa) \left(\int_\RR |p_{b_*}(\cdot;0)|^2 V\right) \int_\RR
\chi_{_{\lambda^{r-1}}}(\eta)\widehat{\Phi}_{\lambda}(\eta)d\eta \\
= - \chi\left(|\kappa|<\lambda^{r-1}\right) \mathcal{R}\left(\widehat{\Phi}_{\lambda}\right)(\kappa),
\end{multline}
with $\big\Vert \mathcal{R}\big(\widehat{\Phi}_{\lambda}\big) \big\Vert_{L^{2,-1}} 
\ \le\ C\lambda^{\alpha(r)}\ \big\Vert \widehat{\Phi}_{\lambda} \big\Vert_{L^{2,1}}, $
and $\alpha(r)=\max(\frac{1}{2}-2r,2r,\frac{r+1}{2})$.

From now on, we set $r=1/8$, $\alpha=1/4$, which yield optimal estimates.
Applying Lemma~\ref{lem:technical} with $\beta=1-r=7/8$,
\begin{equation}
A\ =\ \frac{1}{8\pi^2}\partial_k^2 E_{b_*}(0)\ \text{ and }\ B\ =\ -\int_{-\infty}^{\infty} |u_{b_*}(x;0)|^2 V(x) dx\quad \Big(\textrm{assumed to be positive}\Big),
\label{A1BhatV_per}
\end{equation}
 we deduce that there exists a solution $\left( \theta^2, \widehat{\Phi}_{\lambda} \right)$ of the rescaled near-frequency equation~\eqref{eq:near_per_rescaled_completion}, satisfying
\begin{equation}\label{eq:bounds_rescaled_near_completion}
\|\widehat{\Phi}_{\lambda} - \widehat{f}_0\|_{L^{2,1}}\ \leq\ C\ \lambda^{\frac14} \qquad \text{ and } \qquad |\theta^2 - \theta_0^2|\ \leq\ C\ \lambda^{\frac14}\ . 
\end{equation} 
Here $\left( \theta_0^2, \widehat{f}_0 \right)$ is a solution of the homogeneous equation
\[
\widehat{\mathcal{L}} _{0,\lambda}\left( \theta_0,\widehat{f}_0 \right) = (4\pi^2 A \xi^2 + \theta^2) \widehat{f}_0 -B\ \chi\left(|\xi|<\lambda^{-\frac78}\right) \int_{-\infty}^{\infty}\chi\left(|\eta|<\lambda^{-\frac78}\right)\widehat{f}_0(\eta)d\eta = 0,
\]
as described in Lemma~\ref{lem:homogeneous}. 
Thus $\widetilde{\psi}_\nr(\xi) \ = \ \frac{1}{\lambda} \widehat{\Phi}_{\lambda}\left(\frac{\xi}{\lambda}\right)$ and $E^\lambda = E_{b_*}(0) - \lambda^2\theta^2(\lambda)$ are well-defined (and satisfy the Ansatz of Lemma~\ref{lem:R0}), and $\widetilde{\psi}_\fr$ is uniquely determined as the solution of~\eqref{eq:far_per_prime}; see Lemma~\ref{prop:far_per_intermsof_near_per}. It follows that
\begin{equation}
\psi^\lambda(x) \ \equiv \ \psi_\fr \ + \ \psi_\nr  \ \ \equiv \ \psi_\fr \ + \ \int_{-1/2}^{1/2} \widetilde{\psi}_\nr(k)u_{b_*}(x;k) \ dk \label{def-psi-lambda}
\end{equation} 
is well-defined.

There remains to prove estimates~\eqref{Elambda-Qper} and~\eqref{psilambda-Qper}.  Recalling that $E^\lambda\ = \ E_{b_*}(0) \ - \ \lambda^2\theta^2$,~\eqref{eq:bounds_rescaled_near_completion} implies
$
 \left| E^\lambda -  (E_{b_*}(0)  -  \lambda^2 \theta_0^2)  \right|  \leq  C \lambda^{2+1/4}
$.
By Lemma~\ref{lem:homogeneous}, one has 
$ \left|\theta_0(\lambda)-\frac{B}{2\sqrt{A}}\right| \ \leq \ C(A,B)\lambda^{\frac78},$
so that one can set
 \[ E_2\ \equiv\ -\frac{B^2}{4A} \ = \ -\ \frac{ \left| \int_{-\infty}^{\infty} |u_{b_*}(x;k_*)|^2 V(x)dx\right|^2}{\frac1{2\pi^2}\partial_k^2 E_{b_*}(k_*) };\]
 and estimate~\eqref{Elambda-Qper} follows.

We now turn to a proof of the eigenfunction approximation~\eqref{psilambda-Qper}.
Recall
\begin{align*}\psi_\nr(x) \ &\equiv \ \int_{-1/2}^{1/2} \widetilde{\psi}_\nr(k)u_{b_*}(x;k) \ = \ \int_{-1/2}^{1/2} \frac{1}{\lambda} \widehat{\Phi}_{\lambda}\left(\frac{k}{\lambda}\right)e^{2\pi i k x }p_{b_*}(x;k) \ dk \\
&= \  \int_{-1/2\lambda}^{1/2\lambda} \chi\left(|\xi|<\lambda^{-\frac78}\right)\widehat{\Phi}_{\lambda}\left(\xi\right)e^{2\pi i \lambda\xi x }p_{b_*}(x;\lambda\xi) \ d\xi\\
&= \int_{\RR} \chi\left(|\xi|<\lambda^{-\frac78}\right)\widehat{\Phi}_{\lambda}\left(\xi\right)e^{2\pi i \lambda\xi x }p_{b_*}(x;0) \ d\xi\\
& \qquad + \int_{\RR} \chi\left(|\xi|<\lambda^{-\frac78}\right)\widehat{\Phi}_{\lambda}\left(\xi\right)e^{2\pi i \lambda\xi x }(\lambda\xi)\partial_k p_{b_*}(x;k') \ d\xi\\
&= u_{b_*}(x;0) \int_{\RR} \chi\left(|\xi|<\lambda^{-\frac78}\right)\widehat{f}_0(\xi)e^{2\pi i \lambda\xi x } \ d\xi \\
&\qquad + u_{b_*}(x;0) \int_{\RR} \chi\left(|\xi|<\lambda^{-\frac78}\right)\left(\widehat{\Phi}_{\lambda}-\widehat{f}_0\right)(\xi)e^{2\pi i \lambda\xi x } \ d\xi \\
& \qquad + \int_{\RR} \chi\left(|\xi|<\lambda^{-\frac78}\right)\widehat{\Phi}_{\lambda}\left(\xi\right)e^{2\pi i \lambda\xi x }(\lambda\xi)\partial_k p_{b_*}(x;k') \ d\xi\\
&=I_1(x)+I_2(x)+I_3(x),
\end{align*}
with $|k'|=|k'(\lambda\xi)|<\lambda^{\frac18}$.
 Now, since $\chi\left(|\xi|<\lambda^{-\frac78}\right)\widehat{f}_0(\xi)=\widehat{f}_0(\xi)$, one has
\[ I_1(x) \ \equiv \ u_{b_*}(x;0) \mathcal{F}^{-1}\left\{\widehat{f}_0\right\}(\lambda x) .
\] 
By~\eqref{eq:asymptotic-f0} in Lemma~\ref{lem:homogeneous}, one has 
\begin{align} 
&\sup_{x\in\RR} \left\vert\ I_1(x) \ - \ \frac{1}{B}u_{b_*}(x;0)\exp\left(-\frac{\lambda B}{2A}|x|\right)\ \right\vert \nn  \\
&\qquad  = \ \sup_{x\in\RR} \left\vert\ u_{b_*}(x;0)\left\{\mathcal{F}^{-1}\left\{ \widehat{f}_0 \right\}(\lambda x) -\frac{1}{B}\exp\left(-\frac{\lambda B}{2A}|x|\right)\right\}\ \right\vert 
 \leq \ C \|p_{b_*}(\cdot;0)\|_{L^{\infty}} \ \lambda^{7/8}.
\label{eq:I1}\end{align}
and $ \|p_{b_*}(\cdot;0)\|_{L^{\infty}}$ is bounded; see Lemma~\ref{lem:estimates}.

Let us now estimate $I_2(x)$ and $I_3(x)$. One has 
\begin{align}
|I_2(x)| \ &\equiv \ \left| u_{b_*}(x;0) \int_{\RR} \chi\left(|\xi|<\lambda^{-\frac78}\right)\left(\widehat{\Phi}_{\lambda}-\widehat{f}_0\right)(\xi)e^{2\pi i \lambda\xi x } \ d\xi \right|\nn \\
&\leq | p_{b_*}(x;0) | \int_{\RR} \frac{ \chi\left(|\xi|<\lambda^{-\frac78}\right)}{(1+|\xi|^2)^{1/2}}(1+|\xi|^2)^{1/2}\left|\widehat{\Phi}_{\lambda}(\xi)-\widehat{f}_0(\xi)\right|\ d\xi \nn \\
& \leq \ C \| p_{b_*}(\cdot;0) \|_{L^{\infty}} \left\|\widehat{\Phi}_{\lambda}-\widehat{f}_0\right\|_{L^{2,1}} \leq \ C(A,B) \| p_{b_*}(\cdot;0) \|_{L^{\infty}} \lambda^{1/4},
\label{eq:I2}\end{align}
where the last inequality comes from~\eqref{eq:bounds_rescaled_near_completion}.
Similarly,
\begin{align}
|I_3(x)| \ &\equiv \ \left| \int_{\RR} \chi\left(|\xi|<\lambda^{-\frac78}\right)\widehat{\Phi}_{\lambda}\left(\xi\right)e^{2\pi i \lambda\xi x }(\lambda\xi)\partial_k p_{b_*}(x;k') \ d\xi\right| \nn \\
&\leq \lambda \sup_{|k'|<\lambda^{1-7/8}} \| \partial_k p_{b_*}(\cdot;k') \|_{L^{\infty}} \int_{\RR} \chi\left(|\xi|<\lambda^{-\frac78}\right)|\xi| \left|\widehat{\Phi}_{\lambda}(\xi)\right| d\xi \nn \\
& \leq \ C \sup_{|k'|<\lambda^{1/8}} \| \partial_k p_{b_*}(\cdot;k') \|_{L^{\infty}} \left\|\widehat{\Phi}_{\lambda}\right\|_{L^{2,1}} \ \lambda .
\label{eq:I3}\end{align}

By~\eqref{eq:I1},~\eqref{eq:I2} and~\eqref{eq:I3}, one has
\[ \psi_{\nr} \ = \ I_1(x)+I_2(x)+I_3(x) \ = \ \frac{2}{B}u_{b_*}(x;0)\exp\left(\frac{-\lambda B}{2A}|x|\right) \ + \ \psi_{\text{rem}}(x), \]
with $\|\psi_{\text{rem}}\|_{L^{\infty}} \lesssim \lambda^{1/4}$.

Finally, let us note that by Sobolev embeddings, one has
\[\|\psi_\fr \|_{L^\infty} \ \leq \ \| \psi_\fr \|_{H^2} \leq C\lambda^{1-1/4} \| \psi_\nr \|_{L^2}=C \lambda^{1/2-1/4} \| \widehat{\Phi}_{\lambda}\|_{L^2}\leq C \lambda^{1/4} ,
\]
where we use Proposition~\ref{prop:far_per_intermsof_near_per} with $r=1/8$, and~\eqref{eq:norm-equivalence}.

It follows that $\psi^\lambda=\psi_\nr+\psi_\fr$ satisfies
\[\sup_{x\in\mathbb R}\ \left|\ \psi^\lambda(x) \ - \ \frac1B u_{b_*}(x;0)\exp(\lambda\alpha_0|x|)\ \right| \ \leq \ C \lambda^{1/4} \ , \quad \text{with}\ \alpha_0=-\frac{B}{2A} . \]
Since $\psi^\lambda$ is defined up to a multiplicative constant,~\eqref{psilambda-Qper} holds.
This completes the proof of Theorem~\ref{thm:per_result}.
\end{proof}

\appendix

\section{General properties of $E_b(k)$ and derivatives $\partial_k^jE_b(k_*)$, where $E_b(k_*)$ is the endpoint of a spectral band}\label{pf-of-lem:band-edge}

To make our discussion more self-contained, we prove Lemma~\ref{lem:band-edge}, which concerns the spectrum of the eigenvalue problem, for $E$ fixed, 
\begin{equation}\label{eq:e-value-problem}
\left( -\partial_x^2 + Q(x) \right)\psi(x;E) = E\psi(x;E), \quad Q(x+1) = Q(x), 
\end{equation}
with solutions which satisfy 
\[
\psi(x+1;E) = \rho \psi(x;E)\ \qquad \rho\in\CC.
\]

Let $\phi_1(x;E)$ and $\phi_2(x;E)$ be two linearly independent solutions of~\eqref{eq:e-value-problem} such that 
\begin{align*}
\phi_1(0;E) &= 1, &  \phi_2(0;E) &= 0, \\
\phi_1'(0;E) &= 0, &  \phi_2'(0;E) &= 1.
\end{align*}
The functions $\phi_1(x+1;E)$ and $\phi_2(x+1;E)$ are two other linearly independent solutions to~\eqref{eq:e-value-problem}, so that we can write 
\begin{align}
\phi_1(x+1;E) & = A_{11}\phi_1(x;E) + A_{12}\phi_2(x;E), \label{eq:phi-1-per} \\
\phi_2(x+1;E) & = A_{21}\phi_1(x;E) + A_{22}\phi_2(x;E). \label{eq:phi-2-per}
\end{align}
Note that the matrix $\left( A_{ij} \right)$ is nonsingular. In general, every solution of~\eqref{eq:e-value-problem} has the form 
\begin{equation}
\psi(x;E) = c_1 \phi_1(x;E) + c_2 \phi_2(x;E). 
\end{equation}
As we are specifically interested in solutions which satisfy $\psi(x+1;E) = \rho \psi(x;E)$, one has the following identity
\begin{align}
\psi(x+1;E) = \rho\psi(x;E) \Leftrightarrow & \ c_1(\phi_1(x+1;E)-\rho\phi_1(x;E)) + c_2 (\phi_2(x+1;E)-\rho\phi_2(x;E)) = 0 \nonumber \\ 
 \Leftrightarrow & \left( c_1 \left(A_{11} - \rho \right) + c_2 A_{21} \right) \phi_1(x;E) + \left( c_1 A_{12} + c_2 \left( A_{22} - \rho \right) \right) \phi_2(x;E) = 0 \nonumber \\
 \Rightarrow & 
\begin{cases}
c_1 \left(A_{11} - \rho \right) + c_2 A_{21} = 0 ,\\
c_1 A_{12} + c_2 \left( A_{22} - \rho \right) = 0.
\end{cases} \label{cond:solvability}
\end{align}
The solvability condition~\eqref{cond:solvability} is satisfied for nontrivial $c_1$ and $c_2$ if 
\begin{equation}
\det(A - \rho I) = 0, \quad \text{i.e.}\quad \rho^2 - \left( A_{11} + A_{22} \right)\rho + \det(A) = 0.
\end{equation}
Using that the Wronskian, $W\left[\phi_1,\phi_2\right](x;E)\equiv \phi_1(x;E)\phi_2'(x;E)-\phi_1'(x;E)\phi_2(x;E)$, is constant with respect to $x$, one has 

\[
\det(A) = W\left[\phi_1,\phi_2\right](1;E) =W\left[\phi_1,\phi_2\right](0;E)= 1.
\]

Therefore $\rho$ must satisfy $\rho^2 - D(E)\rho + 1 = 0$, where we define the discriminant
\begin{equation}\label{eq:D-def}
D(E)\ \equiv \ A_{11} +A_{22}\ =\ \phi_1(1;E)+\phi_2'(1;E) \ .
\end{equation}

We note that the two solutions of the equation $\rho^2 - D(E)\rho + 1 = 0$ satisfy $|\rho|\leq1$ if and only if the discriminant $|D(E)|\leq2$. In that case, one can write $\rho = e^{\pm 2\pi ik}$, with $k\in (-1/2, 1/2]$, and 
\begin{equation}\label{eq:D-behaviour}
D(E) = 2\cos(2\pi k).
\end{equation}
As $|\rho|=1$, $\psi(x;E)$ is a bounded solution to~\eqref{eq:e-value-problem}, and $E=E_b(k)$ is in the continuous spectrum of $H_Q\equiv -\frac{d^2}{dx^2}+Q$. More precisely, for $E=E_b(k)$, one has
\[ \psi(x;E_b(k)) \ =\ u_b(x;k) \ = \ e^{2\pi ikx}p_b(x;k), \quad p_b(x+1;k) = p_b(x;k),\]
where $\left\{ E_b(k), p_b(x;k) \right\}_{b\geq 0}$ is the eigenpair solution to~\eqref{eq:e-value}, as defined in Section~\ref{sec:background}.

Let us now rewrite Lemma~\ref{lem:band-edge} which states some of the properties associated with the stability bands.

\begin{Lemma}[Lemma~\ref{lem:band-edge}]\label{lem:band-edge-in Appendix}
Assume $E_b(k_*)$ is an endpoint of a spectral band of $-\partial_x^2 + Q(x)$, which borders on a spectral gap;
 see \eqref{whatisgap}.
%%%
%%%
Then $k_*\in\{0,1/2\}$ and the following results hold:
\begin{enumerate}
\item $E_b(k_*)$ is a simple eigenvalue of the eigenvalue problem~\eqref{k-evp}.
\item \subitem $b$ even: $E_b(0)$ corresponds to the {\em left (lowermost)} end point of the band,
\subitem \phantom{$b$ even:} $E_b(1/2)$ corresponds to the {\em right (uppermost)} end point.
\subitem $b$ odd: $E_b(0)$ corresponds to the {\em right (uppermost)} end point of the band,
\subitem \phantom{$b$ odd:} $E_b(1/2)$ corresponds to the {\em left (lowermost)}  end point.
\item $\partial_k E_b(k_*) = 0$,\ $\partial_k^3E_b(k_*) = 0$;
\item \subitem $b$ even: $\partial_k^2 E_b(0) > 0$, $\partial_k^2 E_b(1/2) < 0$;
\subitem $b$ odd: $\partial_k^2 E_b(0) < 0$, $\partial_k^2 E_b(1/2) > 0$; 
\end{enumerate}
\end{Lemma}

The proof of Lemma~\ref{lem:band-edge-in Appendix} is a consequence of the following result, concerning the problem~\eqref{eq:e-value-problem}, and which is proved in the first two chapters of~\cite{eastham1973spectral} and part I of~\cite{magnus1979hill}.
\begin{figure}[htp]
 \begin{center}
 \includegraphics[width=0.8\textwidth]{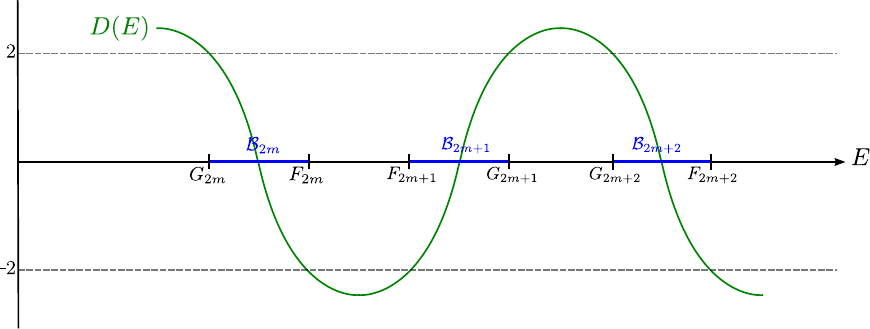}
 \end{center}
 \caption{Sketch of the discriminant, $D(E)$, and stability bands $\mathcal B_b=[G_b,F_b]$.}
\label{fig:discriminant}
 \end{figure}
\begin{Theorem} \label{thm:eastham}
Consider the equation~\eqref{eq:e-value-problem}, and define $D(E)$ with~\eqref{eq:D-def}. Denote the edges of the stability bands as
\[
G_0 < F_0 \leq F_1 < G_1 \leq G_2 < F_2 \leq F_3 < G_3 \dots 
\]
Then the following facts hold (see Figure~\ref{fig:discriminant} for an illustration):
\begin{itemize}
\item[$\mathrm{I}$] In the interval $[G_{2m},F_{2m}]$, $D(E)$ decreases from $2$ to $-2$.

\item[$\mathrm{I'}$] In the interval $(G_{2m},F_{2m})$, $D'(E)<0$.

\item[$\mathrm{II}$] In the interval $[F_{2m+1},G_{2m+1}]$, $D(E)$ increases from $-2$ to $2$.

\item[$\mathrm{II'}$] In the interval $(F_{2m+1},G_{2m+1})$, $D'(E)>0$.

\item[$\mathrm{III}$] In $(-\infty, G_0)$ and $(G_{2m+1},G_{2m+2})$, $D(E) > 2$.

\item[$\mathrm{IV}$] In $(F_{2m},F_{2m+1})$, $D(F) < -2$.

\item[$\mathrm{V}$] $D(E) = \pm 2$ and $D'(E) = 0$ if and only if $E$ is a double eigenvalue. Furthermore, $D''(E)<0$ if $D(E)=2$ and $D''(E)>0$ if $D(E)=-2$. 

\end{itemize}

\end{Theorem}

\noindent{\em Proof of Lemma~\ref{lem:band-edge-in Appendix}.} Let us recall that one has from~\eqref{eq:D-behaviour} that the discriminant satisfies
$
D(E_b(k)) = 2\cos(2\pi k).
$
It follows that as $k$ increases continuously from $0$ to $1/2$, $D(E)$ decreases continuously from 2 to -2. Therefore by $\mathrm{I}$ and $\mathrm{II}$, $E_{2m}(k) $ increases continuously from $G_{2m}$ to $F_{2m}$ as $k$ increases continuously from $0$ to $1/2$, and as $k$ decreases continuously from $0$ to $-1/2$. Similarly, $E_{2m+1}(k) $ decreases continuously from $G_{2m+1}$ to $F_{2m+1}$ as $k$ increases continuously from $0$ to $1/2$, and as $k$ decreases continuously from $0$ to $-1/2$. This proves claim $\mathit{2}$. 

We now turn to part $\mathit{1}$. Let $E_b(k_*)$ correspond to a band edge, that is if there exists a gap between the $b^{th}$ band and the closest consecutive one. Without loss of generality, we assume $E_b(k_*)$ to be the lowermost edge of an even band, for example $G_{2m}$ in Figure~\ref{fig:discriminant}. Therefore, for any $\delta > 0$ sufficiently small,
\begin{equation}\label{eq:disc-reqmnts}
D(E_b(k_*) - \delta) > 2 \ \text{ and } \ D(E_b(k_*) + \delta) < 2.
\end{equation}
Assume for the sake of contradiction that $E_b(k_*)$ is a double eigenvalue, which means, by part V of Theorem~\ref{thm:eastham}, that $D'(E_b(k_*)) = 0$ and $D''(E_b(k_*)) < 0$. Now, Taylor expand the discriminant about $E_b(k_*)$,
\begin{align*}
D(E) & = D(E_b(k_*)) + D'(E_b(k_*))(E-E_b(k_*)) + \frac{1}{2}D''(E_b(k_*))(E-E_b(k_*))^2 + \mathcal{O}\left( (E-E_b(k_*))^3 \right) \\
& = -2 + \frac{1}{2}D''(E_b(k_*))(E-E_b(k_*))^2 + \mathcal{O}\left( (E-E_b(k_*))^3 \right).
\end{align*} 
Since $D''(E_b(k_*)) < 0$, we have $D(E_b(k_*) - \delta) \approx 2 + (1/2)D''(E_b(k_*))\delta^2 < 2$, which is a contradiction of~\eqref{eq:disc-reqmnts}. Therefore part $\mathit{1}$ is proven and we have 
that
at the band edges, $E_b(k_*)$, the derivative of the discriminant is nonzero,
\begin{equation}\label{eq:D-deriv-not-zero}
\frac{dD}{dE} (E_b(k_*)) \neq 0.
\end{equation}

To see the first identity in part $\mathit{3}$, note that differentiating $D(E_b(k)) = 2\cos(2\pi k)$ with respect to $k$ yields
$
-4\pi \sin(2\pi k) = \frac{dD}{dE}(E_b(k))  \times  \frac{dE_b}{dk}(k).
$
Using~\eqref{eq:D-deriv-not-zero}, we conclude that $\frac{dE_b}{dk}(k) = 0$ if and only if $k = 0$ or $k=1/2$.

To prove part $\mathit{4}$, we differentiate $D(E_b(k))$ twice with respect to $k$ and evaluate at $k_*$:
\begin{equation*}
-8\pi^2 \cos(2\pi k_*) = \frac{d^2}{dk^2}(D\circ E_b)(k_*) = D''(E_b(k_*))\left( \frac{dE_b}{dk} \right)^2 (k_*) + D'(E_b(k_*))\frac{d^2E_b}{dk^2}(k_*) .
\end{equation*}
Therefore, by $\mathrm{I'}$, $\mathrm{II'}$, and~\eqref{eq:D-deriv-not-zero} we conclude $\mathit{4}$. 

Similarly, to show the second identity of part $\mathit{3}$, we differentiate once more $D(E_b(k))$ with respect to $k$:
\[
16\pi^3\sin(2\pi k)  = D'(E_b(k))\frac{d^3E_b}{dk^3}(k) + 3 D''(E_b(k))\frac{d^2E_b}{dk^2}\frac{d\ E_b}{dk}(k)+ D'''(E_b(k))\left( \frac{dE_b}{dk} \right)^3(k).
\]
Evaluated at $k_*$, we have
$
0 = D'(E_b(k))\frac{d^3E_b}{dk^3}(k) ,
$
which concludes the proof of Lemma~\ref{lem:band-edge} once we again note~\eqref{eq:D-deriv-not-zero}.
\qed

\section{Regularity of $k\mapsto E_b(k)$ and $k\mapsto u_b(x;k)$}\label{regularity-app}

In this section we give a self-contained discussion of the regularity with respect to $k$ of the Floquet-Bloch eigenvalues and eigenstates. 

Consider the $k-$ pseudo-periodic eigenvalue problem for each $k\in(-1/2,1/2]$:
\begin{align}
&\left(-\partial_x^2+Q(x)\right)u(x;k) =\ E u(x;k),\ \ u(x+1;k)=e^{2\pi ik}u(x;k)
\label{k-pseudo}\end{align}
Introducing the Floquet-Bloch phase explicitly via
$u(x;k)\ =\ e^{2\pi ikx}\ p(x;k), $
we obtain the equivalent formulation
\begin{align}
&H_Q(k) p(x;k)\ =\left(-(\partial_x+2\pi ik)^2+Q(x)\right)p(x;k) =\ E p(x;k),\ \ p(x+1;k)=p(x;k)\ .
\label{k-per}
\end{align}
For each $k\in(-1/2,1/2]$, the eigenvalue problem ~\eqref{k-per} (equivalently~\eqref{k-pseudo}) has a discrete sequence of eigenvalues
$E_0(k)\le E_1(k)\le E_2(k) \le \cdots\le E_n(k)\le\cdots\ \ .$

It can be proved, using the min-max characterization of eigenvalues of a self-adjoint operator that 
the maps $k\mapsto E_b(k),\ b=0,1,\dots$, are locally Lipschitz continuous. A proof based on standard perturbation follows from results in~\cite{RS4}. An elementary proof is given in Appendix A of~\cite{FW-dirac}.

In the present paper, we require a Taylor expansion of the $E_b(k)$ near $k=k_*$, for which $E_b(k_*)$ is the endpoint of a spectral band, which borders on a spectral gap. By part $\mathit{5}$ of Lemma~\ref{lem:band-edge-in Appendix}, the eigenvalue $E_b(k_*)$ is simple. We prove the following
\begin{Theorem}\label{reg-in-k}
Suppose $E_*$ is the endpoint of a spectral band of $-\partial_x^2+Q(x)$, which borders on a gap.
 Thus, $E_*=E_b(k_*)$ for $k_*\in\{0,1/2\}$ and the corresponding eigenspace of solutions to~\eqref{k-per} has dimension equal to $1$. We denote the normalized eigenfunction by $p(x;k_*)$;
 \[ \int_0^1|p(y;k_*)|^2 dy\ =\ 1.\]
 Then, there exists $\rho>0$ such that for all complex $k$ in a complex disc centered at $k_*$, $B_\rho(k_*)=\{k\in\mathbb{C}: |k-k_*|<\rho\}$, the following holds:
 \begin{enumerate}
 \item $k\mapsto E_b(k)$ is analytic on $B_\rho$.
 \item There is a map $k\mapsto p_b(x;k)$, such that any eigenvector corresponding to $E_b(k)$ is a multiple of $p_b(x;k)$, where $H_Q(k)p_b(x;k)=E_b(k)p_b(x;k)$.
 \item Moreover, we can choose $k\mapsto p_b(x;k),\ k\in B_\rho$ to be analytic and such that 
\[\int_0^1 | p_b(x;k)|^2 dx\ =\ 1.\]
 \end{enumerate}
 \end{Theorem} 
 
 \noindent {\it Proof of Theorem~\ref{reg-in-k}:} Let $k=k_*+\kappa$, where $\kappa$ will be chosen to be sufficiently small. The periodic eigenvalue problem~\eqref{k-per} may be rewritten as
 \begin{align}
&H_Q(k_*)p(x;k_*+\kappa)\ -\ \left(4\pi i\kappa(\partial_x+2\pi ik_*)-4\pi^2\kappa^2\right)p(x;k_*+\kappa)\ =\ E\ p(x;k_*+\kappa)\label{k-per1a},\\
& p(x+1;k_*+\kappa)=p(x;k_*+\kappa)\ ,\ x\in\mathbb R\ .
\label{k-per1b}
\end{align}
We seek an eigen-solution of~\eqref{k-per1a}-\eqref{k-per1b} in the form
\begin{align}
p(x;k_*+\kappa)\ &=\ p(x;k_*)+\eta(x;\kappa),\label{eta-def}\\
E(k_*+\kappa)\ &=\ E_*+ \mu(\kappa),\label{E-def}
\end{align}
where we assume that $\eta(\cdot;\kappa) \perp p(\cdot;\kappa)$. Substitution into~\eqref{k-per1a}-\eqref{k-per1b} yields the following equation for $\eta(x;\mu,\kappa)$:
\begin{equation}
\left(H_Q(k_*)-E_*\right)\eta\ -\ \left(4\pi i\kappa(\partial_x+2\pi ik_\star)-4\pi^2\kappa^2+\mu\right)\eta
= \left(4\pi i\kappa(\partial_x+2\pi ik_\star)-4\pi^2\kappa^2+\mu\right) p(\cdot;k_*)\ .
\label{eta-eqn}\end{equation}
Now, introduce the projection operators
\begin{equation}
\Pi f=\left\langle p(\cdot;k_*),f\right\rangle p(x;k_*),\ \ {\rm and}\ \ \Pi_\perp=I-\Pi.\nn
\end{equation}

Applying $\Pi_\perp$ to~\eqref{eta-eqn} yields
\begin{equation}
\left(H_Q(k_*)-E_*\right)\eta\ -\ \Pi_\perp\left(4\pi i\kappa(\partial_x+2\pi ik_\star)-4\pi^2\kappa^2+\mu\right)\eta
\ =\ 4\pi i \kappa\ \Pi_\perp\partial_x p(\cdot;k_*) = 4\pi i\kappa\partial_x p(\cdot;k_*).
\label{Q-perp-project}
\end{equation}
Next, applying $\Pi$ to~\eqref{eta-eqn}, {\it i.e.} taking the inner product of~\eqref{eta-eqn} with $p(\cdot;k_*)$, yields
\begin{equation}
\mu-8\pi^2 k_*\kappa-4\pi^2\kappa^2\ +\ 4\pi i\ \kappa\ \left\langle p(\cdot;k_*),\partial_x\eta(\cdot;\mu,\kappa)\right\rangle\ =\ 0. \label{Q-project}
\end{equation}
We shall now solve~\eqref{eta-eqn} for $\eta$, substitute the result into~\eqref{Q-project} and obtain a closed equation for the eigenvalue correction $\mu=\mu(\kappa)$. 
Let 
\begin{equation}
\mu\ =\ \kappa \mu_1,\ \ \eta=\ \kappa\eta_1.
\label{mu1eta1}
\end{equation}
Equations~\eqref{Q-perp-project} and~\eqref{Q-project} become
\begin{align}
&\left(H_Q(k_*)-E_*\right)\eta_1\ -\ \kappa \Pi_\perp\left(4\pi i(\partial_x+2\pi ik_\star)-4\pi^2\kappa +\mu_1\right)\eta_1 \ =\ 4\pi i \Pi_\perp\partial_x p(\cdot;k_*),
\label{Q-perp-project1}\\
&\mu_1-8\pi^2 k_*-4\pi^2\kappa\ +\ 4\pi i\kappa\left\langle p(\cdot;k_*),\partial_x\eta_1(\cdot)\right\rangle\ =\ 0 .\label{Q-project1}
\end{align}
Let $R(E_*)\Pi_\perp=(H_Q(k_*)-E_*)^{-1} \Pi_\perp$. 
Then, 
\begin{align}
\eta_1(x;\mu_1,\kappa)\ &=\ 4\pi i \ \Big(I-\kappa R(E_*)\Pi_\perp \left(4\pi i(\partial_x+2\pi ik_\star)-4\pi^2\kappa+\mu_1\right) \Big)^{-1}\
 \ R(E_*)\ \Pi_\perp \partial_x p(\cdot;k_*)\ ,
 \label{eta1-solved}
 \end{align}
 where we take $|\kappa|<\rho$, with $\rho$ chosen so that the Neumann series for the operator on the right hand side of~\eqref{eta1-solved} converges. Note that the mapping
\[ (\mu_1,\kappa)\mapsto \eta_1(x;\mu_1,\kappa) \]
is an analytic map from $\{(\mu_1,\kappa): |\mu_1|<1,\ |\kappa|<\rho'\}$ to $H^2_{\rm per}(\mathbb R)$.

 Subsitution of~\eqref{eta1-solved} into~\eqref{Q-project1} gives the scalar equation
 \begin{equation}
 \mathcal{G}(\mu_1,\kappa)\ =\ 0\ ,
 \label{bif-eqn}
 \end{equation}
 where 
 \begin{align}
 \mathcal{G}(\mu_1,\kappa)\ &=\ 
 \mu_1-8\pi^2 k_*-4\pi^2\kappa\ +\ 4\pi i\kappa\left\langle p(\cdot;k_*),\partial_x\eta_1(\cdot;\mu_1,\kappa)\right\rangle\ .
\label{Gdef}\end{align}

We now claim that~\eqref{bif-eqn} can be solved for $\mu_1=\mu_1(\kappa)$, which is defined and analytic for $ |\kappa|<\rho'$, where $0<\rho'\le\rho$. If this claim is valid, then $\eta_1(x;\mu_1(\kappa),\kappa)$ is well-defined and analytic in $\kappa$ for $|\kappa|<\rho'$ and finally 
\begin{align}
p(x;k_*+\kappa)\ &=\ p(x;k_*)+\kappa\eta_1(x;\mu_1(\kappa),\kappa)\label{bloch-kappa}\\
E(k_*+\kappa)\ &=\ E_*+ \kappa\mu_1(\kappa)\label{E-kappa}
\end{align}
are defined and analytic Floquet-Bloch eigensolutions for $|\kappa|<\rho'$.

Now~\eqref{bif-eqn} is easily solved for $\mu_1=\mu_1(\kappa)$ via the Implicit Function Theorem.
Indeed, we have $\mathcal{G}(8\pi^2 k_*,0)=0$ and $\left.\partial_{\mu_1}\mathcal{G}(\mu_1,\kappa)\right|_{(8\pi^2 k_*,0)}=1\ne0$.
This completes the proof of Theorem~\ref{reg-in-k}.\qed

\section{The bootstrap: proof of Corollary~\ref{cor:Elambdaprecise-per} }\label{proof-of-Corollary-per}
We give the proof of Corollary~\ref{cor:Elambdaprecise-per}, on the refined expansion of the bifurcation of eigenvalues of $H_Q+\lambda V=-\partial_x^2+Q(x)+\lambda V(x)$, for $Q(x)$ periodic.  Corollary~\ref{cor:Elambdaprecise}, in the case of $Q(x)\equiv0$, is obtained along the same lines, using $p(x;k)=1$, $E(k)\equiv 4\pi^2k^2$ for $k\in\RR$, {\em etc.}

\noindent{\em Proof of Corollary~\ref{cor:Elambdaprecise-per}.} We know, by Theorem~\ref{thm:per_result}, that there exists $(\psi^\lambda,E^\lambda)$ solution of the eigenvalue problem
$
(H_Q+\lambda V)\psi^\lambda= E^\lambda \psi^\lambda.
$
Moreover, 
$E^\lambda$ is in the gap of the continuous spectrum of $\spec(H_{Q})=\spec(H_{Q+\lambda V})$, near an edge $E_*=E_{b_*}(k_*)$. In the following, we assume that $k_*=0$ (the case where $k_*=1/2$ can be treated using the same method). 

We next seek an integral equation for $\psi^\lambda$ by applying the resolvent $R_Q(E^\lambda)$ to the differential equation for $\psi^\lambda$. A construction of the resolvent kernel, $R_Q(x,y;E^\lambda)$, proceeds as follows. Recall the discriminant, $D(E)$, introduced in Appendix~\ref{pf-of-lem:band-edge} as the trace of the monodromy matrix defined by the linearly independent solutions $\phi_j(x;E), j=1,2$: $D(E)=\phi_1(1;E)+\phi_2'(1;E)$.

 Since $E_*$ is a band edge and $E^\lambda$ is in a gap, we have  $D(E_*)=2$ and $D(E^\lambda)>2$. Therefore,  there exists $\kappa=\kappa(\lambda)>0$ with
\[E^\lambda=E(i\lambda\kappa)=E(-i\lambda\kappa),\quad D(E^\lambda)=e^{2\pi\lambda\kappa}+e^{-2\pi\lambda\kappa}>2,\]
Additionally, we define $\psi_\pm\equiv\psi_\pm(x;E^\lambda)$, the solutions of
\[
\left( -\partial_x^2 + Q(x) \right)\psi_\pm = E^\lambda \psi_\pm, 
\qquad 
\psi_\pm(x+1;E^\lambda) = e^{\pm 2\pi \lambda\kappa}\psi_\pm(x;E^\lambda).
\]
More precisely, $\psi_\pm$ are defined as
\begin{align}
&\psi_\pm(x) \ \equiv \ p_{b_*}(x;\mp i\kappa)e^{\pm2\pi\lambda\kappa x}, \quad \text{with} \label{def-psi}\\
&\left(-(\partial_x-2\pi\lambda\kappa)^2+Q(x)\right)p_{b_*}(x;i\kappa) =\ E^\lambda p_{b_*}(x;i\kappa)\ , \qquad p_{b_*}(x+1;i\kappa) \ = \ p_{b_*}(x;i\kappa).\label{eqn-pb}
\end{align}
which is well-defined for $\lambda$ small enough, by Theorem~\ref{reg-in-k}.

With those definitions, the resolvent operator $ R_Q(E^\lambda)=(-\partial_x^2 +Q - E^\lambda)^{-1}$ has kernel 
\[R_Q(x, y;E^\lambda) = \left\{ \begin{array}{cc} \dfrac{\psi_+(x)\psi_-(y)}{W[\psi_\pm]} & \text{ if } y>x, \\
\dfrac{\psi_+(y)\psi_-(x)}{W[\psi_\pm]} & \text{ if } y<x. \end{array} \right.\]
where  $W[\psi_\pm]\equiv\psi_+'(x)\psi_-(x)-\psi_+(x)\psi_-'(x)$. Thus, for any bounded function $f$,
\[R_Q[f](x;E^\lambda)= \int_{-\infty}^\infty R_Q(x, y;E^\lambda) f(y)\ dy,\]
we have
$ (-\partial_x^2 +Q - E^\lambda)R_Q[f](x;E^\lambda) \ = \ f$. 
 It follows that $\psi^\lambda$ satisfies the integral equation
\[ \psi^\lambda(x)+\lambda \int_\RR R_Q(x,y;E^\lambda) V(y) \psi^\lambda(y)\ dy \ = \ 0.\]
Multiplying by $u_{b_*}(x;0)V(x)$ and integrating along $x$ yields
\begin{equation}\label{eqn:intVpsi}
\int_\RR V(x) u_{b_*}(x;0)\psi^\lambda(x)dx+\lambda\iint_{\RR^2} u_{b_*}(x;0) V(x)R_Q(x,y;E^\lambda) V(y) \psi^\lambda(y)\ dx\ dy \ = \ 0.
\end{equation}
We will deduce from~\eqref{eqn:intVpsi} the precise behavior of $\kappa$ (and therefore $E^\lambda-E_{b_*}(0)$) as $\lambda$ tends to zero, using the following
\begin{Lemma}\label{lem:conj}
Let $E_*=E_{b_*}(0)$ be an edge of the continous spectrum, and let the hypotheses of Theorem~\ref{thm:per_result} be satisfied, so that $E^\lambda$ exists. Define $R_Q(x,y;E^\lambda) $ as above. Then for $\lambda>0$ small enough, one has
\begin{equation}\label{expansion-RQ}
R_Q(x,y;E^\lambda) \ = \ \frac{u_{b_*}(x;0)u_{b_*}(y;0)}{2\lambda\kappa\frac{\partial_k^2 E(0)}{4\pi}}e^{-2\pi\lambda\kappa|x-y|}+R_Q^{(0)}(x,y)+\lambda\kappa R_Q^{(1)}(x,y),
\end{equation} 
where $R_Q^{(0)}$ is skew-symmetric: $R_Q^{(0)}(x,y)=-R_Q^{(0)}(y,x)$; and $R_Q^{(0)},R_Q^{(1)}$ are bounded:
\begin{align*}
 |R_Q^{(0)}(x,y)| +  |R_Q^{(1)}(x,y)|\ &\leq \ Ce^{-2\pi\lambda \kappa|x-y|} \ \leq\ C \ ,
\end{align*}
where $C$ is a constant, uniform with respect to $\lambda\kappa$.
\end{Lemma}
In order to ease the reading, we postpone the proof of this result to the end of this section, and carry on with the proof of Corollary~\ref{cor:Elambdaprecise-per}.
 By Lemma~\ref{lem:conj}, and since $u_{b_*}(x;0)$ is uniformly bounded (see Lemma~\ref{lem:estimates}), one has the low-order estimate
\begin{equation}\label{est0}\left\vert\ R_Q(x,y;E^\lambda) \ - \ \frac{4\pi}{\partial_k^2 E(0)}\frac{u_{b_*}(x;0)u_{b_*}(y;0)}{2\lambda\kappa} \ \right\vert \ \leq \ C(1+|x-y|+\lambda\kappa), \end{equation}
where we used 
$ \left\vert\ e^{-\lambda\kappa|x-y|} \ - \ 1 \ \right\vert \ \leq \ C\lambda\kappa|x-y|.$

Plugging~\eqref{est0} into~\eqref{eqn:intVpsi} and using $(1+|x|)V\in L^1$, yields
\begin{equation}\label{estInt0}
\Big|\ \int_\RR V(x) u_{b_*}(x;0)\psi^\lambda(x)dx+\frac{2\pi}{\kappa\partial_k^2 E(0)}\iint_{\RR^2} u_{b_*}(x;0)^2 V(x)u_{b_*}(y;0) V(y) \psi^\lambda(y)\ dx\ dy \ \Big| 
\leq \ C\lambda(1+\lambda\kappa).
\end{equation}
Now we use the fact that by Theorem~\ref{thm:per_result}, one has $\big\| \psi^\lambda(x) \ - \ u_{b_*}(x;0)\exp\big(\lambda\alpha_0|x|\big) \big\|_{L^\infty} \ \lesssim \ \lambda^{1/4}$, so that $\lim\limits_{\lambda\to0}\int V(x)u_{b_*}(x;0) \psi^\lambda(x)dx= \int V(x)u_{b_*}(x;0)^2\neq0$. It follows that for $\lambda$ sufficiently small, one can divide out $\int V(x)u_{b_*}(x;0) \psi^\lambda(x)dx$, and deduce from~\eqref{estInt0}
\[
\Big|\ \kappa+\frac{2\pi}{\partial_k^2 E(0)}\int_{\RR} u_{b_*}(x;0)^2 V(x)\ dx \ \Big| 
\leq \ C\lambda\kappa(1+\lambda\kappa),
\]
from which it follows the low-order estimate of $\kappa$:
\begin{equation}\label{estkappa0}\Big|\ \kappa+\frac{2\pi}{\partial_k^2 E(0)}\int_{\RR} u_{b_*}(x;0)^2 V(x)\ dx \ \Big| 
\leq \ C\lambda.
\end{equation}

Let us now derive higher order estimates.
For any $x,y\in\RR^2$,
$\left| e^{-2\pi\lambda\kappa|x-y|} - 1 + 2\pi\lambda\kappa|x-y| \right| \ \leq 4\pi^2\lambda^2\kappa^2|x-y|^2 ,\ $
so that one has from Lemma~\ref{lem:conj},
\begin{equation}\label{est1}\left\vert\ R_Q(x,y;E^\lambda) -  \frac{2\pi}{\partial_k^2 E(0)}\frac{u_{b_*}(x;0)u_{b_*}(y;0)(1-2\pi\lambda\kappa|x-y|)}{\lambda\kappa}  -  R_Q^{(0)}(x,y) \ \right\vert \ \leq \ C\lambda(1+|x|^2+|y|^2).\end{equation}

Plugging~\eqref{est1} into~\eqref{eqn:intVpsi}, and using $(1+|x|)V\in L^1$, yields
\begin{multline}\label{estInt1} \Big| \int u_{b_*}(x;0) V(x) \psi^\lambda(x)dx+\frac{2\pi}{\partial_k^2 E(0)}\frac{1}{\kappa}\iint V(x)u_{b_*}(x;0)^2(1-2\pi\lambda\kappa|x-y|)u_{b_*}(y;0)V(y) \psi^\lambda(y)\ dx\ dy \\
+\frac{\lambda}{2}\iint V(x)u_{b_*}(x;0)R_Q^{(0)}(x,y)V(y) \psi^\lambda(y)\ dx\ dy \Big| \ \leq \ C\lambda^2.\end{multline}

Let us now use that by Theorem~\ref{thm:per_result}, $\sup_{x\in\RR} |\psi^\lambda(x) \ - \ u_{b_*}(x;0)\exp\big(\lambda\alpha_0|x|\big) | \ \lesssim \ \lambda^{1/4}$, so 
$ |\psi^\lambda(x) \ - \ u_{b_*}(x;0) | \ \leq \ C(\lambda^{1/4}+\lambda|x|). $
Thus~\eqref{estInt1} becomes
\begin{multline} \label{estInt1a}\Big| \left(\int u_{b_*}(\cdot;0) V \psi^\lambda\right)\left(1+\frac{1}{\kappa}\frac{2\pi}{\partial_k^2 E(0)}\int V(x)u_{b_*}(x;0)^2\ dx\ \right)\\ -\lambda\frac{4\pi^2}{\partial_k^2 E(0)}\iint V(x)u_{b_*}(x;0)^2|x-y|u_{b_*}(y;0)^2V(y) \ dx\ dy \\
+\frac{\lambda}{2}\iint V(x)u_{b_*}(x;0)R_Q^{(0)}(x,y)V(y) u_{b_*}(y;0)\ dx\ dy \Big| \ \leq \ C\lambda^{1+1/4},\end{multline}
and one deduces from~\eqref{estkappa0} that
$\left\vert\ \kappa \left(\int u_{b_*}(\cdot;0) V \psi^\lambda\right)^{-1} \ + \ 2\pi(\partial_k^2 E(0))^{-1} \ \right\vert \ \leq \ C\lambda^{1/4}.$
Therefore, multiplying~\eqref{estInt1a} by $\kappa \left(\int u_{b_*}(\cdot;0) V \psi^\lambda\right)^{-1} $ yields 
\begin{multline} \label{estkappa1}\Big|\kappa+\frac{2\pi}{\partial_k^2 E(0)}\int V(x)u_{b_*}(x;0)^2\ dx +\lambda\frac{8\pi^3}{(\partial_k^2 E(0))^2}\iint V(x)u_{b_*}(x;0)^2|x-y|u_{b_*}(y;0)^2V(y) \ dx\ dy \\
-\frac{\lambda}{4}\iint V(x)u_{b_*}(x;0)R_Q^{(0)}(x,y)V(y) u_{b_*}(y;0)\ dx\ dy \Big| \ \leq \ C\lambda^{1+1/4}.\end{multline}
Finally, we note that since $R_Q^{(0)}(x,y)=-R_Q^{(0)}(y,x)$ by Lemma~\ref{lem:conj}, the last term in~\eqref{estkappa1} vanishes. Thus the above estimate, together with the following Lemma, completes the proof of Corollary~\ref{cor:Elambdaprecise-per}.\qed

\begin{Lemma}
Let $E_*=E_{b_*}(0)$ be an edge of the continuous spectrum, and let hypotheses of Theorem~\ref{thm:per_result} be satisfied, so that $E^\lambda$ exists. Then for $\lambda$ small enough, one has $E^\lambda=E(i\lambda\kappa)$, and
$
 E^\lambda-E_* \ = \ -\frac12\lambda^2\kappa^2\partial_k^2 E_{b_*}(0) \ + \ \O(\lambda^4) \ .$

\end{Lemma}
\begin{proof}
We Taylor expand $D(E)$ about $E_*=E_{b_*}(0)$. 
\begin{align}
D(E) & = D(E_*) + D'(E_*)(E-E_*) + \mathcal{O}\left( (E-E_*)^2 \right) = 2 + D'(E_*)(E-E_*) + \mathcal{O}\left( (E-E_*)^2 \right).\label{expand-D}
\end{align} 
Let's first apply~\eqref{expand-D} to $E=E_{b_*}(k)$ in the spectral band. One has
$ D(E_{b_*}(k))=e^{2\pi i k}+e^{-2\pi i k}=2-4\pi^2k^2+\mathcal{O}\left(k^3 \right).$
Finally, since $\partial_k E_{b_*}(0)=\partial_k^3 E_{b_*}(0)=0$, one has 
$E_{b_*}(k) \ = \ E_* \ + \ \frac12\partial_k^2 E(0)k^2+\mathcal{O}\left(k^4 \right).$
Identifying with~\eqref{expand-D}, it follows $ D'(E_*)\big(\frac12\partial_k^2 E(0) \big) \ = \ -4\pi^2$, thus
$ D'(E_*) \ = \ \frac{-8\pi^2}{\partial_k^2 E_{b_*}(0) }.$

Next let's apply~\eqref{expand-D} to $E=E^\lambda$, recalling 
$ D(E^\lambda)=e^{2\pi\lambda\kappa}+e^{-2\pi\lambda\kappa}=2+4\pi^2\lambda^2\kappa^2+
\mathcal{O}\left(\lambda^4\kappa^4 \right).$
One has from ~\eqref{Elambda-Qper} in Theorem~\ref{thm:per_result} that $E^\lambda-E_*=\O( \lambda^2)$, and from~\eqref{estkappa0} that $\kappa=\O(1)$. Consequently,~\eqref{expand-D} yields
\[4\pi^2\lambda^2\kappa^2 \ = \ D'(E_*)(E^\lambda-E_*) \ + \ \O(\lambda^4)\ = \ \frac{-8\pi^2}{\partial_k^2 E_{b_*}(0) }(E^\lambda-E_{b_*}(0)) \ + \ \O(\lambda^4).\]
Finally, we deduce
$ E^\lambda-E_* \ = \ -\frac12\lambda^2\kappa^2\partial_k^2 E_{b_*}(0)\ + \ \O(\lambda^4)$
and the lemma is proved. 
\end{proof}
We conclude this section by the proof of Lemma~\ref{lem:conj}

\noindent{\em Proof of Lemma~\ref{lem:conj}}. 
Let us Taylor-expand $\psi_\pm$, as defined by~\eqref{def-psi}--\eqref{eqn-pb}. One has $\psi_\pm(x)e^{\mp 2\pi\lambda\kappa x}  \equiv  p_{b_*}(x;\mp i\lambda\kappa)$, thus
\begin{align}\label{expansion-psi-plus} 
\psi_+(x)e^{-2\pi\lambda\kappa x} \ &= \  p_{b_*}(x;0) \ - \ i\lambda\kappa \partial_k p_{b_*}(x;0) \ - \ \frac{(\lambda\kappa)^2}2 \partial_\kappa^2 p(x;0) \ +\ i  \frac{(\lambda\kappa)^3}6 \partial_\kappa^3 p(x;i\gamma_+), \\
\label{expansion-psi-minus} 
\psi_-(x)e^{2\pi\lambda\kappa x} \ &= \  p_{b_*}(x;0) \ + \ i\lambda\kappa \partial_k p_{b_*}(x;0) \ - \ \frac{(\lambda\kappa)^2}2 \partial_\kappa^2 p(x;0)\ -\ i  \frac{(\lambda\kappa)^3}6 \partial_\kappa^3 p(x;i\gamma_-) ,
\end{align}
with $-\lambda\kappa\leq \gamma_+\leq 0\leq \gamma_-\leq\lambda\kappa$. 

\begin{Remark} Note that 
$\kappa\mapsto p_b(x;k_*+\kappa)\in L^2(\RR)$ is analytic in a complex neighborhood $|\kappa|<\kappa_1$. By the equation for $p_b$, $\kappa\mapsto p_b(x;k_*+\kappa)\in H^2(\RR)$ is analytic and thus $\partial_k^3p_b(x;k)$ and $\partial_x\partial_k^3p_b(x;k)$ are well-defined and uniformly bounded for $k$ near $k_*$ and $x$ in any compact set.
\end{Remark}
Since $p_{b_*}(x;0)=u_{b_*}(x;0)$, it follows
\begin{multline}\label{expansion-WRQ} 
W[\psi_\pm] R_Q(x, y;E^\lambda)\ =\\
 \left\{ \begin{array}{cc} \left( u_{b_*}(x;0)u_{b_*}(y;0)+i\lambda\kappa r^{(0)}(x,y;\lambda\kappa)+(\lambda\kappa)^2r^{(1)}_+(x,y) \right)e^{2\pi\lambda\kappa(x-y)}& \text{ if } y>x, \\
\left( u_{b_*}(y;0)u_{b_*}(x;0)+i\lambda\kappa r^{(0)}(y,x;\lambda\kappa)+(\lambda\kappa)^2r^{(1)}_-(x,y) \right)e^{2\pi\lambda\kappa(y-x)}& \text{ if } y<x. \end{array} \right.\end{multline}
with 
\[r^{(0)}(x,y;\lambda\kappa) \ \equiv \ p_{b_*}(x;0) \partial_k p_{b_*}(y;0) -\partial_k p_{b_*}(x;0)p_{b_*}(y;0)\ = \ -r^{(0)}(y,x;\lambda\kappa) , \]
and $r^{(1)}_\pm(x,y)$ is bounded, uniformly with respect to $\lambda\kappa$.

Let us now turn to $W[\psi_\pm]\equiv\psi_+'(x)\psi_-(x)-\psi_+(x)\psi_-'(x)$. From~\eqref{expansion-psi-plus}--\eqref{expansion-psi-minus}, one has
\begin{multline*}  W[\psi_\pm] \ = \ 2\lambda\kappa\Big(\ 2\pi p_{b_*}(x;0)^2\ -i p_{b_*}(x;0)\partial_x \partial_k p_{b_*}(x;0)+i \big(\partial_xp_{b_*}(x;0)\big) \big(\partial_k p_{b_*}(x;0)\big)\ \Big)\\ + \ (\lambda\kappa)^3w_r(x;\lambda\kappa),
\end{multline*}
with $w_r(x)$ uniformly bounded, independently of $x$ and $\lambda\kappa$.

Since $W[\psi_\pm]$ is independent of $x$, one has $ W[\psi_\pm] \ = \ \int_0^1W[\psi_\pm]\ dx $ and thus
\begin{multline*} W[\psi_\pm]  =  2\lambda\kappa \int_0^1\Big(\ 2\pi p_{b_*}(x;0)^2 -i p_{b_*}(x;0)\partial_x \partial_k p_{b_*}(x;0)+i \big(\partial_x p_{b_*}(x;0)\big) \big(\partial_k p_{b_*}(x;0)\big)\ \Big)\ dx\\ + \ (\lambda\kappa)^3\int_0^1 w_r(x;\lambda\kappa) \ dx.\end{multline*}
Using that $\int_0^1 p_{b_*}(x;0)^2\ dx \ = \ \int_0^1 u_{b_*}(x;0)^2\ dx \ = \ 1$, one deduces
\begin{equation}\label{expansion-W0}
W[\psi_\pm] \ = \ 2\lambda\kappa\Big(2\pi +2i\int_0^1 p(x;0)\partial_x\partial_k p(x;0)\ dx\Big) \ + \ \O\left((\lambda\kappa)^3\right).
\end{equation}
Now, let us recall that $p_{b_*}(x;i\kappa)$ satisfies~\eqref{eqn-pb}. Deriving twice with respect to $k=i\kappa$, one obtains
\begin{multline*}
\left(-(\partial_x-2\pi\kappa)^2+Q(x)-E(i\kappa)\right)\partial_k^2 p(x;i\kappa) =\ 2\partial_k E(i\kappa) \partial_k p(x;i\kappa)+\partial_k^2 E(i\kappa) p(x;i\kappa)\\ -8\pi i(\partial_x-2\pi\kappa)\partial_k p(x;i\kappa) -8\pi^2 p(x;i\kappa) \ .
\end{multline*}
We now apply this identity at $\kappa=0$, and take the inner product with $p_{b_*}(x;0)$.
 It follows
$
0 =\ \partial_k^2 E(0) -8\pi i\int_0^1 p(x;0)\partial_x\partial_k p(x;0) \ dx -8\pi^2 \ .
$
Therefore,~\eqref{expansion-W0} becomes
\begin{equation}\label{expansion-W}
W[\psi_\pm] \ = \ 2\lambda\kappa\frac{\partial_k^2 E(0)}{4\pi} \ + \ \O\left((\lambda\kappa)^3\right).
\end{equation}
 Finally,~\eqref{expansion-WRQ} and~\eqref{expansion-W} clearly imply~\eqref{expansion-RQ}, and Lemma~\ref{lem:conj} is proved.
\qed


\begin{thebibliography}{10}

\bibitem{Allaire-Piatnitski:05}
G.~Allaire and A.~Piatnitski.
\newblock Homogenization of the {S}chr\"odinger equation and effective mass
  theorems.
\newblock {\em Comm. Math. Phys.}, 258(1):1--22, 2005.

\bibitem{Ashcroft-Mermin:76}
N.~W. {Ashcroft} and N.~D. {Mermin}.
\newblock {\em Solid State Physics}.
\newblock Harcourt, Orlando, FL, 1976.

\bibitem{BLP}
A.~Bensoussan, J.-L. Lions, and G.~Papanicolaou.
\newblock {\em Asymptotic analysis for periodic structures}.
\newblock AMS Chelsea Publishing, Providence, RI, 2011.
\newblock Corrected reprint of the 1978 original.

\bibitem{B:03}
M.~{Birman}.
\newblock On the homogenization procedure for periodic operators in a
  neighborhood of an edge of an internal gap.
\newblock {\em Algebra i Analyz}, 15(4):61--66, 2003.
\newblock Russian. English transl., St. Petersbg Math. J. 15(4):507--513.

\bibitem{BS:03}
M.~{Birman} and T.~{Suslina}.
\newblock Second order periodic differential operators. {T}hreshold properties
  and homogenization.
\newblock {\em Algebra i Analyz}, 15(2):1--108, 2003.
\newblock Russian. English transl., St. Petersbg. Math. J. 15(5):639--714.

\bibitem{BS:99}
M.~S. Birman and T.~A. Suslina.
\newblock Two-dimensional periodic {P}auli operator. {T}he effective masses at
  the lower edge of the spectrum.
\newblock In {\em Mathematical results in quantum mechanics ({P}rague, 1998)},
  volume 108 of {\em Oper. Theory Adv. Appl.}, pages 13--31. Birkh\"auser,
  Basel, 1999.

\bibitem{BS:06}
M.~S. Birman and T.~A. Suslina.
\newblock Homogenization of a multidimensional periodic elliptic operator in a
  neighborhood of an edge of an inner gap.
\newblock {\em Zap. Nauchn. Semin. POMI}, 318:60--74, 2004.
\newblock Russian. English transl., J. Math. Sci., 136(2):3682--3690.

\bibitem{Borisov11}
D.~I. Borisov.
\newblock On the spectrum of a two-dimensional periodic operator with a small
  localized perturbation.
\newblock {\em Izv. Ross. Akad. Nauk Ser. Mat.}, 75(3):29--64, 2011.

\bibitem{Borisov-Gadylshin:08}
D.~I. Borisov and R.~R. Gadyl'shin.
\newblock On the spectrum of a periodic operator with a small localized
  perturbation.
\newblock {\em Izvestiya: Mathematics}, 72(4):659--688, 2008.

\bibitem{BR:11}
J.~Bronski and Z.~Rapti.
\newblock Counting defect modes in periodic eigenvalue problems.
\newblock {\em SIAM J. Math. Anal.}, 43(2):803--827, 2011.

\bibitem{Cai:06}
K.~Cai.
\newblock Dispersion for Schr\"odinger operators with one gap periodic potentials in ${\mathbb R}$,
Dynamics Part. Diff. Eq. 3 (2006), 71.

\bibitem{Courant-Hilbert-I}
R.~Courant and D.~Hilbert.
\newblock {\em Methods of {M}athematical {P}hysics}.
\newblock Interscience Publishers, Inc., New York, N.Y., 1953.

\bibitem{cuccagna:08}
S.~Cuccagna.
\newblock Dispersion for {S}chr\"odinger equation with periodic potential in
  1{D}.
\newblock {\em Comm. P.D.E.}, 33(11), 2008.

\bibitem{Deift-Hempel:86}
P.~Deift and R.~Hempel.
\newblock On the existence of eigenvalues of the schr{\"o}dinger operator
  h-$\lambda$w in a gap of $\sigma$(h).
\newblock {\em Comm. Math. Phys.}, 103(3):461--490, 1986.

\bibitem{Dohnal-Pelinovsky-Schneider:09}
T.~Dohnal, D.~Pelinovsky, and G.~Schneider.
\newblock Coupled-mode equations and gap solitons in a two-dimensional
  nonlinear elliptic problem with a separable periodic potential.
\newblock {\em J. Nonl. Sci.}, 19:95--131, 2009.

\bibitem{Dohnal-Uecker:09}
T.~Dohnal and H.~Uecker.
\newblock Coupled mode equations and gap solitons for the {2D}
  {Gross-Pitaevskii} equation with a non-separable periodic potential.
\newblock {\em Physica D}, 238:860--879, 2009.

\bibitem{DVW3}
V.~Duch\^ene, I.~Vuki{\'c}evi{\'c}, and M.~Weinstein.
\newblock Defect modes induced by high frequency perturbations of the periodic
  1d {S}chr\"odinger operator.
\newblock In progress.

\bibitem{DVW:12}
V.~Duch\^ene, I.~Vuki{\'c}evi{\'c}, and M.~Weinstein.
\newblock Scattering and localization properties of highly oscillatory
  potentials.
\newblock {\em Comm. Pure Appl. Math.}, 1:83--128, 2014.

\bibitem{DW:11}
V. Duch\^ene and M.I. Weinstein.
\newblock Scattering, homogenization and interface effects for oscillatory
	potentials with strong singularities.
\newblock Multiscale Model. and Simul., Vol 9, 1017-1063, 2011.


\bibitem{eastham1973spectral}
M.~Eastham.
\newblock {\em The spectral theory of periodic differential equations}.
\newblock Scottish Academic Press, Edinburgh, 1973.

\bibitem{FW-dirac}
C.~Fefferman and M.~Weinstein.
\newblock Wavepackets in honeycomb structures and two-dimensional {D}irac
  equations.
\newblock {\em Comm. Math. Phys.}, 2013. doi: 10.1007/s00220-013-1847-2.

\bibitem{Figotin-Klein:97}
A.~Figotin and A.~Klein.
\newblock Localized classical waves created by defects.
\newblock {\em J. Statistical Phys.}, 86(1/2):165--177, 1997.

\bibitem{Figotin-Klein:98}
A.~Figotin and A.~Klein.
\newblock Midgap defect modes in dielectric and acoustic media.
\newblock {\em SIAM J. Appl. Math.}, 58(6):1748--1773, 1998.

\bibitem{Firsova:96}
N.~Firsova.
\newblock On the time decay of a wave packet in a one-dimensional finite band periodic lattice
\newblock {\em J. Math. Phys.}, 37 (1996), 1171-1181.

\bibitem{Gesztesy-Simon:93}
F.~Gesztesy and B.~Simon.
\newblock A short proof of {Z}heludev's theorem.
\newblock {\em Trans. Amer. Math. Soc.}, 335:329--340, 1993.

\bibitem{Hoefer-Weinstein:11}
M.~Hoefer and M.~Weinstein.
\newblock Defect modes and homogenization of periodic schr{\"o}dinger
  operators.
\newblock {\em SIAM J. Mathematical Analysis}, 43:971--996, 2011.

\bibitem{Ilan-Weinstein:10}
B.~Ilan and M.~Weinstein.
\newblock Band-edge solitons, nonlinear {S}chr{\"o}dinger/{G}ross-{P}itaevskii
  equations and effective media.
\newblock {\em Multiscale Model. Simul.}, 8(4):1055--1101, 2010.

\bibitem{klaus:77}
M.~Klaus.
\newblock On the bound state of {S}chr\"odinger operators in one dimension.
\newblock {\em Annals of Physics}, 108:288--300, 1977.

\bibitem{Korotyaev:09}
E.~Korotyaev.
\newblock 1d {S}chr\"odinger operator with periodic plus compactly supported
  potentials.
\newblock preprint arXiv:0904.2871v1.

\bibitem{Korotyaev:11}
E.~Korotyaev.
\newblock Resonance theory for perturbed {H}ill operator.
\newblock {\em Asymptotic Analysis}, 74:199--227, 2011.

\bibitem{Kuchment-01}
P.~Kuchment.
\newblock The mathematics of photonic crystals.
\newblock In {\em Mathematical modeling in optical science}, volume~22 of {\em
  Frontiers Appl. Math.}, pages 207--272. SIAM, Philadelphia, PA, 2001.

\bibitem{magnus1979hill}
W.~Magnus and S.~Winkler.
\newblock {\em Hill's Equation}.
\newblock Dover Publications, Inc, New York, 1979.

\bibitem{Nirenberg:74}
L.~Nirenberg.
\newblock {\em Topics in Nonlinear Functional Analysis}, volume~6 of {\em
  Courant Institute Lecture Notes}.
\newblock American Mathematical Society, Providence, Rhode Island, 2001.

\bibitem{parzygnat2010sufficient}
A.~Parzygnat, K.~Lee, Y.~Avniel, and S.~Johnson.
\newblock Sufficient conditions for two-dimensional localization by arbitrarily
  weak defects in periodic potentials with band gaps.
\newblock {\em Physical Review B}, 81(15), 2010.

\bibitem{Pelinovsky-Schneider:07}
D.~E. Pelinovsky and G.~Schneider.
\newblock Justification of the the coupled-mode approximation for a nonlinear
  elliptic problem with a periodic potential.
\newblock {\em Appl. Anal.}, 86(8):1017--1036, 2007.

\bibitem{RS4}
M.~Reed and B.~Simon.
\newblock {\em Methods of modern mathematical physics. {IV}. {A}nalysis of
  operators}.
\newblock Academic Press, New York, 1978.

\bibitem{Rofe-Beketov:64}
F.~S. Rofe-Beketov.
\newblock A test for finiteness of the number of discrete levels introduced
  into the gaps of a continuous spectrum by perturbation of a periodic
  potential.
\newblock {\em Soviet Math. Dokl.}, 5:689--692, 1964.

\bibitem{simon1976bound}
B.~Simon.
\newblock The bound state of weakly coupled {S}chr{\"o}dinger operators in one
  and two dimensions.
\newblock {\em Annals of Physics}, 97(2):279--288, 1976.

\bibitem{Wang-Yang-Chen:07}
J.~Wang, J.~Yang, and Z.~Chen.
\newblock Two-dimensional defect modes in optically induced photonic lattices.
\newblock {\em Phys. Rev. A}, 76:013828, 2007.

\end{thebibliography}
\end{document}